\documentclass[amsmath,amssymb,nofootinbib,superscriptaddress,twocolumn]{revtex4-2}

\usepackage{amsthm}
\usepackage{bm}
\usepackage{graphicx}
\usepackage{hyperref}
\usepackage{orcidlink}

\newcommand{\ket}[1]{\left|#1\right\rangle}
\newcommand{\bra}[1]{\left\langle#1\right|}
\newcommand{\braket}[2]{\left\langle#1\middle|#2\right\rangle}
\newcommand{\ketbra}[2]{\left|#1\right\rangle\hspace{-1ex}\left\langle#2\right|}

\newcommand{\etaNA}{\eta^{(\mathrm{NA})}}
\newcommand{\etaA}{\eta^{(\mathrm{A})}}

\newcommand{\trine}{\mathbb{M}_\mathrm{trine}}
\newcommand{\SIC}{\mathbb{M}_\mathrm{SIC}}

\newtheorem{theorem}{Theorem}
\newtheorem{proposition}[theorem]{Proposition}
\newtheorem{lemma}[theorem]{Lemma}
\newtheorem{corollary}[theorem]{Corollary}

\hypersetup{colorlinks=true,citecolor=blue}

\begin{document}

\title{Reducing measurement error with adaptivity}

\author{James Byrne \!\orcidlink{0009-0007-7210-1643}}
\email{james.m.byrne@bristol.ac.uk}
\affiliation{Quantum Engineering Centre for Doctoral Training, University of Bristol, Tyndall Avenue, Bristol, BS8 1FD, UK}
\affiliation{School of Mathematics, University of Bristol, Fry Building, Woodland Road, Bristol, BS8 1UG, UK}

\author{Noah Linden \!\orcidlink{0009-0001-2342-2084}}
\affiliation{School of Mathematics, University of Bristol, Fry Building, Woodland Road, Bristol, BS8 1UG, UK}

\author{Paul Skrzypczyk \!\orcidlink{0000-0002-9343-9041}\,}
\affiliation{H. H. Wills Physics Laboratory, University of Bristol, Tyndall Avenue, Bristol, BS8 1TL, UK}

\begin{abstract}
    \noindent We show that adaptivity, also called feed-forward, is a powerful resource in reducing the error of measurement circuits that combine multiple uses of a noisy measuring device. In previous work, it has been shown that error in measurements can be mitigated by using the measuring device multiple times. So far, the most effective protocols have been parallel measurement schemes, where all measurements are simultaneous. We extend this idea to adaptive measurement schemes, where the results of previous measurements can influence our choice of processing further down the circuit. We show that adaptive measurement circuits can in general reduce overall measurement errors further than any non-adaptive measurement circuit could with the same number of measurements. In particular, we show that for a large class of noisy two-outcome qubit measurements, such an adaptive advantage can exist when as few as three measurements are used. We also show that the adaptive advantage is unbounded across the class of noisy two-outcome qubit measurements, as the number of uses of the device increases. As part of this work, we devise and use methods for finding optimal measurement circuits in both the non-adaptive and adaptive cases. In addition, we prove general results about the limits of such circuits, both in measuring a qubit, and more generally, a qudit.
\end{abstract}

\maketitle

\section{Introduction}

\noindent While the platonic ideal of a projective measurement is convenient in theory, it is unphysical \cite{guryanova-2020}, and all actual quantum measurements will be noisy in some way. We consider a single-shot scenario, where the effect of this noise is to introduce errors into the output distribution of the measurement, which we could attempt to reduce in two ways. One is an engineering problem, of improving the control over the quantum system and/or the measuring device itself. The other, which we consider in this paper, is an optimisation problem of how best to use the device that we already have. In particular, we allow ourselves to make use of our measuring device multiple times within a single measurement circuit.

\begin{figure*}[t]
    \centering
    \includegraphics[scale=1.25]{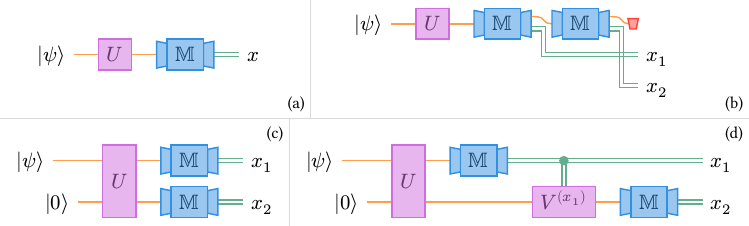}
    \caption{An illustration of multiple measurement schemes with two uses of a measuring device $\mathbb{M}$; note that in general we consider $N$ uses of the measuring device. In all circuit diagrams, single orange lines represent quantum information, while double green lines represent classical information. \textbf{(a)} The most straightforward circuit is to simply use the device once to measure the incoming state, perhaps with the use of a unitary $U$ to better align the state and the measuring device. \textbf{(b)} One approach to using multiple measurements is to measure the same state repeatedly with the given measuring device. For certain measuring devices, this can be more informative than the single measurement \cite{haapasalo-2016,bullock-2021}. \textbf{(c)} A different approach is to split the quantum information across multiple qubits and measure them all simultaneously, similar to how a classical code may be used to create redundancy to protect against errors. This gives more information than a single measurement in general, and approaches the performance of a projective measurement as the number of measurements made is increased \cite{günther-2021,hicks-2022,kim-2026,linden-2025}. \textbf{(d)} Our work considers adaptive circuits, where multiple qubits are used, but measurements are not performed simultaneously, allowing time for processing dependent on previous measurement results before later measurements are performed. This gives a broader class of measurement circuits than in (c), and due to our choice of figure of merit \eqref{eq:total-error-2}, also encompasses sequential circuits of the form in (b). Finding when such circuits improve over those in (c), and by how much, is the object of this paper.}
    \label{fig:introductory-circuits}
\end{figure*}

Figure \ref{fig:introductory-circuits}(a) shows the simplest possible measurement circuit, with the unitary $U$ chosen to make the output of the measurement $\mathbb{M}$ maximally informative. One way to introduce more measurements into this circuit is to repeatedly measure the same system, to gradually extract more and more classical information. This is shown in Figure \ref{fig:introductory-circuits}(b), and has been proven to be effective for certain measurements \cite{haapasalo-2016,bullock-2021}. Another approach is to use ancillary qubits and entangling gates, and perform multiple parallel measurements, as shown in Figure \ref{fig:introductory-circuits}(c). Previous work such as \cite{günther-2021,hicks-2022,kim-2026} have shown that repetition codes and other classical codes can be used to split the quantum information in one qubit across many qubits, and that measuring all qubits gives a reduction in error compared to just measuring the original one to begin with.

Indeed, parallel measurement circuits turn out to be far more generally applicable than sequential measurement circuits, which as yet have only been proven to be effective for carefully chosen measuring devices. In \cite{linden-2025}, it was shown that given any measurement device, any other measurement can be perfectly reproduced in the asymptotic limit of a large number of parallel uses of the device, with average error decreasing exponentially. Moreover, the problem of optimising circuits was investigated, and it was shown that in general approaches mirroring classical codes can be beaten by quantum approaches which can use the entire state space. Along with this new idea of measurement reproduction, several open questions were raised, concerning the theoretical limits of such protocols, as well as potential applications such as speeding up measurements \cite{corlett-2025}.

Our goal is to answer one such question: whether we can reduce errors even further by allowing for \emph{adaptive} procedures, where not all uses of the measuring device are performed in one go, and previous measurement results can inform how we process the quantum system for future measurements. This idea was mentioned previously as a possibility in \cite{linden-2025}, but not explored further there. An example of such a circuit is illustrated in Figure \ref{fig:introductory-circuits}(d). The resource of adaptivity\footnote{There are a number of disparate uses of the term \emph{adaptive} in quantum information processing tasks, but we use the word solely to refer to the procedure of making measurements and then operating on the quantum system conditional on the results of these measurements.} (or feed-forward, as it is also known) has been demonstrated to be powerful for numerous applications including quantum-information-theoretic tasks \cite{bennett-1993, li-2022}, quantum computation \cite{raussendorf-2001, briegel-2009, wiseman-2009, foss-feig-2023, iqbal-2024}, fault tolerance more specifically \cite{shor-1995, tansuwannont-2023, berthusen-2025, pokharel-2025, shirizly-2025}, and metrology \cite{higgins-2007, wiseman-2009, palittapongarnpim-2016}, so we have good reason to believe it may help us here too. As a stepping stone on the way to this, we will also find optimal circuits using parallel measurements, which we will henceforth refer to as \emph{non-adaptive}, to give us a target to beat using adaptivity.

In this paper, we find circuits for measuring a qubit with minimal error in both the adaptive and non-adaptive cases. Moreover, we show that an \emph{adaptive advantage} exists, that is, with a fixed number of measurements, it is possible to reduce the error further by using an adaptive circuit over any possible non-adaptive one. We further show that the ratio of optimal non-adaptive error to optimal adaptive error can be arbitrarily large, therefore the relative advantage of adaptivity is in general unbounded. However, we also show that such adaptive advantage is subtle, and there are scenarios where it doesn't exist. We also extend this to measuring a $d$-dimensional system, and find general upper and lower bounds on the minimal possible error we can hope to achieve.

While there are assumptions made within this work (notably we will assume that there is zero gate noise, and that we are unconstrained by time or space limitations), our results suggest that adaptive measurement circuits are a useful tool to add to the toolbox, and demonstrate the potential for them to be practically viable and useful in future.

\subsection{Outline of this paper}

\noindent Initially, we consider measurement circuits that take a single qubit as input, and return a single classical bit as output; if there are multiple measurements then some classical post-processing will be required to reduce the measurement outputs to a single bit, for instance a majority vote. For any such circuit, we define the error probabilities $\varepsilon_0$, the probability of an input of $\ket{0}$ giving an output of $1$, and $\varepsilon_1$, the probability of an input of $\ket{1}$ giving an output of $0$. Their sum, $\eta = \varepsilon_0 + \varepsilon_1$, is the \emph{total error} of the circuit, and is the figure of merit we want to minimise.

In Section \ref{sec:example}, we consider the physically relevant example of the \emph{imperfect $Z$ measurement} \eqref{eq:imperfect-Z}, a quantum analogue to the classical asymmetric bit-flip channel, parametrised by $p$ and $q$, the probability of measuring $\ket{0}$ and getting result $1$, and the probability of measuring $\ket{1}$ and getting result $0$ respectively. This is a canonical choice as up to unitary equivalence it encompasses every possible two-outcome qubit POVM. Supposing that we have access to three uses of such a measuring device, we show how to reduce the problem of finding the best non-adaptive circuit to one of two cases \eqref{eq:optimal-non-adaptive-3}, depending on the values of the parameters $p$ and $q$. We then repeat this in the adaptive case, again finding that the optimal adaptive circuit takes one of two forms \eqref{eq:optimal-adaptive-3}. One of these is the same as one of the non-adaptive optima, showing that in one region of the $(p,q)$ parameter space there is no advantage to be gained with adaptivity when using three imperfect $Z$ measurements. However, the other does use adaptivity, and we show that in the rest of the parameter space, an adaptive advantage exists, illustrated in Figure \ref{fig:noisy-Z-adaptive-advantage}. We repeat these optimisations with four uses of the measurement instead of three, and find an adaptive advantage for almost the entire parameter space, all except for some edge cases where the optimal measurement circuit can be proven to be non-adaptive regardless of the number of measurements. Finally, we quantify the relative adaptive advantage as the ratio of the optimal total errors achievable non-adaptively and adaptively \eqref{eq:relative-adaptive-advantage}, and show that it can be arbitrarily large for carefully chosen parameters $p$, $q$ and enough uses of the measuring device.

In Section \ref{sec:qubit-measurements}, we again consider the problem of distinguishing $\ket{0}$ from $\ket{1}$, but now with an arbitrary known POVM (in any finite dimension) as a measuring device. We argue how finding the optimal non-adaptive circuit breaks down into optimising the quantum pre-processing that takes place before measurement, and the classical post-processing that takes place after, illustrated in Figure \ref{fig:non-adaptive}. While these optimisations can't be separated, we show that we can meaningfully reduce the search space of possible measurement circuits to a finite size \eqref{eq:maximise-spectral-diameter}, \eqref{eq:abelian-non-adaptive-total-error}, in order to more easily find the optimum. We then consider adaptive circuits, where we show that a recursive structure (Figure \ref{fig:adaptive-recursive}) determines the best achievable adaptive circuits, and consequently that the optimal adaptive circuit is also one of a finite set of possible optima \eqref{eq:finite-achievable-errors}. We also provide a geometric interpretation of this optimisation, as looking for an optimal vertex of a polygon of possible error pairs $(\varepsilon_0, \varepsilon_1)$ that is recursively defined to match the recursive structure of the adaptive circuit. With this knowledge of optimal measurement circuits, we end the section by studying several example measurements. We demonstrate the range of possibilities for adaptive advantage, from measurements like the qubit SIC-POVM \eqref{eq:tetra} where an adaptive advantage exists for any number of uses greater than one, to measurements like the trine \eqref{eq:trine} where we show that there is no adaptive advantage no matter how many uses are allowed.

In Section \ref{sec:qudit-measurements}, we generalise further, allowing the input to be a $d$-dimensional qudit, and the output correspondingly to be some number in $\{0, 1, ..., d-1\}$, for any given $d \geqslant 2$. We similarly generalise our error probabilities to $\varepsilon_{j|i}$, the probability that an input of $\ket{i}$ produces an output of $j$, for $i,j \in \{0, 1, ..., d-1\}$, which we only consider to be an error if $j \neq i$. Then, the total error is the sum of all such probabilities \eqref{eq:total-error-d}, across all values of $i,j$ such that $j \neq i$. The increase in dimensions loses us geometric intuition, but many of the analytical arguments from the previous section carry over into $d$ dimensions, and we find a lower bound on the minimal total error of an adaptive circuit with a fixed number of uses of our measurement in terms of the minimal total error that one measurement could achieve in the qubit case \eqref{cor:adaptive-lower-bound-d}. We then return to our example measurements, and find asymptotic upper bounds on the minimal total error, which we can apply to get general, if quite loose, asymptotic upper bounds for arbitrary measurements, in Propositions \ref{prop:asymptotic-upper-bound} and \ref{prop:generalised-asymptotic-upper-bound}. Indeed, we show that any non-trivial measurement can approximate an arbitrary projective measurement with error decreasing exponentially in the number of measurements made.

\section{Two-outcome qubit measurements} \label{sec:example}

\noindent In this section, we will investigate measurement circuits involving the \emph{imperfect $Z$ measurement} \eqref{eq:imperfect-Z}, which up to unitary equivalence represents any destructive\footnote{Indeed, we will note in Section \ref{sec:qubit-measurements} that the assumption of destructivity is not required, but for simplicity we will keep it here.} measurement on a qubit with two outcomes. We introduce this class of measurements along with the \emph{total error} figure of merit \eqref{eq:total-error-2} in Section \ref{sec:example-setup}. In Section \ref{sec:multiple-measurements}, we will see what happens when we have three uses of this measurement to make, and will derive the non-adaptive optimum \eqref{eq:optimal-non-adaptive-3} in Section \ref{sec:optimal-non-adaptive-noisy-Z}, and the adaptive optimum \eqref{eq:optimal-adaptive-3} in Section \ref{sec:optimal-adaptive-noisy-Z}. We will find that adaptivity may or may not be advantageous, depending on the parameters $p$ and $q$ of our measuring device, as illustrated in Figure \ref{fig:noisy-Z-adaptive-advantage}. We also consider four uses, in which case there is almost always an adaptive advantage, except for some edge cases. Finally, in Section \ref{sec:relative-adaptive-advantage}, we will quantify the relative adaptive advantage by the ratio of the minimal non-adaptive total error to the minimal adaptive total error \eqref{eq:relative-adaptive-advantage}, which we will show can be arbitrarily large.

\begin{figure*}[t]
    \centering
    \includegraphics[scale=1.25]{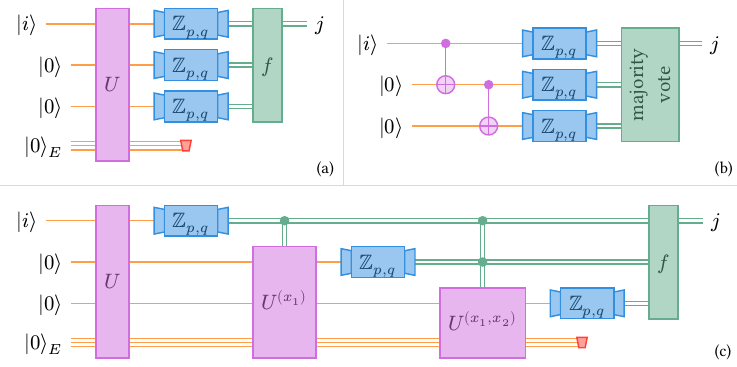}
    \caption{\textbf{(a)} The most general non-adaptive circuit when three uses of the imperfect $Z$ measurement $\mathbb{Z}_{p,q}$ are available. A unitary $U$ creates a state across the three qubits that will be measured and a number of ancillary qubits. After the measurements, the results are interpreted by a classical post-processing function $f$, from three bits to one. \textbf{(b)} An example non-adaptive circuit, which achieves $\varepsilon_0 = 3p^2 - 2p^3$, $\varepsilon_1 = 3q^2 - 2q^3$. If $p,q < \frac{1}{2}$, this is an improvement on the errors achievable when only one measurement was made. \textbf{(c)} The most general adaptive circuit when three uses of the imperfect $Z$ measurement $\mathbb{Z}_{p,q}$ are available. One qubit is measured at a time, and the intermediate unitaries $U^{(x_1)}$ and $U^{(x_1,x_2)}$ are allowed to depend on $x_1$, $x_2$, the classical results of previous measurements.}
    \label{fig:noisy-Z-three-measurements}
\end{figure*}

\subsection{Setup} \label{sec:example-setup}

\noindent We are given an input of a single qubit in state $\ket{i} \in \{\ket{0}, \ket{1}\}$, and we would like to measure the state to work out which one it is. With access to a projective measurement this is trivial, since the two possibilities $\ket{0}$ and $\ket{1}$ are orthogonal, so can be distinguished with certainty.

However, what if the measurement we have access to is noisy, in the sense that it is non-projective? Even if we have exactly characterised the behaviour of the measurement, that is, we know the POVM elements that define it, we remain at the mercy of probability.

We suppose that the only measurement that we can make is the \emph{imperfect $Z$ measurement} $\mathbb{Z}_{p,q}$, with known parameters $p,q \in [0,1]$. This we define as the POVM:
\begin{equation}
    \mathbb{Z}_{p,q} = \{ Z_0, Z_1 \}, \label{eq:imperfect-Z}
\end{equation}

\noindent where $Z_0$ and $Z_1$ are parametrised in terms of $p$ and $q$ as:
\begin{align}
    Z_0 = Z_0(p,q) &= (1-p) \ketbra{0}{0} + q \ketbra{1}{1}; \nonumber \\
    Z_1 = Z_1(p,q) &= p \ketbra{0}{0} + (1-q) \ketbra{1}{1}. 
\end{align}

\noindent In words, $p$ is the probability that we falsely identify $\ket{0}$ as $\ket{1}$, and $q$ is the probability that we falsely identify $\ket{1}$ as $\ket{0}$. The imperfect $Z$s are a simply defined yet very broad class of measurements --- indeed, any two-outcome qubit POVM is unitarily equivalent to some imperfect $Z$ by rotating such that the elements are diagonal in the computational basis. Hence, they will serve as both an intuitive and physically relevant example to investigate.

In general, we think of the measurement circuit as a whole as a box which takes an input of either $\ket{0}$ or $\ket{1}$, and outputs a single bit, either $0$ or $1$, as a result. For any such procedure, we define two error rates $\varepsilon_0$ and $\varepsilon_1$ by:
\begin{align}
    \varepsilon_0 &= \mathbb{P}(\text{output } 1 \mid\text{input} \ket{0}); \nonumber \\
    \varepsilon_1 &= \mathbb{P}(\text{output } 0 \mid \text{input} \ket{1}). \label{eq:errors-2}
\end{align}

\noindent While there are many possible figures of merit we may wish to optimise, our goal will be to minimise the \emph{total error} $\eta$, which we define as the sum of these errors:
\begin{equation}
    \eta = \varepsilon_0 + \varepsilon_1. \label{eq:total-error-2}
\end{equation}

\noindent For an operational interpretation, this is twice the expected probability that an error is made, assuming that the input states $\ket{i} = \ket{0}$ and $\ket{i} = \ket{1}$ are equally likely.

For instance, suppose that we fed the input state $\ket{i}$ into the imperfect $Z$ measurement $\mathbb{Z}_{p,q}$ immediately, and simply outputted as our response the measurement result. This yields:
\begin{equation}
    \varepsilon_0 = p; \quad \varepsilon_1 = q; \quad \eta = p + q. \label{eq:single-measurement-errors}
\end{equation}

\noindent Depending on the values of $p$ and $q$, and the specific application, these errors might be negligibly small, or they could have a significant impact. Either way, we know from previous work \cite{haapasalo-2016,bullock-2021,günther-2021,hicks-2022,linden-2025} that they can be improved, by using $\mathbb{Z}_{p,q}$ multiple times.

\subsection{Multiple measurements} \label{sec:multiple-measurements}

\noindent We now briefly review what has already been shown about improving noisy measurements by using them multiple times. Suppose we have the means of using $\mathbb{Z}_{p,q}$ three\footnote{As we will show in Section \ref{sec:noisy-Z-further-example}, two uses of imperfect $Z$ turns out to be a relatively uninteresting case, as there is never an adaptive advantage to be found, hence we jump straight to three uses in this section.} times instead of just one. We also assume that we can perform arbitrary unitary operations as we please, with no limitations on time or space. One option is to consider repeatedly measuring the same state with $\mathbb{Z}_{p,q}$, as in \cite{haapasalo-2016,bullock-2021}. However, our measurement as defined is destructive, so this previous work would need adapting to be able to apply it here. Instead, as a baseline we consider using $\mathbb{Z}_{p,q}$ three times to measure three qubits in parallel, as in \cite{günther-2021,hicks-2022,linden-2025}, and as depicted in Figure \ref{fig:noisy-Z-three-measurements}(a). For reasons that will soon become clear, we refer to this arrangement as a \emph{non-adaptive} circuit.

First, from our input $\ket{i}$, we create a large (in general entangled) state across the three qubits we will measure, and any number of ancillary qubits we like. Then, after the measurement, we will be left with three bits, say $x_1, x_2, x_3 \in \{0,1\}$, which are the measurement outcomes. From each possible combination of bits, we need to choose whether to guess that the input state was $\ket{0}$ or $\ket{1}$, which we encapsulate by the function $f$.

For instance, a natural choice is the circuit shown in Figure \ref{fig:noisy-Z-three-measurements}(b), emulating the classical length 3 repetition code. First, our input state $\ket{i}$ is copied across to the other two qubits using CNOTs, so the pre-measurement state is $\ket{000}$ if $\ket{i} = \ket{0}$, and $\ket{111}$ if $\ket{i} = \ket{1}$. Then, after measuring, we simply take the majority vote from the measurement results to be our guess $j$. If we calculate the probabilities of each string of outcomes $\mathbf{x}$ given input $\ket{i}$ from this circuit, we can sum them appropriately to find the error rates, and hence total error:
\begin{equation}
    \varepsilon_0 = 3p^2 - 2p^3; \quad \varepsilon_1 = 3q^2 - 2q^3; \nonumber
\end{equation}
\begin{equation}
    \eta = 3p^2 + 3q^2 - 2p^3 - 2q^3,
\end{equation}

\noindent which is a lesser total error than \eqref{eq:single-measurement-errors} if for instance $0 < p,q < \frac{1}{2}$. We may interpret this as the overall measurement effectively becoming $\mathbb{Z}_{\varepsilon_0, \varepsilon_1}$, and note that the error rates $p \mapsto \varepsilon_0 \approx 3p^2$ and $q \mapsto \varepsilon_1 \approx 3q^2$ have decreased.

We will work out whether or not our example circuit in Figure \ref{fig:noisy-Z-three-measurements}(b) is in any way optimal in a moment, but for now let's also consider the wider class of \emph{adaptive} circuits, such as shown in Figure \ref{fig:noisy-Z-three-measurements}(c). Now, instead of assuming that the measurements happen in parallel (i.e. non-adaptively), we can perform them one at a time. Between each pair of measurements, we allow for a unitary operation to take place on the remaining qubits\footnote{We assume for now that the measurement is destructive, so the measured qubit is no longer accessible. As we will note in Section \ref{sec:qubit-measurements}, this assumption can be removed due to the figure of merit we have chosen.}, the action of which may depend on previous measurement results.

Certainly, the class of adaptive circuits is an extension of the class of non-adaptive circuits, but it is an interesting question whether the introduction of adaptivity will be helpful in reducing the minimal total error we can achieve. We may either answer in the affirmative by exhibiting an adaptive circuit with smaller total error than the best non-adaptive circuit, or in the negative by demonstrating that the adaptive circuit that minimises total error doesn't beat the best non-adaptive circuit. Either way, it will be useful to know what the best non-adaptive circuit and the best adaptive measurement circuit are for any given $p,q$, which are the problems we'll tackle in the next two subsections.

\subsection{The optimal non-adaptive circuit} \label{sec:optimal-non-adaptive-noisy-Z}

\noindent First, we want to find the non-adaptive measurement circuit using $\mathbb{Z}_{p,q}$ three times with total error minimised. To find this optimal circuit for general $p,q$, we will first reduce the search space of possible circuits as much as possible.

Optimising over the entire space of unitaries $U$ is unnecessary, as we only care about the effect of $U$ on two states, corresponding to our two possible inputs $\ket{i} = \ket{0}$ or $\ket{i} = \ket{1}$. Instead, let $\rho_i$ be the state of the first three qubits before measurement in Figure \ref{fig:noisy-Z-three-measurements}(a), given an input of $\ket{i}$:
\begin{equation}
    \rho_i = \mathrm{tr}_E \left( U (\ketbra{i}{i} \otimes \ketbra{0}{0} \otimes \ketbra{0}{0} \otimes {\ketbra{\mathbf{0}}{\mathbf{0}}}_E) U^\dagger \right).
\end{equation}

\noindent In addition, let $Z_\mathbf{x} = Z_{x_1} \otimes Z_{x_2} \otimes Z_{x_3}$ for any $\mathbf{x} = (x_1, x_2, x_3) \in \{0, 1\}^3$. Then, the error rates $\varepsilon_0$ and $\varepsilon_1$ are given by:
\begin{equation}
    \varepsilon_0 = \sum_{\mathbf{x} : f(\mathbf{x}) = 1} \mathrm{tr}(\rho_0 Z_\mathbf{x}); \quad \varepsilon_1 = \sum_{\mathbf{x} : f(\mathbf{x}) = 0} \mathrm{tr}(\rho_1 Z_\mathbf{x}), \label{eq:non-adaptive-error-rates}
\end{equation}

\noindent and so the total error is:
\begin{align}
    \eta &= \sum_{\mathbf{x} : f(\mathbf{x}) = 1} \mathrm{tr}(\rho_0 Z_\mathbf{x}) + \sum_{\mathbf{x} : f(\mathbf{x}) = 0} \mathrm{tr}(\rho_1 Z_\mathbf{x}) \nonumber \\ 
    &= \sum_{\mathbf{x} \in \{0,1\}^3} \delta_{f(\mathbf{x}), 1} \, \mathrm{tr}(\rho_0 Z_\mathbf{x}) + \delta_{f(\mathbf{x}), 0} \, \mathrm{tr}(\rho_1 Z_\mathbf{x}).
\end{align}

\noindent This gives us the total error in terms of the arbitrary states $\rho_0$ and $\rho_1$ instead of $U$, so we can optimise them instead as a proxy for optimising $U$.

Supposing that $\rho_0$ and $\rho_1$ are fixed, we see that the contribution to the total error from any given string of outcomes $\mathbf{x} \in \{ 0, 1 \}^3$ is minimised if $f(\mathbf{x}) = i$, where $i$ maximises $\mathrm{tr}(\rho_i Z_\mathbf{x})$. This is called \emph{maximum likelihood decoding}, and defines the optimal post-processing function $f$, with the total error becoming:
\begin{equation}
    \eta = \sum_{\mathbf{x} \in \{0,1\}^3} \min \{ \mathrm{tr}(\rho_0 Z_\mathbf{x}), \mathrm{tr}(\rho_1 Z_\mathbf{x}) \}. \label{eq:total-error-sum-of-minima}
\end{equation}

\noindent Now let's look for the optimal $\rho_0$ and $\rho_1$. Since $Z_\mathbf{x}$ is diagonal in the computational basis, any off-diagonal elements of $\rho_0$ and $\rho_1$ don't contribute to the probabilities, so we may assume without loss of generality that the states are completely dephased before measurement:
\begin{equation}
    \rho_i = \sum_{\mathbf{y} \in \{0, 1\}^3} \lambda_{i}^{(\mathbf{y})} \ketbra{\mathbf{y}}{\mathbf{y}},
\end{equation}

\noindent where $\bm{\lambda}_i = \left( \lambda_i^{(\mathbf{y})} \right)_{\mathbf{y} \in \{0, 1\}^3}$ is a probability vector for both $i = 0$ and $i = 1$. If we let $\mathbf{z}_\mathbf{x} = \mathrm{diag}(Z_\mathbf{x})$, then the total error is:
\begin{equation}
    \eta = \eta(\bm{\lambda}_0, \bm{\lambda}_1) = \sum_{\mathbf{x} \in \{0,1\}^3} \min \{ \bm{\lambda}_0 \cdot \mathbf{z}_\mathbf{x}, \bm{\lambda}_1 \cdot \mathbf{z}_\mathbf{x} \}.
\end{equation}

\noindent We write the total error as an explicit function of $\bm{\lambda}_0$ and $\bm{\lambda}_1$, in order to show that it is concave in both arguments. Indeed, substituting $\bm{\lambda}_0 = (1-t) \bm{\mu}_0 + t \bm{\nu}_0$ for $t \in [0,1]$, we find that for each term of the sum:
\begin{align}
    &\hspace{-1em} \min \{ \left( (1-t) \bm{\mu}_0 + t \bm{\nu}_0  \right) \cdot \mathbf{z}_\mathbf{x}, \bm{\lambda}_1 \cdot \mathbf{z}_\mathbf{x} \} \nonumber \\
    &= \min \{ (1-t) \bm{\mu}_0 \cdot \mathbf{z}_\mathbf{x} + t \bm{\nu}_0 \cdot \mathbf{z}_\mathbf{x}, \nonumber \\
    &\hspace{5em} (1-t) \bm{\lambda}_1 \cdot \mathbf{z}_\mathbf{x} + t \bm{\lambda}_1 \cdot \mathbf{z}_\mathbf{x} \} \nonumber \\
    &\geqslant \min \{ (1-t) \bm{\mu}_0 \cdot \mathbf{z}_\mathbf{x}, (1-t) \bm{\lambda}_1 \cdot \mathbf{z}_\mathbf{x} \} \nonumber \\
    &\hspace{5em} + \min \{ t \bm{\nu}_0 \cdot \mathbf{z}_\mathbf{x}, t \bm{\lambda}_1 \cdot \mathbf{z}_\mathbf{x} \} \nonumber \\
    &= (1-t) \min \{ \bm{\mu}_0 \cdot \mathbf{z}_\mathbf{x}, \bm{\lambda}_1 \cdot \mathbf{z}_\mathbf{x} \} \nonumber \\
    &\hspace{5em} + t \min \{ \bm{\nu}_0 \cdot \mathbf{z}_\mathbf{x}, \bm{\lambda}_1 \cdot \mathbf{z}_\mathbf{x} \},
\end{align}

\noindent using $\min\{a + b, c + d\} \geqslant \min\{a,c\} + \min\{b,d\}$. As the total error is the sum of such terms, it is similarly concave in $\bm{\lambda_0}$:
\begin{equation}
    \eta((1-t) \bm{\mu}_0 + t \bm{\nu}_0, \bm{\lambda}_1) \geqslant (1-t) \eta(\bm{\mu}_0, \bm{\lambda}_1) + t \eta(\bm{\nu}_0, \bm{\lambda}_1),
\end{equation}

\noindent and symmetrically it is concave in $\bm{\lambda_1}$ too.

A concave function is minimised when the arguments are at extremal points of the input domain, and the extrema of the probability simplex are the unit vectors, so the optimal choice of $\rho_0$ and $\rho_1$ is a pair of pure product states of the form:
\begin{equation}
    \rho_0 = \ketbra{\mathbf{y}_0}{\mathbf{y}_0}; \quad \rho_1 = \ketbra{\mathbf{y}_1}{\mathbf{y}_1},
\end{equation}

\noindent for some $\mathbf{y}_0, \mathbf{y}_1 \in \{0, 1\}^3$.

Suppose that the first entries of $\mathbf{y}_0$ and $\mathbf{y}_1$ are the same, i.e. the first qubit is the same in both states $\ket{\mathbf{y}_0}$ and $\ket{\mathbf{y}_1}$. Then, regardless of the input, the first measurement always measures the same state, and so provides no information as to whether the input was $\ket{0}$ or $\ket{1}$. We can therefore assume that the first entries are different, since it is always possible to replicate the effect of having no information from the first measurement by discarding the measurement result and replacing it with classical randomness.

Repeating this argument for all three entries, we deduce that optimally:
\begin{equation}
    \ket{\mathbf{y}_1} = \ket{\mathbf{y}_0 \oplus \mathbf{1}}.
\end{equation}

\noindent Moreover, as the qubits don't have a natural order, we can permute the qubits of $\ket{\mathbf{y}_0}$ and achieve the same error rates provided that the inverse permutation is included in $f$. Also, swapping all zeroes with ones has the effect of swapping $\ket{\mathbf{y}_0}$ and $\ket{\mathbf{y}_1}$, and hence when combined with a bit flip after $f$ will swap $\varepsilon_0$ and $\varepsilon_1$. This doesn't affect the minimal achievable total error, so we can assume that $\ket{\mathbf{y}_0}$ has at most half its qubits in state $\ket{1}$. Thus we are left with just two options for $\ket{\mathbf{y}_0}$:
\begin{equation}
    \ket{\mathbf{y}_0} \in \{ \ket{000}, \ket{001} \}.
\end{equation}

\noindent Now we have all the information to find the optimal circuit across the entire range of values of $p,q$. We can assume that $p + q \leqslant 1$ without loss of generality, since $\mathbb{Z}_{p,q}$ and $\mathbb{Z}_{1-p,1-q}$ are equivalent measurements up to a bit flip on the output. Similarly, we will also assume without loss of generality that $p \leqslant q$, because combining $\mathbb{Z}_{p,q}$ with an $X$ gate before and a bit flip after turns it into $\mathbb{Z}_{q,p}$, and vice versa. Thus, since we know the values of $p$ and $q$, we can transform the problem into an equivalent one with $0 \leqslant p \leqslant q \leqslant 1 - p$, which we will now assume holds from this point forward.

\begin{table*}[t]
    \centering
    \begin{tabular}{c||c|c|c|c}
        $\mathbf{x}$ & $\bra{000} Z_\mathbf{x} \ket{000}$ & $\bra{111} Z_\mathbf{x} \ket{111}$ & $\bra{000} Z_\mathbf{x} \ket{000} - \bra{111} Z_\mathbf{x} \ket{111}$ & $f(\mathbf{x})$ \\ \hline \hline
        $000$ & $(1-p)^3$ & $q^3$ & $(1-p-q)((1-p)^2 + (1-p)q + q^2)$ & 0 \\ \hline
        $001$ or $010$ or $100$ & $p(1-p)^2$ & $q^2 (1-q)$ & $(1-p-q)(p - p^2 + pq - q^2)$ & $\begin{cases} 0 & \text{if } p - p^2 + pq - q^2 \geqslant 0 \\ 1 & \text{if } p - p^2 + pq - q^2 \leqslant 0 \end{cases}$ \\ \hline
        $011$ or $101$ or $110$ & $p^2 (1-p)$ & $q(1-q)^2$ & $-(1-p-q)(q - p^2 + pq - q^2)$ & 1 \\ \hline
        $111$ & $p^3$ & $(1-q)^3$ & $-(1-p-q)(p^2 + p(1-q) + (1-q)^2)$ & 1
    \end{tabular}
    \caption{The optimal post-processing function $f$ when $\rho_0 = \ketbra{000}{000}$, $\rho_1 = \ketbra{111}{111}$ for the range of $0 \leqslant p \leqslant q \leqslant 1 - p$. In the third row, we use that in this range, $q - p^2 + pq - q^2 = q(1-q) + p(q-p) \geqslant 0$.}
    \label{tab:optimal-f-000}
\end{table*}

\begin{table*}[t]
    \centering
    \begin{tabular}{c||c|c|c|c}
        $\mathbf{x}$ & $\bra{001} Z_\mathbf{x} \ket{001}$ & $\bra{110} Z_\mathbf{x} \ket{110}$ & $\bra{001} Z_\mathbf{x} \ket{001} - \bra{110} Z_\mathbf{x} \ket{110}$ & $f(\mathbf{x})$ \\ \hline \hline
        $000$ & $(1-p)^2 q$ & $(1-p)q^2$ & $(1-p-q)(1-p)q$ & 0 \\ \hline
        $001$ & $(1-p)^2 (1-q)$ & $p q^2$ & $(1-p-q)(1-p+pq)$ & 0 \\ \hline
        $010$ or $100$ & $(1-p)pq$ & $(1-p)q(1-q)$ & $-(1-p-q)(1-p)q$ & 1 \\ \hline
        $011$ or $101$ & $p(1-p)(1-q)$ & $pq(1-q)$ & $(1-p-q)p(1-q)$ & 0 \\ \hline
        $110$ & $p^2 q$ & $(1-p)(1-q)^2$ & $-(1-p-q)(1-q+pq)$ & 1 \\ \hline
        $111$ & $p^2 (1-q)$ & $p (1-q)^2$ & $-(1-p-q)p(1-q)$ & 1
    \end{tabular}
    \caption{The optimal post-processing function $f$ when $\rho_0 = \ketbra{001}{001}$, $\rho_1 = \ketbra{110}{110}$ for the range of $0 \leqslant p \leqslant q \leqslant 1 - p$.}
    \label{tab:optimal-f-001}
\end{table*}

Then, for each $\mathbf{x} \in \{0,1\}^3$, we can calculate the probabilities of getting the result $\mathbf{x}$ for both input states, and use maximum likelihood decoding to find the optimal value for $f(\mathbf{x})$. This we do, for $\rho_0 = \ketbra{000}{000}$, $\rho_1 = \ketbra{111}{111}$ in Table \ref{tab:optimal-f-000}, and for $\rho_0 = \ketbra{001}{001}$, $\rho_1 = \ketbra{110}{110}$ in Table \ref{tab:optimal-f-001}.

\begin{figure}[t]
    \centering
    \includegraphics[scale=1.25]{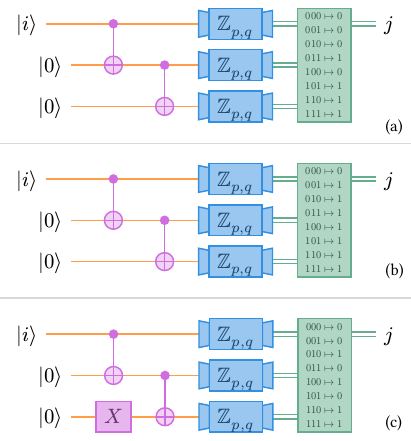}
    \caption{The three candidate non-adaptive circuits to be optimal with three imperfect $Z$ measurements $\mathbb{Z}_{p,q}$ when $0 \leqslant p \leqslant q \leqslant 1-p$. \textbf{(a)} is the same circuit as the example in Figure \ref{fig:noisy-Z-three-measurements}(b). In the main text, we compare the total errors each attains, and show that the optimal circuit is never (a), is \textbf{(b)} when $3p - p^2 - pq - q^2 \leqslant 0$, and is \textbf{(c)} when $3p - p^2 - pq - q^2 \geqslant 0$.}
    \label{fig:non-adaptive-noisy-Z}
\end{figure}

In the first instance, $f$ takes one of two forms, depending on $p$ and $q$, while in the second, $f$ is always the same for any $p,q$ under consideration. Thus, we get three candidate circuits in total, one of which we know will be optimal for any $p,q$ satisfying $0 \leqslant p \leqslant q \leqslant 1-p$, displayed in Figure \ref{fig:non-adaptive-noisy-Z}. The first, (a), comes from $\rho_0 = \ketbra{000}{000}$, $\rho_1 = \ketbra{111}{111}$, when $p - p^2 + pq - q^2 \geqslant 0$, which we note is exactly the same as the one shown in Figure \ref{fig:noisy-Z-three-measurements}(b). If the pre-processing is the same, but instead $p - p^2 + pq - q^2 \leqslant 0$, then maximum likelihood decoding instead gives us the circuit (b), with the strings $001$, $010$, and $100$ being decoded as $1$ instead of $0$. Finally, if $\rho_0 = \ketbra{001}{001}$, $\rho_1 = \ketbra{110}{110}$, then the optimal choice of post-processing is always the same, giving the circuit shown in (c).

Using \eqref{eq:non-adaptive-error-rates}, we can calculate the error rates and hence total error of each of these circuits, and compare across different choices of unitary pre-processing. Labelling the error rates $\varepsilon_0$, $\varepsilon_1$, and $\eta$ with additional subscripts to denote the circuit they come from, we have:
\begin{align}
    \varepsilon_{0,3 (\mathrm{a})} = 3p^2 - 2p^3; &\quad \varepsilon_{1,3 (\mathrm{a})} = 3q^2 - 2q^3; \nonumber \\
    \eta_{3 (\mathrm{a})} = 3p^2 + 3 & q^2 - 2p^3 - 2q^3; \nonumber \\
    \varepsilon_{0,3 (\mathrm{b})} = 3p - 3p^2 & + p^3; \quad \varepsilon_{1,3 (\mathrm{b})} = q^3; \nonumber \\
    \eta_{3 (\mathrm{b})} = 3p - 3 & p^2 + p^3 + q^3; \nonumber \\
    \varepsilon_{0,3 (\mathrm{c})} = p^2 + 2pq - 2p^2 q; & \quad \varepsilon_{1,3 (\mathrm{c})} = 2pq + q^2 - 2pq^2; \nonumber \\
    \eta_{3 (\mathrm{c})} = p^2 + 4pq \, + \, & q^2 - 2p^2 q - 2p q^2. \label{eq:non-adaptive-noisy-Z-errors}
\end{align}

\noindent We already know from the derivation in Table \ref{tab:optimal-f-000} that $\eta_{3 (\mathrm{a})}$ will be smaller than $\eta_{3 (\mathrm{b})}$ if $p - p^2 + pq - q^2 \geqslant 0$, and greater if $p - p^2 + pq - q^2 \leqslant 0$. Comparing the other two pairs:
\begin{align}
    \eta_{3 (\mathrm{a})} - \eta_{3 (\mathrm{c})} &= 2p^2 - 4pq + 2q^2 - 2p^3 + 2p^2 q + 2p q^2 - 2q^3 \nonumber \\
    &= 2(1-p-q)(p-q)^2;
\end{align}
\begin{align}
    \eta_{3 (\mathrm{b})} - \eta_{3 (\mathrm{c})} &= 3p - 4p^2 - 4pq - q^2 + p^3 + 2p^2 q + 2pq^2 + q^3 \nonumber \\
    &= (1-p-q)(3p - p^2 - pq - q^2).
\end{align}

\noindent Since $\eta_{3 (\mathrm{a})} - \eta_{3 (\mathrm{c})} \geqslant 0$ always, then the circuit in Figure \ref{fig:non-adaptive-noisy-Z}(a) can never beat the one in Figure \ref{fig:non-adaptive-noisy-Z}(c), and so we can discount it entirely.

This leaves the two circuits in Figure \ref{fig:non-adaptive-noisy-Z}(b) and Figure \ref{fig:non-adaptive-noisy-Z}(c), of which (b) is optimal when $3p - p^2 - pq - q^2 \leqslant 0$, and (c) is optimal when $3p - p^2 - pq - q^2 \geqslant 0$. We'll call the minimal total error achievable with a non-adaptive circuit $\etaNA$, which will be useful when comparing to adaptive circuits soon:\begin{equation}
    \etaNA = \begin{cases} 3p - 3p^2 + p^3 + q^3 \\ \qquad \qquad \text{if } 3p - p^2 - pq - q^2 \leqslant 0; \\ p^2 + 4pq + q^2 - 2p^2 q - 2p q^2 \\ \qquad \qquad \text{if } 3p - p^2 - pq - q^2 \geqslant 0. \end{cases} \label{eq:optimal-non-adaptive-3}
\end{equation}

\noindent This gives us the minimal total error achievable by a non-adaptive circuit for any values of $p,q$ satisfying $0 \leqslant p \leqslant q \leqslant 1-p$ (and indeed any values of $p,q$ at all by extrapolating to the other three quadrants of the parameter space). Note that where the inequalities overlap, both expressions coincide, so either circuit in Figure \ref{fig:non-adaptive-noisy-Z}(b) or Figure \ref{fig:non-adaptive-noisy-Z}(c) is optimal. We continue to use overlapping inequalities to emphasise this choice of different optima throughout.

\subsection{An adaptive advantage} \label{sec:optimal-adaptive-noisy-Z}

\begin{figure*}[t]
    \centering
    \includegraphics[scale=1.25]{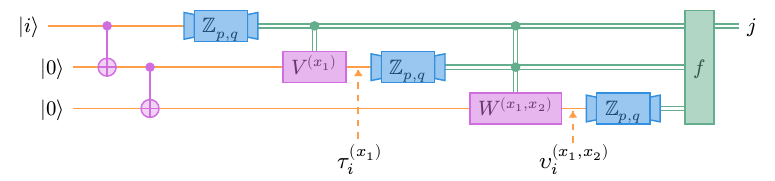}
    \caption{The general form of a class of circuits where one such circuit will be the optimal adaptive circuit with three imperfect $Z$ measurements $\mathbb{Z}_{p,q}$ for any choice of $p,q$. Given an input $\ket{i}$, $\tau_i^{(x_1)} = V^{(x_1)} \ketbra{i}{i} {V^{(x_1)}}^\dagger$ is the state of the second qubit immediately before it is measured, conditional on the first measurement result being $x_1$, while $\upsilon_i^{(x_1,x_2)} = W^{(x_1,x_2)} \ketbra{i}{i} {W^{(x_1,x_2)}}^\dagger$ is the state of the third qubit immediately before it is measured, conditional on the first two measurement results being $x_1$ and $x_2$. $V^{(x_1)}, W^{(x_1,x_2)} \in \{I, X\}$ for all $x_1,x_2$, and $f$ is determined by maximum likelihood decoding.}
    \label{fig:adaptive-noisy-Z-simplification}
\end{figure*}

\begin{figure*}[t]
    \centering
    \includegraphics[scale=1.25]{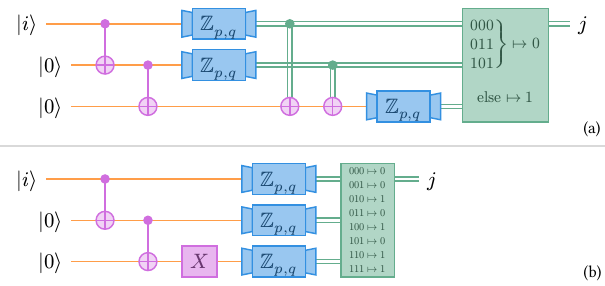}
    \caption{\textbf{(a)} An optimal adaptive circuit with three uses of $\mathbb{Z}_{p,q}$ when $p - p^2 + pq - q^2 \leqslant 0$, where $f$ is determined by maximum likelihood decoding. \textbf{(b)} An optimal adaptive circuit with three uses of $\mathbb{Z}_{p,q}$ when $p - p^2 + pq - q^2 \geqslant 0$, where $f$ is again determined by maximum likelihood decoding. Indeed, this circuit is equivalent to the one in Figure \ref{fig:non-adaptive-noisy-Z}(c), and doesn't utilise adaptivity even though we have allowed for that possibility, so shows that there is no adaptive advantage with three uses of $\mathbb{Z}_{p,q}$ when $p - p^2 + pq - q^2 \geqslant 0$.}
    \label{fig:adaptive-noisy-Z}
\end{figure*}

\noindent Now that we know what the optimal non-adaptive circuit using $\mathbb{Z}_{p,q}$ three times is, we have a target to try and beat with an adaptive circuit. To get a complete picture of where adaptive advantage is possible, we will need to find a way of determining the optimal adaptive circuit, as we have done in the non-adaptive case. Much of the details of this calculation will be similar in spirit to what we have already done for non-adaptive circuits, so we leave the full explanation to Appendix \ref{sec:optimal-adaptive-noisy-Z-argument}. However, we will briefly run through the methodology here.

First, we want to simplify the circuit, from the general form in Figure \ref{fig:noisy-Z-three-measurements}(c) to something more manageable while retaining optimality. We define $\sigma_i$, $\tau_i^{(x_1)}$, and $\upsilon_i^{(x_1,x_2)}$ to be the states of the first, second, and third qubits respectively, right before they are measured, given input $i$, first measurement result $x_1$, and second measurement result $x_2$. Then, much as we showed that $\rho_0$ and $\rho_1$ needed to be pure product states in the computational basis before, we find that without loss of generality we should choose our measured states extremally as:
\begin{align}
    \sigma_0 & =  \ketbra{0}{0}; \quad \sigma_1 = \ketbra{1}{1}; \nonumber \\
    \tau_i^{(x_1)}, &  \, \upsilon_i^{(x_1,x_2)} \in \{ \ketbra{0}{0}, \ketbra{1}{1} \}; \nonumber \\
    \tau_0^{(x_1)} \neq & \, \tau_1^{(x_1)}; \quad \upsilon_0^{(x_1,x_2)} \neq \upsilon_1^{(x_1,x_2)}.
\end{align}

\noindent Any such choice of states can be realised by a circuit of the form shown in Figure \ref{fig:adaptive-noisy-Z-simplification}, with single-qubit unitaries $V^{(x_1)}, W^{(x_1,x_2)} \in \{I, X\}$. Then, as described in Appendix \ref{sec:optimal-adaptive-noisy-Z-argument}, we now go through the options for these unitaries, working backwards through the circuit to minimise total error at each step, and come to the conclusion that one optimal choice for the unitaries $V^{(x_1)}$ and $W^{(x_1,x_2)}$ is:
\begin{align}
    V^{(0)} &= I; \nonumber \\
    V^{(1)} &= I; \nonumber \\
    W^{(0,0)} &= \begin{cases}
        I & \text{if } p - p^2 + pq - q^2 \leqslant 0; \\
        X & \text{if } p - p^2 + pq - q^2 \geqslant 0;
    \end{cases} \nonumber \\
    W^{(0,1)} &= X; \nonumber \\
    W^{(1,0)} &= X; \nonumber \\
    W^{(1,1)} &= \begin{cases}
        I & \text{if } p - p^2 + pq - q^2 \leqslant 0; \\
        X & \text{if } p - p^2 + pq - q^2 \geqslant 0.
    \end{cases}
\end{align}

\noindent This gives rise to two different circuits, shown in Figure \ref{fig:adaptive-noisy-Z}. Through our derivation, we also obtain the optimal total error among all adaptive circuits with three uses of $\mathbb{Z}_{p,q}$, which we denote $\etaA$: \\
\begin{equation}
    \etaA = \begin{cases}
        p - p^2 + 4pq + p^3 - 2p^2 q - 2pq^2 + q^3 \\
        \hspace{4em} \text{if } p - p^2 + pq - q^2 \leqslant 0; \\
        p^2 + 4pq + q^2 - 2p^2 q - 2pq^2 \\
        \hspace{4em} \text{if } p - p^2 + pq - q^2 \geqslant 0.
    \end{cases} \label{eq:optimal-adaptive-3}
\end{equation}

\begin{figure}[t]
    \centering
    \includegraphics[scale=1.25]{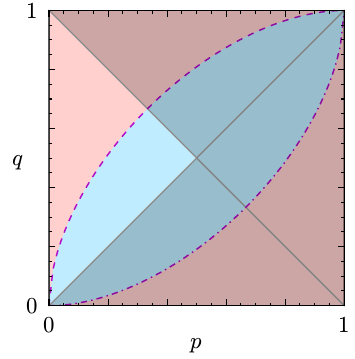}
    \caption{A plot of the values of $p$ and $q$ where we see an adaptive advantage when measuring with $N$ copies of the imperfect $Z$ measurement $\mathbb{Z}_{p,q}$. The lighter triangle on the left shows the region $0 \leqslant p \leqslant q \leqslant 1-p$ which we have been considering, with the results mirroring to the other three quadrants. In the pink region, where $p - p^2 + pq - q^2 < 0$, $p > 0$, and $p + q < 1$ (or a reflected equivalent), we have proven that there is an adaptive advantage for $N = 3$ measurements. In the blue region, it may be shown by a similar method there is an adaptive advantage for $N = 4$. The boundaries between these regions are $p - p^2 + pq - q^2 = 0$ (purple, dashed) and $q - p^2 + pq - q^2 = 0$ (purple, dash-dotted). Along the grey lines (and the axes), there is no adaptive advantage for any $N$.}
    \label{fig:noisy-Z-adaptive-advantage}
\end{figure}

\noindent We note that if $p - p^2 + pq - q^2 \geqslant 0$, the choice $W^{(x_1,x_2)} = X$ is always optimal regardless of $x_1,x_2$, so our optimal circuit can be implemented non-adaptively, such as in Figure \ref{fig:adaptive-noisy-Z}(b). We can also infer this from the optimal non-adaptive total error \eqref{eq:optimal-non-adaptive-3}, which is equal to the optimal adaptive total error $\etaA = p^2 + 4pq + q^2 - 2p^2 q - 2pq^2$ in this region, so there is no adaptive advantage.

Meanwhile, if $p - p^2 + pq - q^2 < 0$, there is no such choice of $W^{(x_1,x_2)}$ independent of $x_1,x_2$, since we are required to take $W^{(0,0)} = I$ and $W^{(0,1)} = X$ for optimality. This alone is not a proof of adaptive advantage, since we did ignore some equally optimal circuits along the way that may have been non-adaptive. However, we can compare $\etaA$ with $\etaNA$, and show that for $p,q$ with $p - p^2 + pq - q^2 < 0$: \\
\begin{align}
    & \eta^{(NA)} - \eta^{(A)} \nonumber \\
    &= \begin{cases}
        2p - 2p^2 - 4pq + 2p^2 q + 2pq^2 \\
        \hspace{4em} \text{if } 3p - p^2 - pq - q^2 \leqslant 0; \\
        - p + 2p^2 - p^3 + q^2 - q^3 \\
        \hspace{4em} \text{if } 3p - p^2 - pq - q^2 \geqslant 0;
    \end{cases} \nonumber \\
    &= \begin{cases}
        2 (1-p-q) p (1-q) \\
        \hspace{4em} \text{if } 3p - p^2 - pq - q^2 \leqslant 0; \\
        - (1-p-q) (p - p^2 + pq - q^2) \\
        \hspace{4em} \text{if } 3p - p^2 - pq - q^2 \geqslant 0;
    \end{cases} \nonumber \\
    &> 0 \quad \text{if } p - p^2 + pq - q^2 < 0, \, p > 0, \, p + q < 1.
\end{align}

\noindent Therefore, we have shown an adaptive advantage for three uses of the imperfect $Z$ measurement $\mathbb{Z}_{p,q}$ for any $p,q$ with $ p - p^2 + pq - q^2 < 0$, $p > 0$, and $p + q < 1$.

More broadly, we can extend these results to the entire parameter space, with the presence or otherwise of adaptive advantage illustrated in Figure \ref{fig:noisy-Z-adaptive-advantage}. The pink region corresponds to where an adaptive circuit like in Figure \ref{fig:adaptive-noisy-Z}(a) can outperform any non-adaptive circuit, whereas the blue region corresponds to where no such advantage of adaptivity exists. This is noteworthy because it shows us that while adaptivity is powerful, whether it actually helps to lower the minimal total error is a more subtle question. We have seen instances where it does, and we have seen instances where it does not.

\begin{figure*}[t]
    \centering
    \includegraphics[scale=1.25]{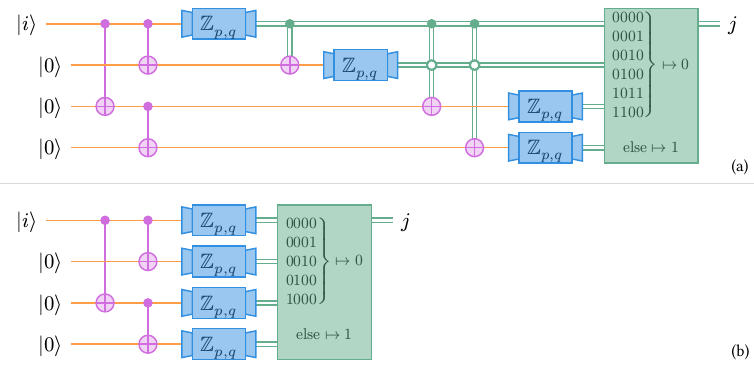}
    \caption{Example \textbf{(a)} adaptive and \textbf{(b)} non-adaptive circuits with four imperfect $Z$ measurements $\mathbb{Z}_{p,q}$ that are optimal when $p - p^2 + pq - q^2 \geqslant 0$. The adaptive circuit gives strictly smaller total error, and hence an adaptive advantage exists, for all $p,q$ where $0 \leqslant p \leqslant q \leqslant 1-p$ and $p - p^2 + pq - q^2 \geqslant 0$ except the edge cases $p = q$ and $p + q = 1$.}
    \label{fig:noisy-Z-four-measurements}
\end{figure*}

\subsubsection{Adaptive advantage in the remaining cases}

\noindent To look for an adaptive advantage when $p - p^2 + pq - q^2 \geqslant 0$, we will allow ourselves four uses of $\mathbb{Z}_{p,q}$. We can do the equivalent calculations with the extra measurement, and we find that the circuit shown in Figure \ref{fig:noisy-Z-four-measurements}(a) is optimal among adaptive circuits for any choice of $p,q$ satisfying both $p - p^2 + pq - q^2 \geqslant 0$ and $0 \leqslant p \leqslant q \leqslant 1-p$. It has:
\begin{align}
    \varepsilon_0 &= 5p^2 - 6p^3 - 2p^2 q + 2pq^2 + 2p^4 + p^3 q - pq^3; \nonumber \\
    \varepsilon_1 &= p^2 - p^3 - p^2 q + pq^2 + 3q^3 + p^3 q - pq^3 - 2q^4; \nonumber \\
    \etaA &= 6p^2 - 7p^3 - 3p^2 q + 3pq^2 + 3q^3 \nonumber \\
    &\hspace{6em}+ 2p^4 + 2p^3 q - 2pq^3 - 2q^4,
\end{align}

\noindent while the optimal non-adaptive circuit is shown in Figure \ref{fig:noisy-Z-four-measurements}(b), with errors:
\begin{align}
    \varepsilon_0 &= 6p^2 - 8p^3 + 3p^4; \nonumber \\
    \varepsilon_1 &= 3q^3 - 2q^4; \nonumber \\
    \etaNA &= 6p^2 - 8p^3 + 4q^3 + 3p^4 - 3q^4.
\end{align}

\noindent The difference between these total errors is:
\begin{align}
    &\etaNA - \etaA \nonumber \\
    &= -p^3 + 3p^2 q - 3pq^2 + q^3 + p^4 - 2p^3 q + 2pq^3 - q^4 \nonumber \\
    &= (q-p)^3 (1-p-q),
\end{align}

\noindent which is always greater than zero except along the boundaries $p = q$ and $p + q = 1$. Therefore, we see an adaptive advantage in the blue region for $N = 4$, and indeed, the adaptive advantage remains for the pink region as well, although using other circuits than the ones shown in Figure \ref{fig:noisy-Z-four-measurements}.

We are left only with six lines where we are yet to observe an adaptive advantage, shown in grey and along the axes in Figure \ref{fig:noisy-Z-adaptive-advantage}. Using the usual symmetry considerations, we can classify these into three cases, corresponding to the three edges of the triangle $0 \leqslant p \leqslant q \leqslant 1-p$:
\begin{itemize}
    \item $p = q$, the symmetric imperfect $Z$;
    \item $p = 1-q$, where $Z_0 = (1-p)I$ and $Z_1 = pI$ are both proportional to the identity, making this a trivial measurement with outcomes uncorrelated to the input state;
    \item $p = 0$, with $0 < q < 1$, which we call a \emph{semi-perfect $Z$}, since a measurement result of $j = 1$ means that $\ket{i} = \ket{1}$ with certainty, while a measurement result of $j = 0$ allows no similarly absolute deduction. Symmetrically, the same behaviour will be seen if $q = 0$, $p = 1$, or $q = 1$, with the other parameter neither $0$ nor $1$.
\end{itemize}

\noindent We argue in Appendix \ref{sec:edge-cases} that in all of these edge cases, there will never be an adaptive advantage, no matter how large $N$ is.

\subsection{Quantifying adaptive advantage} \label{sec:relative-adaptive-advantage}

\noindent As yet, we have only considered whether an adaptive advantage is achievable or not. We will now investigate the ratio\footnote{Note that $R$ is defined only if $\etaA \neq 0$. However, we don't worry about the case when it is zero, as in Section \ref{sec:qubit-measurements} we show that this can only happen if $\etaNA = 0$ as well, in which case we don't need adaptivity.} of the minimal non-adaptive total error to the minimal adaptive total error, which we call the \emph{relative adaptive advantage} $R$: 
\begin{equation}
    R = \frac{\etaNA}{\etaA} \geqslant 1. \label{eq:relative-adaptive-advantage}
\end{equation}

\begin{figure*}[t]
    \centering
    \includegraphics{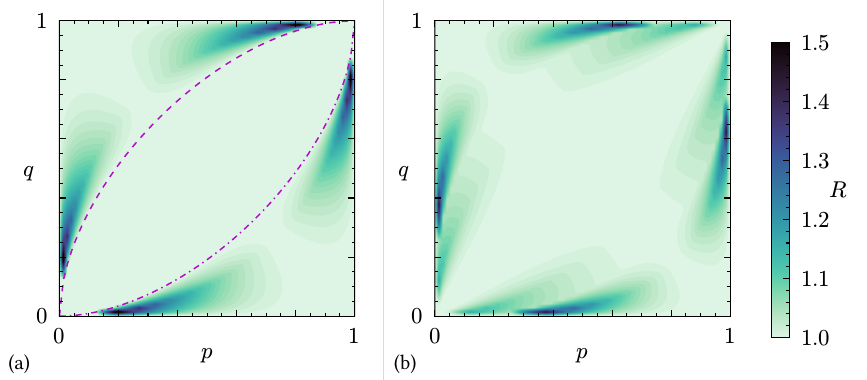}
    \caption{Contour plots of the relative adaptive advantage for $N$ uses of the imperfect $Z$ measurement $\mathbb{Z}_{p,q}$, for \textbf{(a)} $N = 3$ and \textbf{(b)} $N = 4$. The region between the purple dashed and dash-dotted lines in (a) is the region of no adaptive advantage for three measurements, as in Figure \ref{fig:noisy-Z-adaptive-advantage}.}
    \label{fig:adaptive-advantage-contours}
\end{figure*}

\noindent Figure \ref{fig:adaptive-advantage-contours} shows the behaviour of the relative adaptive advantage in the cases we have observed so far: either $N = 3$ or $N = 4$ imperfect $Z$ measurements $\mathbb{Z}_{p,q}$. In both cases, we see that most of the time $R$ is only slightly greater than 1. However, for certain values of $p$ and $q$, particularly when they are small and very asymmetric, the advantage can be far greater. Indeed, we will now show that the relative adaptive advantage is unbounded.

\begin{figure*}[t]
    \centering
    \includegraphics[scale=1.25]{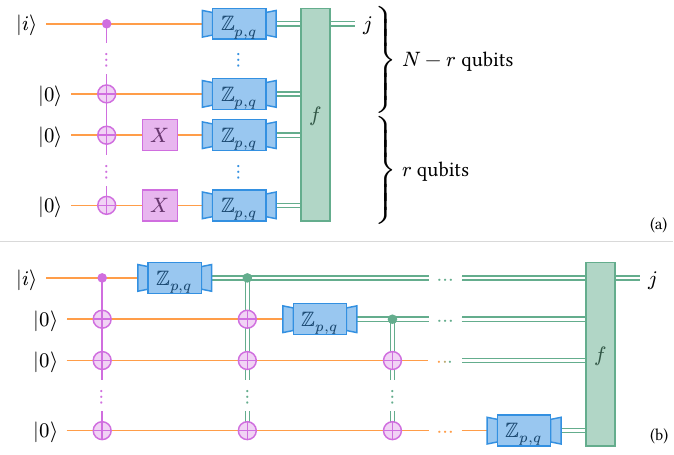}
    \caption{\textbf{(a)} A family of non-adaptive circuits with $N$ measurements of $\mathbb{Z}_{p,q}$, parametrised by $0 \leqslant r \leqslant \frac{1}{2} N$, with $f$ given by maximum likelihood decoding in each case. One such circuit will have total error equal to the minimal non-adaptive total error $\etaNA$. \textbf{(b)} An example adaptive circuit with $N$ measurements of $\mathbb{Z}_{p,q}$ that we use to outperform all non-adaptive circuits. After each measurement result of $1$, the remaining qubits are all bit-flipped with an $X$ gate, while after a measurement result of $0$, we instead apply $I$.}
    \label{fig:arbitrary-adaptive-advantage}
\end{figure*}

Supposing that the number of measurements $N \geqslant 3$ is fixed, we take the parameters of the imperfect $Z$ measurement to be $q = 2^{-N}$ and $p = q^N = 2^{-N^2}$. By following the same argument as we did for three measurements, the optimal non-adaptive circuit will without loss of generality use $\mathbb{Z}_{p,q}^{\otimes N}$ to measure the states $\rho_i = \ketbra{\mathbf{y}_i}{\mathbf{y}_i}$, where $\ket{\mathbf{y}_0} = \ket{0}^{\otimes (N-r)} \ket{1}^{\otimes r}$ for some $r \leqslant \frac{1}{2} N$, and $\ket{\mathbf{y}_1} = \ket{\mathbf{y}_0 \oplus \mathbf{1}} = \ket{1}^{\otimes (N-r)} \ket{0}^{\otimes r}$. One such circuit producing these desired states is shown in Figure \ref{fig:arbitrary-adaptive-advantage}(a).

In Appendix \ref{sec:arbitrary-adaptive-advantage}, we prove a lower bound on the total error of any non-adaptive optimal circuit (and hence on $\etaNA$) of:
\begin{equation}
    \etaNA \geqslant N q^N.
\end{equation}

\noindent In the adaptive case, we upper bound the minimal total error by an upper bound on the total error of the adaptive circuit shown in Figure \ref{fig:arbitrary-adaptive-advantage}(b). It is defined by copying the input with CNOT gates to a total of $N$ qubits, and then measuring each qubit one at a time. After each measurement, if we get a measurement result of 1, we apply $X$ to each of the remaining qubits, while if we get a measurement result of 0, we apply $I$. The minimum adaptive total error $\etaA$ is upper bounded by the total error of this circuit, allowing us to prove in Appendix \ref{sec:arbitrary-adaptive-advantage} that:
\begin{equation}
    \etaA \leqslant 3 q^N.
\end{equation}

\noindent Thus, we find that:
\begin{equation}
    R = \frac{\etaNA}{\etaA} \geqslant \frac{Nq^N}{3q^N} = \frac{N}{3},
\end{equation}

\noindent which is arbitrarily large by choosing $N$ large enough to begin with. Therefore, we can have a relative adaptive advantage of as much as we like, for carefully chosen $N$, $p$, and $q$. We note that this is mostly a mathematical curiosity, as the parameters for $p$ and $q$ are very carefully selected, and very quickly become so minuscule as to be zero for all practical purposes. However, it does leave the door open for significant adaptive advantage to be found in practical parameter regimes.

In summary, through the lens of imperfect $Z$ measurements, we have seen how adaptivity can help us to reduce the total error in measurement procedures that involve multiple measurements taking place. We have also seen cases in which adaptivity cannot improve on the best non-adaptive circuits. Therefore, adaptivity is a powerful tool to be able to use, but finding how best to use it, or even if it is worth using at all, is a difficult question to answer.

\section{General measurements} \label{sec:qubit-measurements}

\noindent We now consider the general problem of approximating a perfect $Z$ measurement, when the only measurement we can perform is non-projective. We suppose that we have $N$ uses of the measuring device, but other than this, we are unrestricted, as we can perform arbitrary quantum channels (i.e. completely positive, trace-preserving maps, which can be implemented as unitaries if an ancillary system is adjoined) at any time with perfect fidelity, and we are not limited by any time or space considerations. Note that we cannot allow operations that may in other places be called channels, but do not fit the definition of completely positive and trace-preserving, for instance those that are trace-decreasing or have classical outputs, as they would require measurements which are our restricted resource.

In this section, we will show how the optimal non-adaptive and adaptive circuits for $N$ uses of our measuring device can be found. In Section \ref{sec:optimal-non-adaptive}, we will look at the non-adaptive case, which will involve finding the best way to simultaneously optimise the quantum circuit before the measurements and the classical decision problem after them. Then, in Section \ref{sec:optimal-adaptive}, we will switch to the adaptive case, where we will show that the circuit has a recursive structure as shown in Figure \ref{fig:adaptive-recursive}. Thus, we can recursively define a set of achievable errors $(\varepsilon_0, \varepsilon_1)$ for each number of uses of the measuring device \eqref{eq:achievable-errors-0}, \eqref{eq:achievable-errors}, and thus determine the optimal adaptive circuit by finding the optimal point in this set. We will also prove a number of useful results along the way, such as Proposition \ref{prop:no-adaptive-advantage} which gives a sufficient condition for no adaptive advantage to exist, and Corollary \ref{cor:adaptive-total-error}, which reduces the search space for optimal adaptive circuits to a finite set. These we will apply to several example measurements in Section \ref{sec:qubit-examples} to find where adaptive advantage does and does not exist.

In both this section and the next, proofs of stated results can be found in Appendix \ref{sec:proofs}.

\subsection{Notation}

\noindent Before we start the general case, we need to define some notation. We let $\mathbb{M}$ be a POVM with $m$ outcomes:
\begin{equation}
    \mathbb{M} = \{ M_x \}_{x \in [m]},
\end{equation}

\noindent where\footnote{The more conventional $[m] = \{1,2,...,m\}$ would work equivalently here, as the elements are used only as labels. However, our convention appeals to the usual numbering of computational basis states, which are indexed from 0.} $[m] = \{0, 1, ..., m-1\}$, and:
\begin{equation}
    0 \leqslant M_x \in \mathcal{B}(\mathcal{H}) \ \ \forall x \in [m]; \quad \sum_{x \in [m]} M_x = \mathbb{I}_\mathcal{H}.
\end{equation}

\noindent We do not require $\mathcal{H}$ to be two-dimensional, as we can still approximate the $Z$ measurement on a qubit by taking an input qubit state and mapping it to a state in $\mathcal{H}$. As such, we let $h = \dim \mathcal{H}$. 

If $\mathbb{M} = \{ M_x \}_{x \in [m]}$ and $\mathbb{N} = \{ N_y \}_{y \in [n]}$ are two POVMs, then we define their tensor product as:
\begin{equation}
    \mathbb{M} \otimes \mathbb{N} = \left\{ M_x \otimes N_y \right\}_{x \in [m], y \in [n]},
\end{equation}

\noindent and similarly, the $N$-fold tensor product of $\mathbb{M}$ with itself is:
\begin{equation}
    \mathbb{M}^{\otimes N} = \left\{ M_\mathbf{x} = \bigotimes_{k = 1}^N M_{x_k} \right\}_{\mathbf{x} \in [m]^N} .
\end{equation}

\noindent We also define four further useful notions. Firstly, the \emph{probability range} \cite{buscemi-2005, heinosaari-2020} of a POVM $\mathbb{M}$ is:
\begin{equation}
    \mathcal{R}(\mathbb{M}) = \{ (\mathrm{tr}(\rho M_0), \mathrm{tr}(\rho M_1), ..., \mathrm{tr}(\rho M_{m-1})) : \rho \in \mathcal{D}(\mathcal{H}) \}, \label{eq:probability-range}
\end{equation}

\noindent i.e. the set of all possible probability distributions that can result from a measurement with $\mathbb{M}$, a subset of the entire probability simplex in $m$ dimensions.

Secondly, the \emph{spectral diameter} \cite{parlett-2003} of a Hermitian operator $M$ is the difference between its maximum and minimum eigenvalues:
\begin{equation}
    \mathrm{diam}(M) = \lambda_\textrm{max}(M) - \lambda_\textrm{min}(M), \label{eq:spectral-diameter}
\end{equation}

\noindent and the \emph{spectral range} is the interval between the minimum and maximum eigenvalues:
\begin{equation}
    \mathcal{I}(M) = [\lambda_\textrm{min}(M), \lambda_\textrm{max}(M)]. \label{eq:spectral-range}
\end{equation}

\noindent Finally, the \emph{convex hull} of a set $S$ is denoted by:
\begin{equation}
    \mathrm{conv}(S) = \{ (1-t)x + ty : x,y \in S, t \in [0,1] \}. \label{eq:convex-hull}
\end{equation}

\subsection{Optimal non-adaptive circuits} \label{sec:optimal-non-adaptive}

\begin{figure}[t]
    \centering
    \includegraphics[scale=1.25]{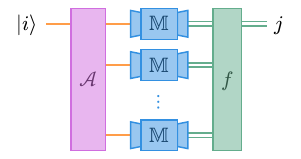}
    \caption{The general non-adaptive circuit we can use to improve the measurement $\mathbb{M}$, where we have $N$ uses of $\mathbb{M}$ in total. Here, $\ket{i} \in \{ \ket{0}, \ket{1}, ..., \ket{d-1} \}$ is our input state, $\mathcal{A} : \mathcal{B}(\mathbb{C}^d) \to \mathcal{B}(\mathcal{H}^{\otimes N})$ is a completely positive, trace preserving map, $f : [m]^N \to [d]$ is an arbitrary function, and $j \in [d]$ is our output. In Section \ref{sec:qubit-measurements} we take $d = 2$.}
    \label{fig:non-adaptive}
\end{figure}

\noindent We first look at the non-adaptive measurement circuits that we can have with $N$ uses of $\mathbb{M}$, for which the most general circuit is as shown in Figure \ref{fig:non-adaptive}. While we allow for a general input state $\ket{\psi}$, we still consider the error rates $\varepsilon_0$ and $\varepsilon_1$ as defined in \eqref{eq:errors-2}, and our figure of merit will remain their sum $\eta$, the total error, which we wish to minimise. This means that for purposes of optimisation, we only need to consider the possible input states $\ket{i} \in \{ \ket{0}, \ket{1} \}$.

There are two possible further generalisations to the circuit in Figure \ref{fig:non-adaptive} that we can make. First, we could have probabilistic post-processing $f$, that is, given a string $x_1 x_2 ... x_N$ of measurement outcomes, we output $j = 0$ with probability $p_{x_1 x_2 ... x_N}$ and $j = 1$ with probability $1-p_{x_1 x_2 ... x_N}$. Second, we could have a classical mixture of such circuits \cite{haapasalo-2011}, where we make a random choice of a circuit from a predetermined ensemble that we are going to run at the very start.

However, the linearity of the total error figure of merit in $\varepsilon_0$ and $\varepsilon_1$ allows us to disregard these generalisations immediately, as they won't decrease our optimal $\eta$. To see why, suppose we have a mixed protocol where circuit $\mathcal{C}_1$ is performed with probability $\alpha$, and circuit $\mathcal{C}_2$ is performed with probability $1-\alpha$. The overall total error is simply a linear combination of those for the individual circuits:
\begin{equation}
    \eta = \alpha \eta_{\mathcal{C}_1} + (1-\alpha) \eta_{\mathcal{C}_2},
\end{equation}

\noindent and so is minimised at either $\alpha = 0$ or $\alpha = 1$ depending on which circuit had the lower total error in the first place. Thus mixed protocols can give no advantage in total error over fixed ones.

Similarly, a circuit with non-deterministic $f$ can be thought of as a classical mixture of circuits with deterministic $f$. Hence, we may assume $f$ to be deterministic, as randomness cannot decrease the total error. These cases are still important to note, as for instance in \cite{linden-2025} with a different choice of figure of merit, mixed protocols can be strictly better.

We then define $\etaNA(\mathbb{M},N)$ to be the minimum total error achievable with such a non-adaptive circuit:
\begin{equation}
    \etaNA(\mathbb{M},N) = \min_{\mathcal{A}, f} \{ \varepsilon_0 + \varepsilon_1 \}.
\end{equation}

\noindent While optimising over both $\mathcal{A}$ and $f$ is a daunting task, if one is fixed, finding the minimal $\eta$ achievable by an optimal choice of the other is reasonably straightforward. This reduces the search space to either all possibilities for $\mathcal{A}$, or all possibilities for $f$, depending on which we choose to fix. We will investigate both in the rest of this section, and will see that depending on $\mathbb{M}$, either direction could be the right choice.

First, suppose that $f$ is fixed. As in our earlier example, we define $\rho_0 = \mathcal{A}(\ketbra{0}{0})$ and $\rho_1 = \mathcal{A}(\ketbra{1}{1})$, and the effective POVM element $Q_0^{(f)}$ by:
\begin{equation}
    Q_0^{(f)} = \sum_{\mathbf{x} : f(\mathbf{x}) = 0} M_\mathbf{x}.
\end{equation}

\noindent Then, for any given $\mathcal{A}$ and $f$, $\varepsilon_0$ and $\varepsilon_1$ are given by:
\begin{equation}
    \varepsilon_0 = 1 - \mathrm{tr} \left( \rho_0 Q_0^{(f)} \right); \quad \varepsilon_1 = \mathrm{tr} \left( \rho_1 Q_0^{(f)} \right).
\end{equation}

\noindent Our approach will be to optimise over $\rho_0$ and $\rho_1$ separately, which we justify with the following lemma:

\begin{lemma}
    Let $\rho_i \in \mathcal{D}(\mathcal{H})$ for each $i \in [d]$. Then, there exists\footnote{This channel $\mathcal{A}$ may be more familiar as the \emph{measure-and-prepare} channel, where the input $\ket{i}$ is measured projectively in the computational basis, with the result being $i$ with certainty, and then the state $\rho_i$ can be prepared according to that classical information. With only our noisy measuring device $\mathbb{M}$, this procedure would not correctly implement the channel. However, the name is misleading, and it can be implemented without a measurement as shown in the proof, so we are still allowed to use it.} a quantum channel $\mathcal{A}$ such that $\mathcal{A}(\ketbra{i}{i}) = \rho_i$ for all $i$. \label{lem:channel-existence}
\end{lemma}

\begin{proof}
    See Appendix \ref{sec:proof-channel-existence}.
\end{proof}

\noindent Using Lemma \ref{lem:channel-existence} with $d = 2$, we see that we can choose $\rho_0$ and $\rho_1$ independently of each other, since we can always be guaranteed of the existence of an appropriate channel $\mathcal{A}$. This allows us achieve a minimal total error of:
\begin{align}
    \eta &= \min_\mathcal{A} \left\{  1 - \mathrm{tr} \left( \rho_0 Q_0^{(f)} \right) + \mathrm{tr} \left( \rho_1 Q_0^{(f)} \right) \right\} \nonumber \\ 
    &= 1 - \max_{\rho_0} \left\{ \mathrm{tr} \left( \rho_0 Q_0^{(f)} \right) \right\} + \min_{\rho_1} \left\{ \mathrm{tr} \left( \rho_1 Q_0^{(f)} \right) \right\} \nonumber \\
    &= 1 - \lambda_\textrm{max} \left( Q_0^{(f)} \right) + \lambda_\textrm{min} \left( Q_0^{(f)} \right) \nonumber \\
    &= 1 - \mathrm{diam} \left( Q_0^{(f)} \right),
\end{align}

\noindent and so minimising over all possible choices of $f$ gives:
\begin{equation}
    \etaNA(\mathbb{M},N) = 1 - \max_f \left\{ \mathrm{diam} \left( Q_0^{(f)} \right) \right\}, \label{eq:maximise-spectral-diameter}
\end{equation}

\noindent where $\mathrm{diam} \left( Q_0^{(f)} \right)$ is the spectral diameter \eqref{eq:spectral-diameter} of $Q_0^{(f)}$.

This gives one way to reduce the search space in finding the optimal non-adaptive circuit. Conversely, suppose $\mathcal{A}$ is fixed, so the states $\rho_i$ are fixed. Then:
\begin{align}
    \eta &= \min_f \left\{ \mathrm{tr} \left( \rho_0 \sum_{\mathbf{x} : f(\mathbf{x}) = 1} M_\mathbf{x} \right) + \mathrm{tr} \left( \rho_1 \sum_{\mathbf{x} : f(\mathbf{x}) = 0} M_\mathbf{x} \right) \right\} \nonumber \\
    &= \min_f \left\{ \sum_{\mathbf{x} \in [m]^N} \sum_{i \in \{0, 1\}} [f(\mathbf{x}) \neq i] \ \mathrm{tr}(\rho_i M_\mathbf{x}) \right\} \nonumber \\
    &= \sum_{\mathbf{x} \in [m]^N} \min_{f(\mathbf{x}) \in \{0, 1\}} \left\{ \sum_{i \in \{0, 1\}} [f(\mathbf{x}) \neq i] \ \mathrm{tr}(\rho_i M_\mathbf{x}) \right\} \nonumber \\
    &= \sum_{\mathbf{x} \in [m]^N} \min_{i \in \{0, 1\}} \left\{ \mathrm{tr}(\rho_i M_\mathbf{x}) \right\},
\end{align}

\noindent with $f$ chosen according to maximum-likelihood decoding, as we observed before in Section \ref{sec:example}. Here, $[P]$ is the Iverson bracket, equal to 0 if $P$ is false and 1 if $P$ is true.

From each fixed state $\rho_i$, we obtain a stochastic vector $\mathbf{q}_i \in \mathcal{R}(\mathbb{M}^{\otimes N})$, with elements $q_i^{(\mathbf{x})} = \mathrm{tr}(\rho_i M_\mathbf{x})$. Then, we can write the minimal non-adaptive total error as:
\begin{equation}
    \etaNA(\mathbb{M},N) = \min_{\mathbf{q}_0, \mathbf{q}_1 \in \mathcal{R}(\mathbb{M}^{\otimes N})} \left\{ \sum_{\mathbf{x} \in [m]^N} \min_{i \in \{0, 1\}} \left\{ q_i^{(\mathbf{x})} \right\} \right\}. \label{eq:total-error-preprocessing}
\end{equation}

\noindent We thus have the two equations \eqref{eq:maximise-spectral-diameter} and \eqref{eq:total-error-preprocessing} that we can use to find $\etaNA$. In general, \eqref{eq:maximise-spectral-diameter} is the easier approach, since the number of possible post-processing functions $f$ is $2^{m^N}$, i.e. finite, whereas the number of possible channels $\mathcal{A}$ is infinite, with no particular structure to the optimisation that we can find where continuity is helpful. The difficulty of using \eqref{eq:total-error-preprocessing} depends heavily on the complexity of the set $\mathcal{R}(\mathbb{M}^{\otimes N})$. The probability range can be difficult to compute anyway, and it isn't made any easier by an $N$-fold tensor product of POVMs. However, there is a sufficient condition that simplifies it hugely, in which case the approach of \eqref{eq:total-error-preprocessing} becomes far superior.

We call the POVM $\mathbb{M} = \{ M_x \}_{x \in [m]}$ \emph{abelian} if all of its elements commute, or equivalently if they are all diagonal when expressed in some basis $\{ \ket{v_a} \}_ {a \in [h]}$. The probability range of an abelian POVM already has a simple form \cite[Proposition 7.4]{buscemi-2005}:
\begin{equation}
    \mathcal{R}(\mathbb{M}) = \mathrm{conv}(\{\mathbf{q}_a\}_{a \in [h]}),
\end{equation}

\noindent where for each $a \in [h]$, $\mathbf{q}_a \in \mathcal{R}(\mathbb{M})$ is obtained from measuring $\ket{v_a}$ with $\mathbb{M}$, i.e. $q_a^{(x)} = \bra{v_a} M_x \ket{v_a}$.

This continues with the $N$ parallel uses of $\mathbb{M}$, since $\mathbb{M}^{\otimes N}$ remains an abelian POVM, with diagonalising basis $\{\ket{v_\mathbf{a}} = \ket{v_{a_1}} \ket{v_{a_2}} ... \ket{v_{a_N}} : \mathbf{a} \in [h]^N\}$. We have that:
\begin{equation}
    \mathcal{R}(\mathbb{M}^{\otimes N}) = \mathrm{conv}( \{ \mathbf{q}_\mathbf{a} \}_{\mathbf{a} \in [h]^N}),
\end{equation}

\noindent where $\mathbf{q}_\mathbf{a}$ is obtained from measuring $\ket{v_\mathbf{a}}$ with $\mathbb{M}^{\otimes N}$.

Then, by linearity, only the extremal points of the convex set $\mathcal{R}(\mathbb{M}^{\otimes N})$ need to be considered when optimising total error. We infer that for an abelian POVM $\mathbb{M}$, \eqref{eq:total-error-preprocessing} simplifies to:
\begin{equation}
    \etaNA(\mathbb{M},N) = \min_{\mathbf{a}_0, \mathbf{a}_1 \in [h]^N} \left\{ \sum_{\mathbf{x} \in [m]^N} \min_{i \in \{0, 1\}} \left\{ q_{\mathbf{a}_i}^{(\mathbf{x})} \right\} \right\}, \label{eq:abelian-non-adaptive-total-error}
\end{equation}

\noindent which is again a finite optimisation, this time over $h^{2N}$ possibilities.

\begin{figure}[t]
    \centering
    \includegraphics[scale=1.25]{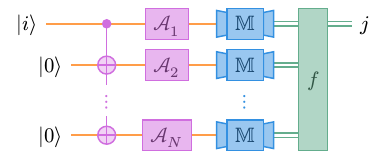}
    \caption{The simplified optimal non-adaptive circuit for an abelian POVM $\mathbb{M}$. First, the basis state $\ket{i}$ is copied across all qubits using CNOT gates, followed by local channels $\mathcal{A}_k: \ketbra{i}{i} \mapsto \ketbra{v_{a_{ik}}}{v_{a_{ik}}}$, where $\mathbf{a}_i = (a_{ik})_{k = 1}^N$. In the $d = 2$ case, the channels $\mathcal{A}_k$ can be assumed to be isometries without loss of generality.}
    \label{fig:non-adaptive-abelian}
\end{figure}

\begin{figure*}[t]
    \centering
    \includegraphics[scale=1.25]{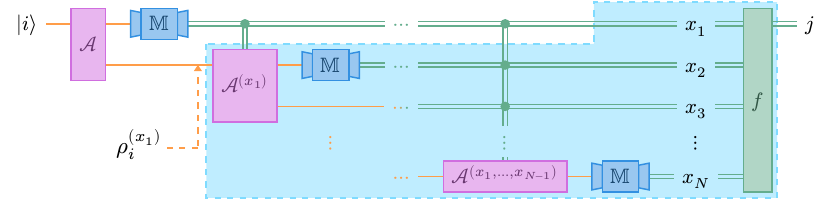}
    \caption{The general adaptive circuit we can use to improve the measurement $\mathbb{M}$, where we have $N$ uses of $\mathbb{M}$ in total. Adaptivity means that the channels $\mathcal{A}, \mathcal{A}^{(x_1)}, ...$ can depend on prior measurement results, as notated in the superscripts. The state of the second qubit after the first is measured is $\rho_{i}^{(x_1)}$. The portion of the circuit highlighted in blue takes a very similar form to the original circuit, except for the input state being different, the dependence on $x_1$, and the number of measurements being one fewer. We exploit this similarity to simplify the general circuit.}
    \label{fig:adaptive}
\end{figure*}

This also simplifies the circuit, since the measured states are always product states, allowing them to be generated using only CNOT gates and local single qubit channels, as shown in Figure \ref{fig:non-adaptive-abelian}. Indeed, because the local channels $\mathcal{A}_k$ are defined only by their action on two orthogonal inputs, we can assume they are isometries, by the following argument on the strings defining the measured states $\mathbf{a}_i = (a_{ik})_{k = 1}^N$:
\begin{itemize}
    \item If $a_{0k} \neq a_{1k}$, then $\mathcal{A}_k$ sends an orthonormal basis to orthonormal states, so is an isometry;
    \item Otherwise, $a_{0k} = a_{1k}$, in which case $\mathcal{A}_k$ is a trivial channel that sends all states to the same place. This means that the measurement result of the $k$th qubit has no correlation to the input and is ignored by maximum likelihood decoding. Therefore we can replace $\mathcal{A}_k$ with an isometry and it won't make any difference.
\end{itemize}

\noindent There are $h(h-1)$ possible isometries mapping the computational basis of $\mathbb{C}^2$ to the diagonalising basis of $\mathbb{M}$ in $\mathcal{H}$. Moreover, any permutation of the $N$ qubit lines in Figure \ref{fig:non-adaptive-abelian} will not affect the total error, provided that the same permutation is made within $f$, so only the number of each isometry among the set of channels $\mathcal{A}_k$ is important. Hence, by the stars and bars argument, the number of circuits that need to be considered is $\binom{N + h(h-1) - 1}{h(h-1) - 1}$, i.e. polynomial in $N$. Thus, for abelian POVMs, we can find the optimal non-adaptive protocol much more efficiently than in general.

\subsection{Optimal adaptive circuits} \label{sec:optimal-adaptive}

\noindent We now turn our attention to adaptive circuits, again with $N$ uses of our measuring device $\mathbb{M}$. The most general\footnote{All of the measurements we are using here are destructive, but there is no loss in generality here, due to the figure of merit allowing us to only consider computational basis input states $\ket{\psi} = \ket{i}$. We can replicate $\ket{i}$ as much as we like at the start using CNOT gates, and we also have all previous measurement results. Then, by Lemma \ref{lem:channel-existence}, getting to a would-be post-measurement state only requires applying a channel dependent on these measurement results to one of the copies of $\ket{i}$.} such circuit is shown in Figure \ref{fig:adaptive}. Before each measurement, the quantum state is split by a channel across two systems, one of which will be measured, and one of which carries on to the next step. Exactly what channel is performed may change depending on previous measurement results, because of adaptivity. This repeats for $N$ measurements, and then the measurement results are classically processed as before to give our output.

\begin{figure*}[t]
    \centering
    \includegraphics[scale=1.25]{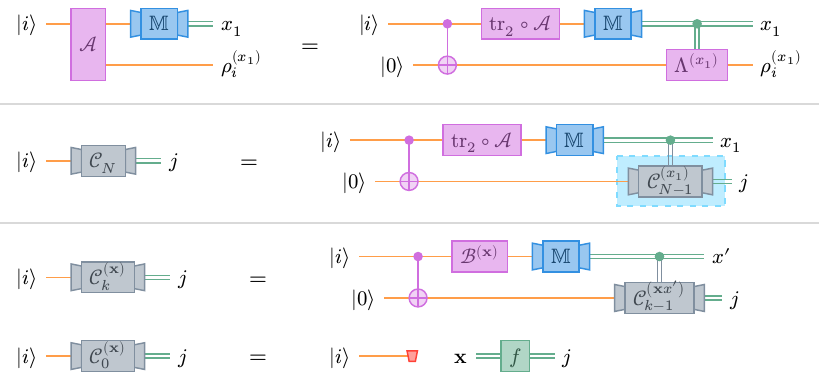}
    \caption{\textbf{(a)} Two circuits with the same action on $\ket{i}$. On the left is a snippet of Figure \ref{fig:adaptive}, up to and including the first measurement. On the right is an equivalent circuit where the state $\ket{i}$ is preserved in the lower register, until right at the end where it is processed into $\rho_i^{(x_1)}$, through the channel $\Lambda^{(x_1)}$ as defined in the main text. \textbf{(b)} We let $\mathcal{C}_N$ be the full circuit shown in Figure \ref{fig:adaptive}, with an input of $\ket{i} \in \{ \ket{0}, \ket{1}, ..., \ket{d-1} \}$, and an output of $j \in [d]$. By the transformation in (a), we can write an equivalent circuit in terms of a CNOT gate, a single qubit channel, one use of $\mathbb{M}$ with result $x_1$, and a circuit $\mathcal{C}_{N-1}^{(x_1)}$ constructed as the composition of $\Lambda^{(x_1)}$ with the blue highlighted part in Figure \ref{fig:adaptive}. \textbf{(c)} The recursive form of the general adaptive circuit, obtained from repeating the same transformation as shown in (b). When there are $k$ measurements left to make, each circuit $\mathcal{C}_k^{(\mathbf{x})}$ (dependent on previous measurement results $\mathbf{x} = (x_1, x_2, ..., x_{N-k})$) can be written in terms of a CNOT gate, a single qubit channel $\mathcal{B}^{(\mathbf{x})}$, one use of $\mathbb{M}$ with output $x' = x_{N-k+1}$, and circuits $\mathcal{C}_{k-1}^{(\mathbf{x} x')}$ with $k-1$ measurements. When $k = 0$, there are no measurements left to make, so the output is solely dependent on the measurement results $\mathbf{x} = (x_1, x_2, ..., x_N)$, via classical processing $f$. In Section \ref{sec:qubit-measurements} we take $d = 2$.}
    \label{fig:adaptive-recursive}
\end{figure*}

Due to our figure of merit, we may assume our input state is either $\ket{0}$ or $\ket{1}$, in which case there are some significant simplifications we can make to the general adaptive circuit. Using Lemma \ref{lem:channel-existence}, we define a channel $\Lambda^{(x_1)} : \mathcal{B}(\mathbb{C}^2) \to \mathcal{B}(\mathcal{H})$ such that $\Lambda^{(x_1)}(\ketbra{i}{i}) = \rho_i^{(x_1)}$, where $\rho_i^{(x_1)}$ is the state of the remainder of the system after the first measurement is made as depicted in Figure \ref{fig:adaptive}. Combining this with a CNOT gate, which copies the input state $\ket{i}$, we find an equivalent form for the start of our general circuit as shown in Figure \ref{fig:adaptive-recursive}(a). This form is preferable as we now only need to perform arbitrary channels on one qubit at a time, which is both practically helpful and reduces the search space for optimal circuits.

Now, let the entire $N$ measurement circuit in Figure \ref{fig:adaptive} be called $\mathcal{C}_N$, and consider the blue highlighted part. Dependent on $x_1$, it takes an input of $\rho_i^{(x_1)}$, uses $N-1$ measurements of $\mathbb{M}$, and outputs $j$. Composing $\Lambda^{(x_1)}$ with it forms an adaptive circuit $\mathcal{C}_{N-1}^{(x_1)}$, again dependent on $x_1$, but now taking an input of $\ket{i}$ to an output of $j$ just like $\mathcal{C}_N$ does. We can thus rewrite the circuit $\mathcal{C}_N$ using the simplified start from Figure \ref{fig:adaptive-recursive}(a), and the $x_1$ dependent circuit $\mathcal{C}_{N-1}^{(x_1)}$, as shown in Figure \ref{fig:adaptive-recursive}(b).

This is exactly the same as our overall circuit, except we've already used $\mathbb{M}$ once, so there are only $N-1$ uses left. Therefore we can repeat this process, and thus we can write our circuit in the recursive form shown in Figure \ref{fig:adaptive-recursive}(c). Supposing that we have $k$ measurements left to make, then the remaining circuit depends on the previous measurement results $\mathbf{x} = (x_1, x_2, ..., x_{N-k})$, so we call it $\mathcal{C}_k^{(\mathbf{x})} = \mathcal{C}_k^{(x_1, x_2, ..., x_{N-k})}$. Each such circuit can be expressed in terms of circuits with $k-1$ measurements $\mathcal{C}_k^{(\mathbf{x},x_{N-k+1})}$, with a dependence on the next measurement result $x_{N-k+1}$. The same process repeats until there are no measurements left, leaving us with a quantum state input as well as a classical input of $N$ previous measurement outcomes. Without a measurement, we can't do anything useful with the quantum state, so all we can do is choose an output $j = f(\mathbf{x})$ dependent on the previous measurement results.

This simplification doesn't only make implementation of the circuit simpler, but it also gives us a way to study the error rates. We let $\varepsilon_0^{(\mathbf{x})}, \varepsilon_1^{(\mathbf{x})}$ be the equivalent of our original errors $\varepsilon_0, \varepsilon_1$ for the circuit $\mathcal{C}_k^{(\mathbf{x})}$, with $k$ implicit in the notation from the $N-k$ previous measurement results in the vector $\mathbf{x} \in [m]^{N-k}$. Then, for $k = 0$:
\begin{equation}
    \left(\varepsilon_0^{(\mathbf{x})}, \varepsilon_1^{(\mathbf{x})}\right) = \begin{cases} (0,1) & \text{if } f(\mathbf{x}) = 0; \\ (1,0) & \text{if } f(\mathbf{x}) = 1, \end{cases}
\end{equation}

\noindent while for $k \geqslant 1$, we sum over the measurement result $x' = x_{N-k+1}$ to obtain the recursive relation:
\begin{equation}
    \varepsilon_i^{(\mathbf{x})} = \sum_{x' \in [m]} \mathrm{tr} \left( \mathcal{B}^{(\mathbf{x})}(\ketbra{i}{i}) M_{x'} \right) \varepsilon_i^{(\mathbf{x} x')} \quad \forall \mathbf{x} \in [m]^{N-k}. \label{eq:recursive-epsilons}
\end{equation}

\noindent Here, $\mathbf{x} x'$ represents the appending of $x'$ onto the end of the string of previous measurement results $\mathbf{x}$.

If we let $q_i^{(\mathbf{x} x')} = \mathrm{tr} \left( \mathcal{B}^{(\mathbf{x})}(\ketbra{i}{i}) M_{x'} \right)$, then $\mathbf{q}_i^{(\mathbf{x})} = (q_i^{(\mathbf{x} x')})_{x' \in [m]} \in \mathcal{R}(\mathbb{M})$, where we recall that $\mathcal{R}(\mathbb{M})$ is the probability range of $\mathbb{M}$ \eqref{eq:probability-range}. Therefore, we can calculate the possible errors $(\varepsilon_0^{(\mathbf{x})}, \varepsilon_1^{(\mathbf{x})})$ achievable with $k$ measurements from the possible errors achievable with $k-1$ measurements $( \varepsilon_0^{(\mathbf{x} x')}, \varepsilon_1^{(\mathbf{x} x')})$, in combination with $\mathcal{R}(\mathbb{M})$.

In particular, we define the \emph{single-circuit achievable error sets} recursively by:
\begin{align}
    E(\mathbb{M}, 0) &= \{(0,1), (1,0)\} \label{eq:achievable-errors-0}; \\ 
    E(\mathbb{M}, N) &= \Bigg\{ \left( \sum_{x \in [m]} \varepsilon_0^{(x)} q_0^{(x)}, \sum_{x \in [m]} \varepsilon_1^{(x)} q_1^{(x)} \right) : \nonumber \\[2ex]
    &\hspace{3em} \left( \varepsilon_0^{(x)}, \varepsilon_1^{(x)} \right) \in E(\mathbb{M}, N-1) \ \forall x \in [m], \nonumber \\
    &\hspace{4em} \mathbf{q}_i = \left( q_i^{(x)} \right)_{x \in [m]} \in \mathcal{R}(\mathbb{M}) \, \forall i \in \{0,1\} \Bigg\}. \label{eq:achievable-errors}
\end{align}

\noindent Allowing for mixed protocols (where two or more possible measurement circuits are chosen between probabilistically) by taking the convex hull \eqref{eq:convex-hull} of $E(\mathbb{M}, N)$ gives us the entire set of achievable errors $(\varepsilon_0, \varepsilon_1)$:

\begin{proposition}
    A pair of errors $(\varepsilon_0, \varepsilon_1)$ can be achieved using an adaptive circuit with $N$ measurements of $\mathbb{M}$ if and only if $(\varepsilon_0, \varepsilon_1) \in \bar{E}(\mathbb{M}, N)$, where:
    \begin{equation}
        \bar{E}(\mathbb{M}, N) = \mathrm{conv}(E(\mathbb{M}, N)).
    \end{equation} \label{prop:achievable-set-is-achievable}
\end{proposition}

\begin{proof}
    See Appendix \ref{sec:proof-achievable-set-is-achievable}.
\end{proof}

\begin{figure}[t]
    \centering
    \includegraphics[scale=1.25]{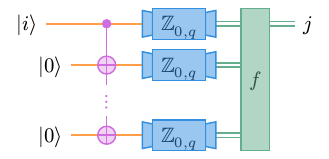}
    \caption{A non-adaptive circuit using the semi-perfect $Z$ measurement $\mathbb{Z}_{0,q}$. The CNOT chain at the start represents $N-1$ CNOT gates controlled on the first qubit acting on each of the others, which has the effect $\ket{i}\ket{0}^{\otimes (N-1)} \mapsto \ket{i}^{\otimes N}$. The function $f$ acts as an OR gate, sending $00...0 \mapsto 0$, and all other strings of measurement outcomes to $1$. We find that not only is this the optimal non-adaptive circuit, but it can't be beaten by an adaptive circuit either.}
    \label{fig:non-adaptive-optimal-noisy-Z}
\end{figure}

\noindent We hence call $\bar{E}(\mathbb{M}, N)$ the \emph{achievable error set} for $N$ adaptive uses of $\mathbb{M}$. Then, we define the quantity $\etaA(\mathbb{M}, N)$ to be the minimum total error achievable with an adaptive circuit, which is:
\begin{equation}
    \etaA(\mathbb{M}, N) = \min_{(\varepsilon_0, \varepsilon_1) \in E(\mathbb{M}, N)} \{ \varepsilon_0 + \varepsilon_1 \}, \label{eq:eta-a}
\end{equation}

\noindent where there is no need to optimise over the entire convex hull, since as before, mixed protocols cannot have lesser total error than non-mixed.

We can relate the non-adaptive and adaptive minimal total errors in various ways. If $N = 1$, we don't have the chance to do anything adaptive, so the non-adaptive and adaptive optima are the same. This will be a useful quantity to consider, so we shorten the notation:
\begin{equation}
    \eta(\mathbb{M}) = \etaNA(\mathbb{M}, 1) = \etaA(\mathbb{M}, 1).
\end{equation}

\noindent If $N > 1$, we don't expect these errors to be necessarily the same, but we know that the adaptive error can only improve on the non-adaptive error, as all non-adaptive protocols can also be considered as adaptive protocols where adaptation happens not to be used. Also, we have noted that any non-adaptive circuit is equivalent to a single measurement of $\mathbb{M}^{\otimes N}$, so:
\begin{equation}
    \etaA(\mathbb{M}, N) \leqslant \etaNA(\mathbb{M}, N) = \eta(\mathbb{M}^{\otimes N}).
\end{equation}

\noindent We can also find a lower bound for the minimal adaptive total error:

\begin{proposition}
    For any $\mathbb{M}$ and $N \geqslant 0$:
    \begin{equation}
        \etaA(\mathbb{M},N) \geqslant \eta(\mathbb{M})^N. \label{eq:adaptive-lower-bound}
    \end{equation} \label{prop:adaptive-lower-bound}
\end{proposition}

\begin{proof}
    See Appendix \ref{sec:proof-adaptive-lower-bound}.
\end{proof}

\noindent In particular, this means that if one use of $\mathbb{M}$ cannot reproduce a $Z$ measurement ($\eta(\mathbb{M}) > 0$), then neither can $N$ uses for arbitrarily large $N$. We may interpret this as a dual result to the established fact that non-orthogonal states cannot be perfectly distinguished with any finite number of measurements, only approaching zero error asymptotically.

Another way to interpret the achievable error sets is as the set of parameters of imperfect $Z$ measurements that we can reproduce with $N$ uses of $\mathbb{M}$. To see this, first we prove the following lemma:

\begin{lemma}
    Let $\mathcal{C}_N$ be any circuit with error rates $(\varepsilon_0, \varepsilon_1)$. Then, the action of the dephasing channel followed by $\mathcal{C}_N$ is the same as the action of measuring with the imperfect $Z$ measurement $\mathbb{Z}_{\varepsilon_0, \varepsilon_1}$. \label{lem:reproduction-of-noisy-Z}
\end{lemma}

\begin{proof}
    See Appendix \ref{sec:proof-reproduction-of-noisy-Z}.
\end{proof}

\noindent Then, we can combine both Proposition \ref{prop:achievable-set-is-achievable} and Lemma \ref{lem:reproduction-of-noisy-Z} to give:

\begin{corollary}
    $(\varepsilon_0, \varepsilon_1) \in \bar{E}(\mathbb{M},N)$ if and only if $N$ adaptive uses of $\mathbb{M}$ can reproduce the imperfect $Z$ measurement $\mathbb{Z}_{\varepsilon_0, \varepsilon_1}$. \label{cor:achievable-sets-as-imperfect-Zs}
\end{corollary}

\begin{figure}[t]
    \centering
    \includegraphics[scale=1.25]{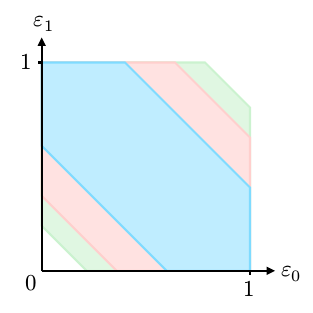}
    \caption{An example of the shape of the achievable sets $\bar{E}(\mathbb{M},1)$ (blue), $\bar{E}(\mathbb{M},2)$ (pink and blue), and $\bar{E}(\mathbb{M},3)$ (green, pink, and blue) when Proposition \ref{prop:no-adaptive-advantage} applies. The sets encroach towards the origin by a fixed proportion with each added measurement, hence the exactly exponentially decreasing total error.}
    \label{fig:no-adaptive-advantage-achievable}
\end{figure}

\noindent For instance, suppose that $(0,q) \in \bar{E}(\mathbb{M}, 1)$, i.e. that one use of $\mathbb{M}$ can exactly reproduce one use of the semi-perfect $Z$ measurement $\mathbb{Z}_{0,q}$. If we use the circuit in Figure \ref{fig:non-adaptive-optimal-noisy-Z}, where each $\mathbb{Z}_{0,q}$ is implemented using one use of $\mathbb{M}$, and $f$ is defined as the function OR, then the error rates are:
\begin{equation}
    \varepsilon_0 = q^N; \quad \varepsilon_1 = 0,
\end{equation}

\noindent since an input of $\ket{0}$ always gives measurement results $00...0$. Therefore, $q^N$ is a total error that we can achieve non-adaptively, so:
\begin{equation}
    \etaNA(\mathbb{M}, N) \leqslant q^N.
\end{equation}

\noindent In combination with Proposition \ref{prop:adaptive-lower-bound}, we can thus find a sufficient condition for there to be no adaptive advantage:

\begin{proposition}
    If $(0, \eta(\mathbb{M})) \in \bar{E}(\mathbb{M}, 1)$, then for any $N \geqslant 0$:
    \begin{equation}
         \etaA(\mathbb{M}, N) = \etaNA(\mathbb{M}, N) = \eta(\mathbb{M})^N,
    \end{equation} \label{prop:no-adaptive-advantage}
\end{proposition}

\begin{proof}
    See Appendix \ref{sec:proof-no-adaptive-advantage}.
\end{proof}

\noindent While this condition looks asymmetric, we note that the achievable sets $\bar{E}(\mathbb{M}, N)$ are symmetric in swapping $\varepsilon_0$ and $\varepsilon_1$, since the effect of putting an $X$ gate at the start and a bit flip at the end of the circuit is to swap the errors. Indeed, the condition is equivalent to $\bar{E}(\mathbb{M}, 1)$ being a hexagon with vertices $\{ (0,1), (0,\eta), (\eta,0), (1,0), (1,1-\eta), (1-\eta,1) \}$, for $\eta = \eta(\mathbb{M})$, as shown in Figure \ref{fig:no-adaptive-advantage-achievable}.

A final consideration we make is how we may go about calculating $\etaA(\mathbb{M},N)$. In order to calculate the achievable error sets (or more specifically, the single-circuit achievable sets $E(\mathbb{M}, N)$), we express the recursive formula \eqref{eq:achievable-errors} as:
\begin{widetext}
    \begin{align}
        E(\mathbb{M},N) &= \left\{ \left( \sum_{x \in [m]} \varepsilon_0^{(x)} q_0^{(x)}, \sum_{x \in [m]} \varepsilon_1^{(x)} q_1^{(x)} \right) : \left( \varepsilon_0^{(x)}, \varepsilon_1^{(x)} \right) \in E(\mathbb{M}, N-1) \ \forall x, \ \mathbf{q}_i = \left( q_i^{(x)} \right)_{x \in [m]} \in \mathcal{R}(\mathbb{M}) \ \forall i \right\} \nonumber \\
        &= \left\{ \left( \mathrm{tr} \left( \rho_0 \sum_{x \in [m]} \varepsilon_0^{(x)} M_x \right), \mathrm{tr} \left( \rho_1 \sum_{x \in [m]} \varepsilon_1^{(x)} M_x \right) \right) : \left( \varepsilon_0^{(x)}, \varepsilon_1^{(x)} \right) \in E(\mathbb{M}, N-1) \ \forall x, \ \rho_i \in \mathcal{D}(\mathcal{H}) \ \forall i \right\} \nonumber \\
        &= \bigcup_{\{(\varepsilon_0^{(x)}, \varepsilon_1^{(x)})\} \subseteq E(\mathbb{M}, N-1)} \mathcal{I} \left( \bm\varepsilon_0 \cdot \mathbf{M} \right) \times \mathcal{I} \left( \bm\varepsilon_1 \cdot \mathbf{M} \right) \qquad \text{where } \bm\varepsilon_i \cdot \mathbf{M} = \sum_{x \in [m]} \varepsilon_i^{(x)} M_x. \label{eq:union-of-rectangles}
    \end{align}
\end{widetext}

\noindent Recall that $\mathcal{I}(\bm\varepsilon_i \cdot \mathbf{M})$ is the spectral range \eqref{eq:spectral-range} of $\bm\varepsilon_i \cdot \mathbf{M}$, which is an interval, so the Cartesian product $\mathcal{I} \left( \bm\varepsilon_0 \cdot \mathbf{M} \right) \times \mathcal{I} \left( \bm\varepsilon_1 \cdot \mathbf{M} \right)$ is a rectangle. Therefore \eqref{eq:union-of-rectangles} decomposes $E(\mathbb{M}, N)$ as a union of rectangles.

To find $(\varepsilon_0, \varepsilon_1)$ minimising total error, we know that at every step, we will always want to pick the lower left corner of each rectangle, so we define sets $e(\mathbb{M},N)$ using the same iterative process but only including those lower left corners:
\begin{align}
    e(\mathbb{M},0) &= \{(0,1), (1,0)\}; \nonumber \\
    e(\mathbb{M},N) &= \big\{ ( \lambda_\mathrm{min} \left( \bm\varepsilon_0 \cdot \mathbf{M} \right), \lambda_\mathrm{min} \left( \bm\varepsilon_1 \cdot \mathbf{M} \right) ) : \nonumber \\
    &\hspace{3em} \left( \varepsilon_0^{(x)}, \varepsilon_1^{(x)} \right) \in e(\mathbb{M}, N-1) \ \forall x \big\}. \label{eq:finite-achievable-errors}
\end{align}

\noindent These sets are finite, but exactly define the lower left boundary of $E(\mathbb{M},N)$, and by symmetry we can reflect across $\varepsilon_0 + \varepsilon_1 = 1$ to define the upper right boundary, so that in total, we encompass the entire set in the convex hull:

\begin{proposition}
    For any $\mathbb{M}$ and $N \geqslant 0$:
    \begin{equation}
        \mathrm{conv}(e(\mathbb{M}, N) \cup g(\mathbb{M}, N)) = \bar{E}(\mathbb{M}, N),\label{eq:finite-convex-hull}
    \end{equation}

    \noindent where $g(\mathbb{M}, N) = \left\{ (1-\varepsilon_1, 1-\varepsilon_0) : (\varepsilon_0, \varepsilon_1) \in e(\mathbb{M}, N) \right\} $ is the reflection of $e(\mathbb{M}, N)$ in the line $\varepsilon_0 + \varepsilon_1 = 1$. \label{prop:achievable-error-polygon}
\end{proposition}

\begin{proof}
    See Appendix \ref{sec:proof-achievable-error-polygon}.
\end{proof}

\noindent In particular, $\bar{E}(\mathbb{M}, N)$ is a convex set defined as the convex hull of a finite number of points in the plane, so it is a polygon. By the linearity of total error, the minimal total error throughout the entire set must be found at one of the vertices of the polygon:

\begin{corollary}
    For any $\mathbb{M}$ and $N \geqslant 0$:
    \begin{equation}
        \etaA(\mathbb{M}, N) = \min_{(\varepsilon_0, \varepsilon_1) \in v(\mathbb{M}, N)} \{ \varepsilon_0 + \varepsilon_1 \},
    \end{equation}
    
    \noindent where $v(\mathbb{M}, N) \subseteq e(\mathbb{M}, N)$ is the set of vertices of the polygon $\bar{E}(\mathbb{M}, N)$ that lie on or below the line $\varepsilon_0 + \varepsilon_1 = 1$. \label{cor:adaptive-total-error}
\end{corollary}

\noindent Overall, this gives us several methods to calculate $\etaA(\mathbb{M}, N)$. First, we may calculate it from its definition \eqref{eq:eta-a} via finding the single-circuit achievable error sets $E(\mathbb{M},N)$ or their convex hulls $\bar{E}(\mathbb{M},N)$. By Proposition \ref{prop:achievable-error-polygon}, we need only consider recursively defined finite sets with the same convex hull, which makes exhaustive approaches possible. Corollary \ref{cor:adaptive-total-error} reduces the search space even further to the vertices of the polygon $\bar{E}(\mathbb{M},N)$. Alternatively, Proposition \ref{prop:adaptive-lower-bound} gives us a lower bound on $\etaA(\mathbb{M}, N)$, and Proposition \ref{prop:no-adaptive-advantage} provides a sufficient condition that means it will be met non-adaptively, with both $\etaA(\mathbb{M}, N)$ and $\etaNA(\mathbb{M}, N)$ simply calculated.

\subsection{Further examples} \label{sec:qubit-examples}

\noindent We conclude this section by applying this formalism to some interesting example POVMs, to show when adaptive advantage can be found. First, we will return briefly to the imperfect $Z$ from Section \ref{sec:example}, and then we will look at other qubit and qudit\footnote{Recall that we can use a POVM acting on $\mathcal{H}$ to approximate a qubit $Z$ measurement even if $h = \dim \mathcal{H} \neq 2$.} POVMs.

\subsubsection{Imperfect $Z$} \label{sec:noisy-Z-further-example}

\noindent To start with, we have a loose end to tie up with the imperfect $Z$ measurement $\mathbb{Z}_{p,q}$ as defined in \eqref{eq:imperfect-Z}. In Section \ref{sec:example}, we saw how for most values of $p,q$, it is possible to find adaptive circuits that beat the best non-adaptive circuits, for $N = 3$ or $N = 4$. We may ask whether it is ever possible to get an adaptive advantage with imperfect $Z$ for $N = 2$.

First, we want to know the set $e(\mathbb{Z}_{p,q},1)$. As before we assume that $0 \leqslant p \leqslant q \leqslant 1-p$, with the rest of the parameter space following by symmetry. In particular, since $p + q \leqslant 1$, then:
\begin{equation}
    e(\mathbb{Z}_{p,q},1) = \{ (0,1), (p,q), (q,p), (1,0) \}.
\end{equation}

\begin{figure}[t]
    \centering
    \includegraphics[scale=1.25]{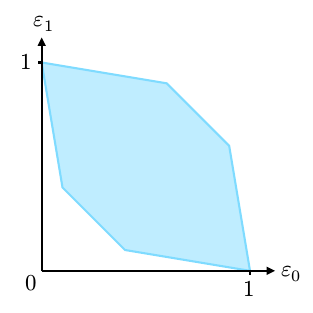}
    \caption{The achievable error set $\bar{E}(\mathbb{Z}_{p,q},1)$ for $p = 0.1$ and $q = 0.4$. In the general case, the shape is similar, being a hexagon with vertices $\{ (0,1), (p,q), (q,p), (1,0), (1-p,1-q), (1-q,1-p) \}$.}
    \label{fig:noisy-Z-achievable}
\end{figure}

\noindent The achievable set $\bar{E}(\mathbb{Z}_{p,q},1)$ is a hexagon with vertices at $\{ (0,1), (p,q), (q,p), (1,0), (1-p,1-q), (1-q,1-p) \}$, as illustrated in Figure \ref{fig:noisy-Z-achievable}. Thus via Corollary \ref{cor:adaptive-total-error}, the minimal total error achievable with one measurement of $\mathbb{Z}_{p,q}$ is:
\begin{equation}
    \eta(\mathbb{Z}_{p,q}) = p + q.
\end{equation}

\begin{table*}[t]
    \centering
    \includegraphics[scale=0.75]{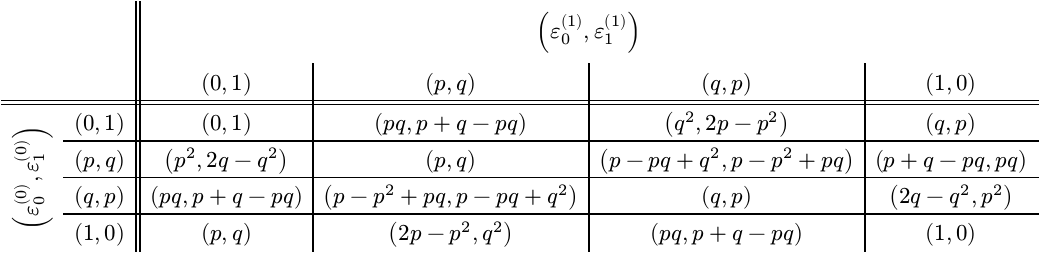}
    \caption{The errors $(\varepsilon_0, \varepsilon_1) = \left( \lambda_\mathrm{min} \left( \varepsilon_0^{(0)} M_0 + \varepsilon_0^{(1)} M_1 \right), \lambda_\mathrm{min}\left( \varepsilon_1^{(0)} M_0 + \varepsilon_1^{(1)} M_1 \right) \right)$ comprising $e(\mathbb{Z}_{p,q},2)$, for each choice of $\left( \varepsilon_0^{(0)}, \varepsilon_1^{(0)} \right), \left( \varepsilon_0^{(1)}, \varepsilon_1^{(1)} \right) \in e(\mathbb{Z}_{p,q}, 1)$. We assume as usual that $0 \leqslant p \leqslant q \leqslant 1-p$, in which case the smallest achievable total error is $\etaA(\mathbb{Z}_{p,q}, 2) = 2p - p^2 + q^2$.}
    \label{tab:two-noisy-Z-errors}
\end{table*}

\noindent By the recursive definition \eqref{eq:finite-achievable-errors}, each element of the set $e(\mathbb{Z}_{p,q},2)$ is defined by a combination of two elements of $e(\mathbb{Z}_{p,q},1)$. This is shown in Table \ref{tab:two-noisy-Z-errors}, with the list of possible total errors being:
\begin{align}
    &\hspace{-1em} \{ \varepsilon_0 + \varepsilon_1 : (\varepsilon_0, \varepsilon_1) \in e(\mathbb{Z}_{p,q},2) \} \nonumber \\
    &= \{ 1, p+q, 2p-p^2+q^2, 2q+p^2-q^2 \}
\end{align}

\noindent The smallest of these is $\etaA(\mathbb{Z}_{p,q}, 2)$, which for $0 \leqslant p \leqslant q \leqslant 1-p$ is always:
\begin{equation}
    \etaA(\mathbb{Z}_{p,q}, 2) = 2p - p^2 + q^2.
\end{equation}

\noindent However, this minimal total error is also achievable non-adaptively, by copying $\ket{i}$ with a CNOT gate, measuring both copies with $\mathbb{Z}_{p,q}$, and post-processing via $f =  \mathrm{AND}$. This circuit gives $\varepsilon_0 = q^2$ and $\varepsilon_1 = 2p - p^2$, so achieves the same total error of $2p - p^2 + q^2$. Therefore, for $N = 2$ there is no adaptive advantage, so we need to go to $N = 3$ at minimum.

One might at this point wonder what $\etaNA(\mathbb{Z}_{p,q},N)$ and $\etaA(\mathbb{Z}_{p,q},N)$ are for arbitrary $p,q,N$, whether there is an adaptive advantage, and if so, how much? Unfortunately, we do not currently have answers to these questions, although we do have asymptotic bounds that we discuss in Section \ref{sec:qudit-measurements}.

\subsubsection{The trine}

\noindent The imperfect $Z$ measurements encompass all possible two-outcome measurements on a qubit up to unitary equivalence, but the language of POVMs allows us to understand fundamentally different measurements as well, with more than two outcomes. One such example is the \emph{trine} measurement, defined in terms of the three qubit states $\ket{\psi_0} = \ket{0}$, $\ket{\psi_1} = \frac{1}{2} \ket{0} + \frac{\sqrt{3}}{2} \ket{1}$, and $\ket{\psi_2} = \frac{1}{2} \ket{0} - \frac{\sqrt{3}}{2} \ket{1}$, as:
\begin{equation}
    \trine = \left\{ M_i = \tfrac{2}{3} \ketbra{\psi_i}{\psi_i} : i \in \{ 0, 1, 2 \} \right\}. \label{eq:trine}
\end{equation}

\noindent Unlike the imperfect $Z$, we cannot understand this physically as noise on an otherwise projective $Z$ measurement, indeed being defined in terms of perfect rank-one projectors, it is extremal and can't be described as \emph{noisy} in a meaningful sense. Instead, it is probing more information about the state than just where it falls on the superposition spectrum between $\ket{0}$ and $\ket{1}$, although this comes at the cost of randomness in the outcome. Nonetheless, in our scenario, we suppose that it is the only measuring device we have access to, so we will have to use it as best we can to try to recreate a projective $Z$ measurement.

If we only get one use of the trine, then using the symmetry of swapping the POVM elements, we find that $e(\trine,1)$ has four elements, coming from the post-processing functions that map either zero, one, two, or three of the measurement outcomes to $j = 0$:
\begin{align}
    \left( \lambda_\mathrm{min}(0), \lambda_\mathrm{min}(M_0 + M_1 + M_2) \right) & = (0,1); \nonumber \\
    \left( \lambda_\mathrm{min}(M_0), \hspace{4ex} \lambda_\mathrm{min}(M_1 + M_2) \right) & = (0,\tfrac{1}{3}); \nonumber \\
    \left( \lambda_\mathrm{min}(M_0 + M_1), \hspace{4ex} \lambda_\mathrm{min}(M_2) \right) & = (\tfrac{1}{3},0); \nonumber \\
    \left( \lambda_\mathrm{min}(M_0 + M_1 + M_2), \lambda_\mathrm{min}(0) \right) & = (1,0),
\end{align}

\begin{figure}[t]
    \centering
    \includegraphics[scale=1.25]{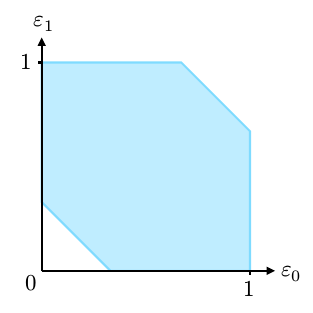}
    \caption{The achievable error set $\bar{E}(\trine,1)$. By comparison to Figure \ref{fig:no-adaptive-advantage-achievable}, we see that it takes the form where Proposition \ref{prop:no-adaptive-advantage} applies, and so the trine will not benefit from any adaptive advantage when used to approximate the projective $Z$ measurement.}
    \label{fig:trine-achievable}
\end{figure}

\noindent and so the best total error that we can achieve is $\eta(\trine) = \frac{1}{3}$ by Corollary \ref{cor:adaptive-total-error}. Moreover, as can be seen from the above or the calculated achievable error set $\bar{E}(\trine,1)$ shown in Figure \ref{fig:trine-achievable}, the conditions for Proposition \ref{prop:no-adaptive-advantage} are satisfied, so we can deduce that:
\begin{equation}
    \etaA(\trine, N) = \etaNA(\trine, N) = \tfrac{1}{3^N}. \label{eq:eta-trine}
\end{equation}

\begin{figure}[t]
    \centering
    \includegraphics[scale=1.25]{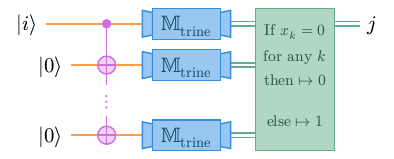}
    \caption{An optimal circuit for $N$ uses of the trine measurement $\trine$. If the input state is $\ket{1}$, the measurement outcome 0 is always impossible, so whenever we see a 0, we know that the input state was $\ket{0}$, hence we output $j = 0$. Otherwise, any string of outcomes consisting of just 1s and 2s is more likely to have come from an input of $\ket{1}$, so we output $j = 1$. This achieves total error $\frac{1}{3^N}$, which is minimal.}
    \label{fig:trine-optimal}
\end{figure}

\noindent Thus adaptivity will never give any advantage when we only have the trine measurement available. From the proof of Proposition \ref{prop:no-adaptive-advantage}, we can also find a circuit that achieves this optimal total error, shown in Figure \ref{fig:trine-optimal}.

\subsubsection{The qubit SIC-POVM}

\noindent One of the most theoretically interesting qubit measurements is the \emph{symmetric informationally complete POVM} (SIC-POVM):
\begin{equation}
    \SIC = \left\{ M_i = \tfrac{1}{2} \ketbra{\psi_i}{\psi_i} : i \in \{ 0, 1, 2, 3 \} \right\} \label{eq:tetra},
\end{equation}

\noindent where the states $\{ \ket{\psi_i} : i \in \{0,1,2,3\} \}$ form a regular tetrahedron in the Bloch sphere. We take:
\begin{align}
    \ket{\psi_0} &= \ket{0}; & \ket{\psi_1} &= \tfrac{1}{\sqrt{3}} \ket{0} + \tfrac{\sqrt{2}}{\sqrt{3}} \ket{1}; \nonumber \\
    \ket{\psi_2} &= \tfrac{1}{\sqrt{3}} \ket{0} + e^{\frac{2\pi i}{3}} \tfrac{\sqrt{2}}{\sqrt{3}} \ket{1}; & \ket{\psi_3} &= \tfrac{1}{\sqrt{3}} \ket{0} + e^{\frac{4\pi i}{3}} \tfrac{\sqrt{2}}{\sqrt{3}} \ket{1}.
\end{align}

\noindent Despite its interesting properties, for us it is another example we can use to approximate a projective $Z$ measurement, and indeed, we will find it behaves differently to the examples we have seen so far. Again, we start with $N = 1$ use of $\SIC$, with $e(\SIC, 1)$ given by the five points:
\begin{align}
    \left( \lambda_\mathrm{min}(0), \lambda_\mathrm{min}(M_0 + M_1 + M_2 + M_3) \right) & = (0,1); \nonumber \\
    \left( \lambda_\mathrm{min}(M_0), \hspace{4ex} \lambda_\mathrm{min}(M_1 + M_2 + M_3) \right) & = (0,\tfrac{1}{2}); \nonumber \\
    \left( \lambda_\mathrm{min}(M_0 + M_1), \hspace{4ex} \lambda_\mathrm{min}(M_2 + M_3) \right) & = (\tfrac{3 - \sqrt{3}}{6},\tfrac{3 - \sqrt{3}}{6}); \nonumber \\
    \left( \lambda_\mathrm{min}(M_0 + M_1 + M_2), \hspace{4ex} \lambda_\mathrm{min}(M_3) \right) & = (\tfrac{1}{2},0); \nonumber \\
    \left( \lambda_\mathrm{min}(M_0 + M_1 + M_2 + M_3), \lambda_\mathrm{min}(0) \right) & = (1,0).
\end{align}

\begin{figure}[t]
    \centering
    \includegraphics[scale=1.25]{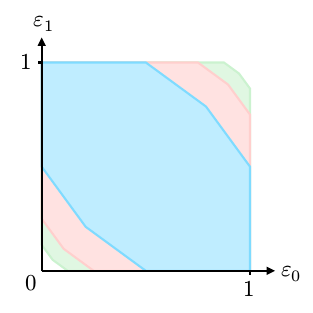}
    \caption{The achievable error sets $\bar{E}(\SIC,1)$ (blue), $\bar{E}(\SIC,2)$ (red and blue), and $\bar{E}(\SIC,3)$ (green, red, and blue). Like in Figure \ref{fig:no-adaptive-advantage-achievable}, the shape of the achievable error set remains similar as $N$ increases, with the lower left border shrinking towards the origin by a factor of $\frac{1}{2}$ each time. We can deduce by the recursive definition that this pattern continues indefinitely, which tells us the minimal adaptive total error for every $N$.}
    \label{fig:tetra-achievable}
\end{figure}

\noindent The achievable error set is as shown in blue in Figure \ref{fig:tetra-achievable}, with $\eta(\SIC) = 1 - \frac{1}{\sqrt{3}}$ by Corollary \ref{cor:adaptive-total-error}. We also note immediately that Proposition \ref{prop:no-adaptive-advantage} does not apply to the SIC-POVM, so we may have a chance at adaptive advantage like we saw for imperfect $Z$.

Now we suppose we have $N = 2$ uses of $\SIC$ to make. In the non-adaptive case, we use \eqref{eq:maximise-spectral-diameter}, checking all possible post-processing functions $f$, of which there are $2^{16} = 65\ 536$ as $\SIC^{\otimes 2}$ has 16 elements. For each $f$ we need to find the spectral diameter of $Q_0^{(f)}$, which is a $4 \times 4$ matrix, so this is simple to do numerically to any given precision. Computing this, we find that the minimal non-adaptive total error we can achieve is:
\begin{equation}
    \etaNA(\SIC, 2) = \eta(\SIC^{\otimes 2}) = \tfrac{1}{4}. \label{eq:two-non-adaptive-tetras}
\end{equation}

\begin{figure*}[t]
    \centering
    \includegraphics[scale=1.25]{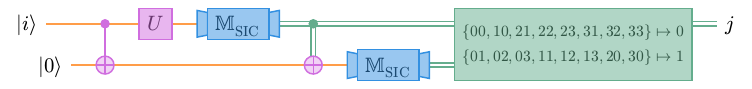}
    \caption{An adaptive circuit using the SIC-POVM $\SIC$ twice that achieves a lesser total error than any non-adaptive circuit with two uses of $\SIC$. The unitary $U$ is as defined in \eqref{eq:tetra-unitary}, while the classically controlled X gate after the first measurement is applied if the first measurement outcome is 0 or 1, while no operation is performed if the first measurement outcome is 2 or 3. The total error achieved by this circuit is $\frac{1}{2} - \frac{\sqrt{3}}{6} = 0.211...$, while the minimal non-adaptive total error is $\frac{1}{4} = 0.25$.}
    \label{fig:tetra-adaptive}
\end{figure*}

\noindent However, the adaptive circuit shown in Figure \ref{fig:tetra-adaptive} can achieve a smaller total error of $\frac{1}{2} - \frac{\sqrt{3}}{6}$, making use of the unitary:
\begin{align}
    U &= \tfrac{\sqrt{3 + \sqrt{3}}}{\sqrt{6}} \ketbra{0}{0} - \tfrac{\sqrt{3 - \sqrt{3}}}{\sqrt{6}} \ketbra{0}{1} \nonumber \\
    &\qquad + \tfrac{\sqrt{3 - \sqrt{3}}}{\sqrt{6}} \ketbra{1}{0} + \tfrac{\sqrt{3 + \sqrt{3}}}{\sqrt{6}} \ketbra{1}{1}. \label{eq:tetra-unitary}
\end{align}

\noindent Thus, the SIC-POVM can benefit from adaptivity with as few as two measurements, unlike the imperfect $Z$ which required at least three, or the trine, which never saw any adaptive advantage.

For greater $N$, we first look at the adaptive case. The achievable error set $\bar{E}(\SIC,2)$ is calculated from the previous achievable error set $\bar{E}(\SIC,1)$, and as shown in Figure \ref{fig:tetra-achievable}, it takes on a very similar shape. Indeed, the lower left boundary of $\bar{E}(\SIC,2)$ between the axes is exactly shrunk by a factor of $\frac{1}{2}$ towards the origin compared to $\bar{E}(\SIC,1)$. Since these sets are defined recursively, we deduce that this pattern must continue for all $N$, and hence that:
\begin{align}
    v(\SIC,N) &= \Big\{ (0,1), \left( 0, \tfrac{1}{2^N} \right), \left( \tfrac{1}{2^N}, 0 \right), (1,0), \nonumber \\
    &\hspace{2em} \left( \tfrac{1}{2^N} \left( 1 - \tfrac{1}{\sqrt{3}} \right) , \tfrac{1}{2^N} \left( 1 - \tfrac{1}{\sqrt{3}} \right) \right) \Big\},
\end{align}

\noindent Therefore, by Corollary \ref{cor:adaptive-total-error}:
\begin{equation}
    \etaA(\SIC,N) = \tfrac{1}{2^{N-1}} \left( 1 - \tfrac{1}{\sqrt{3}} \right). \label{eq:eta-tetra}
\end{equation}

\noindent Meanwhile, in the non-adaptive case, we will show by finding matching lower and upper bounds that $\etaNA(\SIC, N) = \frac{1}{2^N}$. First, for an upper bound, we demonstrate a circuit that achieves this total error for all $N$. This circuit is effectively the same as in Figure \ref{fig:trine-optimal}, with the post-processing function sending any string containing a 0 to $j = 0$, and any string not containing any 0s to $j = 1$, such that:
\begin{align}
    \varepsilon_0 &= \lambda_\mathrm{min} \left( \mathbb{I}_2^{\otimes N} - (M_1 + M_2 + M_3)^{\otimes N} \right) = 0; \nonumber \\
    \varepsilon_1 &= \lambda_\mathrm{min} \left( (M_1 + M_2 + M_3)^{\otimes N} \right) = \tfrac{1}{2^N}.
\end{align}

\noindent and so the total error of $\frac{1}{2^N}$ is achieved non-adaptively:
\begin{equation}
    \etaNA(\SIC,N) \leqslant \tfrac{1}{2^N}. \label{eq:tetra-non-adaptive-upper-bound}
\end{equation}

Now, from \eqref{eq:two-non-adaptive-tetras}, we have that $(0, \eta(\SIC^{\otimes 2})) = (0, \frac{1}{4}) \in E(\SIC^{\otimes 2}, 1)$, so Proposition \ref{prop:no-adaptive-advantage} applies to $\SIC^{\otimes 2}$, and thus for all $N' \geqslant 0$:
\begin{equation}
    \etaA(\SIC^{\otimes 2}, N') = \tfrac{1}{4^{N'}}.
\end{equation}

\noindent Allowing $N'$ uses of $\SIC^{\otimes 2}$ adaptively is a looser condition than allowing $2N'$ uses of $\SIC$ non-adaptively, so $\etaNA(\SIC, 2N') \geqslant \frac{1}{4^{N'}}$. Hence, if $N$ is even, we take $N' = \frac{1}{2} N$, and in combination with \eqref{eq:tetra-non-adaptive-upper-bound} deduce that:
\begin{equation}
    \etaNA(\SIC,N) = \tfrac{1}{2^N} \quad \text{for $N$ even}. \label{eq:tetra-non-adaptive-even}
\end{equation}

\begin{figure*}[t]
    \centering
    \includegraphics[scale=1.25]{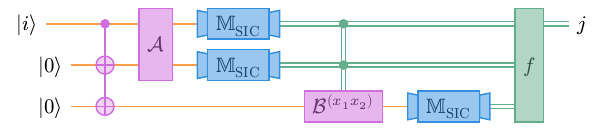}
    \caption{For three uses of $\SIC$, rather than finding an optimal non-adaptive circuit, we find it is easier to find an optimal circuit of the form shown, where the third measurement is performed adaptively. For each $x_1, x_2 \in [4]$, any choice of $\left( \varepsilon_0^{(x_1 x_2)}, \varepsilon_1^{(x_1 x_2)} \right) \in e(\SIC, 1)$ tells us what $\mathcal{B}^{(x_1 x_2)}$ should be, as well as the optimal post-processing $f$ on $\mathbf{x} = (x_1, x_2, x_3)$ for each $x_3 \in [4]$. Then, with such a choice made for each $x_1, x_2$, we have $(\varepsilon_0, \varepsilon_1) = \left( \lambda_\mathrm{min} \left( \sum_{x_1, x_2 \in [4]} \varepsilon_0^{(x_1 x_2)} M_{x_1} \otimes M_{x_2} \right), \lambda_\mathrm{min} \left( \sum_{x_1, x_2 \in [4]} \varepsilon_1^{(x_1 x_2)} M_{x_1} \otimes M_{x_2} \right) \right)$.}
    \label{fig:tetra-three-measurements}
\end{figure*}

\noindent For $N = 3$, we would like to be able to calculate $\etaNA(\SIC, 3)$ directly, but with $2^{4^3} = 2^{64}$ possible post-processing functions $f$, this is infeasible by brute force. Instead, we consider an adaptive circuit with a measurement of $\SIC^{\otimes 2}$ first, followed by a final measurement of $\SIC$ adaptively, as shown in Figure \ref{fig:tetra-three-measurements}. While this circuit looks a little different to the usual fully adaptive circuits, the same method applies. For each pair of possible results for the first two measurements $(x_1, x_2) \in [4]^2$, we need to make a choice of $\mathcal{B}^{(x_1,x_2)}$, corresponding along with maximum likelihood post-processing to some $\left( \varepsilon_0^{(x_1 x_2)}, \varepsilon_1^{(x_1 x_2)} \right) \in \bar{E}(\SIC, 1)$. Then, any such choice for each $(x_1, x_2)$ makes the optimal choice of $\mathcal{A}$ to map $\ketbra{0}{0}$ and $\ketbra{1}{1}$ to the minimal eigenstates of $\sum_{x_1, x_2 \in [4]} \varepsilon_0^{(x_1 x_2)} M_{x_1} \otimes M_{x_2}$ and $\sum_{x_1, x_2 \in [4]} \varepsilon_1^{(x_1 x_2)} M_{x_1} \otimes M_{x_2}$ respectively, such that:
\begin{align}
    \varepsilon_0 &= \lambda_\mathrm{min} \left( \sum_{x_1, x_2 \in [4]} \varepsilon_0^{(x_1 x_2)} M_{x_1} \otimes M_{x_2} \right); \nonumber \\    
    \varepsilon_1 &= \lambda_\mathrm{min} \left( \sum_{x_1, x_2 \in [4]} \varepsilon_1^{(x_1 x_2)} M_{x_1} \otimes M_{x_2} \right).
\end{align}

\noindent By Corollary \ref{cor:adaptive-total-error}, we know that $\left( \varepsilon_0^{(x_1 x_2)}, \varepsilon_1^{(x_1 x_2)} \right) \in v(\SIC, 1)$ will in fact suffice, and indeed, we note that two elements $(0,1)$ and $(1,0)$ would never be chosen over $(0,\frac{1}{2})$ and $(\frac{1}{2},0)$. Therefore, we need only consider:
\begin{align}
    \left( \varepsilon_0^{(x_1 x_2)}, \varepsilon_1^{(x_1 x_2)} \right) &\in \Big\{ \left( 0, \tfrac{1}{2} \right), \left( \tfrac{1}{2}, 0 \right), \nonumber \\
    &\hspace{-2em} \left( \tfrac{1}{2} \left( 1 - \tfrac{1}{\sqrt{3}} \right) , \tfrac{1}{2} \left( 1 - \tfrac{1}{\sqrt{3}} \right) \right) \Big\}.
\end{align}

\noindent We thus have three options for what to do for each $(x_1, x_2) \in [4]^2$, meaning that we have $3^{4^2} = 3^{16}$ possibilities to check in total, or about 43 million. This is a vast improvement over $2^{64}$, and can be done by exhaustive search, with the minimal total error achievable among all such circuits being $\eta = \frac{1}{8}$.

Non-adaptive three measurement circuits are a more restrictive category than those of the form of Figure \ref{fig:tetra-three-measurements}, so $\etaNA(\SIC, 3) \geqslant \frac{1}{8}$. In combination with \eqref{eq:tetra-non-adaptive-upper-bound}, we thus obtain:
\begin{equation}
    \etaNA(\SIC, 3) = \tfrac{1}{8} = \tfrac{1}{2^3}.
\end{equation}

\noindent Then, for any odd $N$ greater than or equal to 5, we consider doing an adaptive circuit consisting of one measurement of $\SIC^{\otimes (N - 3)}$, followed by one of $\SIC^{\otimes 3}$ adaptively. By a similar method to Proposition \ref{prop:adaptive-lower-bound}, the eventual total error must be lower bounded by the product of the individual total errors, and so the minimal non-adaptive total error is also lower bounded by this product:
\begin{equation}
    \etaNA(\SIC, N) \geqslant \etaNA(\SIC, N - 3) \ \etaNA(\SIC, 3). \label{eq:tetra-non-adaptive-3}
\end{equation}

\noindent Since $N - 3$ is even, we combine \eqref{eq:tetra-non-adaptive-even} and \eqref{eq:tetra-non-adaptive-3} to deduce that:
\begin{equation}
    \etaNA(\SIC, N) \geqslant \tfrac{1}{2^{N-3}} \tfrac{1}{8} = \tfrac{1}{2^N}.
\end{equation}

\noindent This gives the final lower bound to match the upper bound in \eqref{eq:tetra-non-adaptive-upper-bound}, so for all $N \geqslant 2$:
\begin{equation}
    \etaNA(\SIC, N) = \tfrac{1}{2^N}.
\end{equation}

\noindent This shows that there is an adaptive advantage for any number of uses of the SIC-POVM greater than one, with the relative advantage \eqref{eq:relative-adaptive-advantage} being $R = \frac{1}{4} (3 + \sqrt{3}) \approx 1.18$. We note here the stark contrast to the trine; despite both POVMs being similarly rank one and symmetric, the trine never sees any adaptive advantage, while the SIC-POVM always does.

\subsubsection{The generalised symmetric imperfect $Z$}

\noindent All of the examples we have been using so far are qubit measurements, but the mathematics doesn't limit us to just these. In fact, we can use a POVM in any dimension (referred to as $h$ earlier), even if the measurement we are attempting to approximate remains the projective $Z$ measurement on a qubit. We will still suppose that the input is a qubit, specifically either $\ket{0}$ or $\ket{1}$, but it will need to be processed to a higher-dimensional Hilbert space before measuring.

As such an example, we define a POVM we call the \emph{generalised symmetric imperfect $Z$}, where for any $m \geqslant 2$ and $t \in [0,1]$, the POVM is defined on a system in $h = m$ dimensions by the $m$ elements:
\begin{equation}
    \mathbb{S}_{m,t} = \left\{ S_x = (1-t) \ketbra{x}{x} + \tfrac{t}{m-1} \sum_{y \neq x} \ketbra{y}{y} \right\}_{x \in [m]}. \label{eq:generalised-symmetric-imperfect-Z}
\end{equation}

\noindent In words, given an input $\ket{\psi} = \ket{i}$, this POVM identifies it correctly with probability $1-t$, but with probability $t$ it makes an error and outputs one of the other $m - 1$ results with equal probability. For $m = 2$, this is exactly the symmetric ($p = q = t$) case of the imperfect $Z$.

We notice that $\mathbb{S}_{m,t}$ is an abelian POVM, diagonal in the computational basis, so the most general non-adaptive circuit can be written in the form of Figure \ref{fig:non-adaptive-abelian}. Moreover, we can take the $\mathcal{A}_k$ to be isometries, mapping $\ketbra{0}{0}$ and $\ketbra{1}{1}$ to distinct computational basis states $\ketbra{i_0}{i_0}$ and $\ketbra{i_1}{i_1}$. By the symmetry of $\mathbb{S}_{m,t}$, we can take $\ketbra{i_0}{i_0} = \ketbra{0}{0}$, $\ketbra{i_1}{i_1} = \ketbra{1}{1}$ without loss of generality, since any other such isometry is equivalent up to a permutation before and after the measurement $\mathbb{S}_{m,t}$. This tells us exactly the pre-processing circuit we should use in the non-adaptive case, and hence by combining with maximum likelihood decoding, we get the optimal non-adaptive circuit.

In the adaptive case, a similar argument applies. As we did with $\mathbb{Z}_{p,q}$ in Section \ref{sec:example}, we can argue that all pre-processing should keep states in the diagonalising basis of the abelian POVM (in this case, the computational basis), and by symmetry of the POVM elements, all the pre-processing before measurements can be taken to be $\ketbra{0}{0} \mapsto \ketbra{0}{0}$, $\ketbra{1}{1} \mapsto \ketbra{1}{1}$ as in the non-adaptive case. Thus the optimal adaptive circuit reduces to exactly the optimal non-adaptive circuit, so adaptivity will not help for $\mathbb{S}_{m,t}$ (just as it didn't for $\mathbb{Z}_{p,q}$ when $p = q$).

We are left to find what the maximum likelihood decoding of any string of measurement results is. The probabilities of obtaining each string $\mathbf{x} \in [m]^N$ for each possible input to the circuit are:
\begin{align}
    &\mathbb{P} \left( \mathbf{x} \mid \ket{0} \right) = \prod_{k=1}^N \bra{0} S_{x_k} \ket{0} \nonumber \\
    &\hspace{2ex} = \left( 1 - t \right)^{\#\{x_k = 0\}} \left( \tfrac{t}{m-1} \right)^{\#\{x_k = 1\}} \left( \tfrac{t}{m-1} \right)^{\#\{x_k > 1\}}; \nonumber \\
    &\mathbb{P} \left( \mathbf{x} \mid \ket{1} \right) = \prod_{k=1}^N \bra{1} S_{x_k} \ket{1} \nonumber \\
    &\hspace{2ex} = \left( \tfrac{t}{m-1} \right)^{\#\{x_k = 0\}} \left( 1 - t \right)^{\#\{x_k = 1\}} \left( \tfrac{t}{m-1} \right)^{\#\{x_k > 1\}}.
\end{align}

\noindent Maximum likelihood decoding then falls into two cases. If $t \leqslant \frac{m-1}{m}$ (i.e. not making an error is more likely than any individual error), then $1-t \geqslant \frac{t}{m-1}$, so we decode based on which of 0 and 1 come up more often:
\begin{equation}
    \mathbf{x} \mapsto \begin{cases} 0 & \text{if } \#\{x_k = 1\} \leqslant \#\{x_k = 0\}; \\ 1 & \text{if } \#\{x_k = 1\} > \#\{x_k = 0\} ,\end{cases}
\end{equation}

\noindent and the error rates are:
\begin{align}
    \varepsilon_0 &= \mathbb{P} \left( \#\{x_k = 1\} > \#\{x_k = 0\} | \ket{0} \right); \nonumber \\
    \varepsilon_1 &= \mathbb{P} \left( \#\{x_k = 1\} \leqslant \#\{x_k = 0\} | \ket{1} \right). \label{eq:gen-sym-errors-1}
\end{align}

\noindent If however $t > \frac{m-1}{m}$ (i.e. all possible errors are each more likely than no error), we must instead choose whichever of 0 and 1 come up less often:
\begin{equation}
    \mathbf{x} \mapsto \begin{cases} 0 & \text{if } \#\{x_k = 1\} \geqslant \#\{x_k = 0\}; \\ 1 & \text{if } \#\{x_k = 1\} < \#\{x_k = 0\}. \end{cases}
\end{equation}

\noindent and the error rates are instead:
\begin{align}
    \varepsilon_0 &= \mathbb{P} \left( \#\{x_k = 1\} < \#\{x_k = 0\} | \ket{0} \right); \nonumber \\
    \varepsilon_1 &= \mathbb{P} \left( \#\{x_k = 1\} \geqslant \#\{x_k = 0\} | \ket{1} \right). \label{eq:gen-sym-errors-2}
\end{align}

\noindent In both cases, it doesn't matter what we choose when the number of 0s and 1s are equal, so we choose to process to 0 arbitrarily.

Since we know what the optimal circuit is, calculating the minimal total error for $\mathbb{S}_{m,t}$ requires only calculating the errors $\varepsilon_0$ and $\varepsilon_1$ for this circuit. While we won't be able to do this exactly, we can find an upper bound. 

First, we need a result about the binomial distribution, which we derive from the Chernoff bound found by Hoeffding \cite{hoeffding-1963}:

\begin{lemma}
    Let $p \leqslant \frac{1}{2}$. Then:
    \begin{equation}
        \mathbb{P} \left( \mathrm{Bin}(n,p) \geqslant \tfrac{1}{2} n \right) \leqslant \left( 2 \sqrt{p(1-p)} \right)^n,
    \end{equation}
    
    \noindent where $\mathrm{Bin}(n,p)$ is the binomial distribution with $n$ trials and probability of success $p$. \label{lem:binomial-deviation}
\end{lemma}

\begin{proof}
    See Appendix \ref{sec:proof-binomial-deviation}.
\end{proof}

\noindent We can turn this into a result about particular multinomial distributions, mirroring the distribution we get from measuring generalised symmetric imperfect $Z$:

\begin{lemma}
    Let $X$ be a random variable with distribution:
    \begin{equation}
        X = \begin{cases} 0 & \text{with probability } 1-t; \\ x & \text{with probability } \frac{t}{m-1} \quad \text{ for all } x \in [m] \setminus \{0\}. \end{cases}
    \end{equation}

    \noindent Let $\mathbf{X}_N = (X_k)_{k = 1}^N$ be a vector of iid random variables with this same distribution. Then, for any $j \neq 0$:
    \begin{itemize}
        \item[(a)] If $t \leqslant \frac{m-1}{m}$, (i.e. $0$ more likely than the rest) then:
        \begin{equation}
            \mathbb{P} \left( \#\{X_k = j\} \geqslant \#\{X_k = 0\} \right) \leqslant \left( \tfrac{(m-2)t}{m-1} + 2 \sqrt{\tfrac{t(1-t)}{m-1}} \right)^N.
        \end{equation}
        \item[(b)] If $t \geqslant \frac{m-1}{m}$, (i.e. $0$ less likely than the rest) then:
        \begin{equation}
            \mathbb{P} \left( \#\{X_k = j\} \leqslant \#\{X_k = 0\} \right) \leqslant \left( \tfrac{(m-2)t}{m-1} + 2 \sqrt{\tfrac{t(1-t)}{m-1}} \right)^N.
        \end{equation}
    \end{itemize} \label{lem:multivariate-deviation}
\end{lemma}

\begin{proof}
    See Appendix \ref{sec:proof-multivariate-deviation}.
\end{proof}

\noindent In the case that $t \leqslant \frac{m-1}{m}$, we refer to \eqref{eq:gen-sym-errors-1}. For $\varepsilon_0$, we consider an input of $\ket{0}$ to our circuit, which gives each measurement result exactly the distribution given in Lemma \ref{lem:multivariate-deviation}(a). Hence, applying the lemma with $j = 1$ gives:
\begin{align}
    \varepsilon_0 &\leqslant \mathbb{P} \left( \#\{x_k = 1\} \geqslant \#\{x_k = 0\} | \ket{0} \right) \nonumber \\
    &\leqslant \left( \tfrac{(m-2)t}{m-1} + 2 \sqrt{\tfrac{t(1-t)}{m-1}} \right)^N.
\end{align}

\noindent For $\varepsilon_1$, the input state is now $\ket{1}$, but we can still apply Lemma \ref{lem:multivariate-deviation}(a) by swapping the labelling of 0 and 1 in the distribution given there. This gives us immediately:
\begin{align}
    \varepsilon_1 &= \mathbb{P} \left( \#\{x_k = 1\} \leqslant \#\{x_k = 0\} | \ket{1} \right) \nonumber \\
    &\leqslant \left( \tfrac{(m-2)t}{m-1} + 2 \sqrt{\tfrac{t(1-t)}{m-1}} \right)^N.
\end{align}

\noindent When $t > \frac{m-1}{m}$, we now refer to \eqref{eq:gen-sym-errors-2} instead. Compared to the previous case, the inequalities between $\#\{x_k = 0\}$ and $\#\{x_k = 1\}$ have swapped direction, but this swap is also present in Lemma \ref{lem:multivariate-deviation}(b), so we may apply it similarly:
\begin{align}
    \varepsilon_0 &\leqslant \mathbb{P} \left( \#\{x_k = 1\} \leqslant \#\{x_k = 0\} | \ket{0} \right) \nonumber \\
    &\leqslant \left( \tfrac{(m-2)t}{m-1} + 2 \sqrt{\tfrac{t(1-t)}{m-1}} \right)^N; \nonumber \\
    \varepsilon_1 &= \mathbb{P} \left( \#\{x_k = 1\} \geqslant \#\{x_k = 0\} | \ket{1} \right) \nonumber \\
    &\leqslant \left( \tfrac{(m-2)t}{m-1} + 2 \sqrt{\tfrac{t(1-t)}{m-1}} \right)^N.
\end{align}

\noindent Therefore, regardless of the value of $t$:
\begin{equation}
    \etaNA(\mathbb{S}_{m,t}, N) \leqslant 2 \left( \tfrac{(m-2)t}{m-1} + 2 \sqrt{\tfrac{t(1-t)}{m-1}} \right)^N,
\end{equation}

\noindent where, as a reminder, a bound on $\etaNA(\mathbb{S}_{m,t}, N)$ is the same as a bound on $\etaA(\mathbb{S}_{m,t}, N)$, since the generalised symmetric imperfect $Z$ never has an adaptive advantage.

\section{Qudit measurements} \label{sec:qudit-measurements}

\noindent So far, we have been trying to emulate a projective $Z$ measurement on a qubit. A natural question to ask is how our results generalise if we change the measurement we want to approximate. One avenue to take this down would be to consider an arbitrary target POVM, for instance trying to approximate the trine measurement instead. However, there are two reasons we don't consider this. Firstly, our figure of merit as it stands doesn't make much sense for target measurements that don't have definite outcomes, so we would have to replace it with one that compares the output distributions as a whole as opposed to single-shot error probabilities. Also, as shown in \cite{linden-2025}, any projective measurement can be used in a subroutine to exactly implement any POVM with the same number of outcomes, assuming arbitrary unitary operations can be implemented perfectly and arbitrary ancillas can be adjoined, as we have done throughout. Therefore, if we know how best to approximate a projective measurement on a $d$-dimensional system (i.e. a qudit), for $d \geqslant 2$, then this will give us an idea of how well we can approximate any POVM with $d$ outcomes\footnote{Of course, this isn't perfect. For instance, a trine cannot perfectly implement a projective measurement on a qutrit, so this would suggest it cannot perfectly implement any three outcome POVM. However, we know it can perfectly implement itself, without using the projective measurement subroutine. Therefore we won't be able to make any claims about optimality when approximating a non-projective measurement, but the upper bounds on the total error remain relevant.}.

In this section, we begin in Section \ref{sec:qudit-setup} by seeing how much of the work we did in the qubit case will be transferrable to the qudit case. We will see that the optimisation, while still finite, becomes much more demanding, so instead of focusing on finding the minimal total error exactly, we will bound it as well as possible. We will obtain a lower bound in the general case in Corollary \ref{cor:adaptive-lower-bound-d}, and while we don't find a general upper bound, in Section \ref{sec:qudit-examples} we will prove asymptotic results in Propositions \ref{prop:asymptotic-upper-bound} and \ref{prop:generalised-asymptotic-upper-bound}, through the consideration of our example measurements as we had before.

As with the previous section, proofs of stated results can be found in Appendix \ref{sec:proofs}.

\subsection{Moving from qubits to qudits} \label{sec:qudit-setup}

\noindent Our task is now to approximate the generalised projective $Z$ measurement on a qudit, i.e. the projective measurement onto the computational basis $\{ \ket{0}, \ket{1}, ..., \ket{d-1} \}$, for any given finite dimension $d \geqslant 2$. We will investigate what results from Section \ref{sec:qubit-measurements} extend immediately to higher dimensions, and what results no longer hold. We emphasise that all that is being changed is the input dimension to the circuit, from $2$ to $d$, and that the dimension of the system that the POVM measures, which we have usually called $h$, has no relation to this and needn't change. 

Our input is now $\ket{\psi} \in \mathbb{C}^d$, and our output is now $j \in [d]$. However, as we will define our figure of merit in terms of only the computational basis state inputs $\ket{\psi} = \ket{i} \in \{ \ket{0}, \ket{1}, ..., \ket{d-1} \}$, we will usually assume that $\ket{\psi}$ is one of these. Instead of just two error probabilities $\varepsilon_0$ and $\varepsilon_1$, we now have $d(d-1)$, which we collect along with the success probabilities into a $d \times d$ column stochastic matrix:
\begin{equation}
    \mathcal{E} = (\varepsilon_{j|i})_{j,i \in [d]} \quad \text{for } \varepsilon_{j|i} = \mathbb{P}(\text{output } j \ | \ \text{input } \ket{i}).
\end{equation}

\noindent Then, the total error is:
\begin{equation}
    \eta = \sum_{j \neq i} \varepsilon_{j|i} = d - \mathrm{tr}(\mathcal{E}). \label{eq:total-error-d}
\end{equation}

\noindent Note that this definition of total error is consistent with the previous one \eqref{eq:errors-2}, \eqref{eq:total-error-2} when $d = 2$. The error rates correspond as $\varepsilon_0 = \varepsilon_{1|0}$ and $\varepsilon_1 = \varepsilon_{0|1}$, with:
\begin{equation}
    \mathcal{E} = \begin{pmatrix} 1-\varepsilon_0 & \varepsilon_1 \\ \varepsilon_0 & 1-\varepsilon_1 \end{pmatrix}, \quad \eta = 2 - \mathrm{tr}(\mathcal{E}) = \varepsilon_0 + \varepsilon_1.
\end{equation}

\noindent Non-adaptively, Figure \ref{fig:non-adaptive} still depicts the most general circuit with $N$ uses of $\mathbb{M}$. If we let $\rho_i = \mathcal{A}(\ketbra{i}{i})$, then the errors are given by:
\begin{equation}
    \varepsilon_{j|i} = \mathrm{tr} \left( \rho_i Q_j^{(f)} \right),
\end{equation}

\noindent for effective POVM elements:
\begin{equation}
    Q_j^{(f)} = \sum_{\mathbf{x} : f(\mathbf{x}) = j} M_\mathbf{x}.
\end{equation}

\noindent where as before, $M_\mathbf{x} = M_{x_1} \otimes M_{x_2} \otimes ... \otimes M_{x_N}$, and $f: [m]^N \to [d]$ is the classical post-processing function.

We define the minimal non-adaptive total error as $\etaNA_d(\mathbb{M}, N)$, where the subscript $d$ denotes the dimension of the projective measurement we want to approximate. We can either fix $f$ to get the equivalent of \eqref{eq:maximise-spectral-diameter}:
\begin{equation}
    \etaNA_d(\mathbb{M}, N) = d - \max_{f} \left\{ \sum_{i \in [d]} \lambda_\mathrm{max} \left( Q_i^{(f)} \right) \right\}, \label{eq:total-error-postprocessing-d}
\end{equation}

\noindent or we can fix the measured states $\{ \rho_i \}_{i \in [d]}$, and obtain the equivalent of \eqref{eq:total-error-preprocessing}:
\begin{equation}
    \etaNA_d(\mathbb{M}, N) = d - \max_{\{ \mathbf{q}_i \} \subseteq \mathcal{R}(\mathbf{M}^{\otimes N})} \left\{ \sum_{\mathbf{x} \in [m]^N} \max_{i \in [d]} \left\{ q_i^{(\mathbf{x})} \right\} \right\}, \label{eq:total-error-preprocessing-d}
\end{equation}

\noindent where for each $i \in [d]$, $\mathbf{q}_i$ is the output distribution from measuring $\rho_i$ with $\mathbb{M}^{\otimes N}$ given by:
\begin{equation}
    \mathbf{q}_i = \left( q_i^{(\mathbf{x})} \right)_{\mathbf{x} \in [m]^N} = \left( \mathrm{tr} \left( \rho_i M_\mathbf{x} \right) \right)_{\mathbf{x} \in [m]^N}.
\end{equation}

\noindent Furthermore, if $\mathbb{M}$ is abelian, we can again simplify the circuit to the form of Figure \ref{fig:non-adaptive-abelian}, and simplify the optimisation \eqref{eq:total-error-preprocessing-d} to:
\begin{equation}
    \etaNA_d(\mathbb{M}, N) = d - \max_{\{ \mathbf{a}_i \}_{i \in [d]} \subseteq [h]^N} \left\{ \sum_{\mathbf{x} \in [m]^N} \max_{i \in [d]} \left\{ q_{\mathbf{a}_i}^{(\mathbf{x})} \right\} \right\}.
\end{equation}

\noindent where $h$ is the dimension of the space that $\mathbb{M}$ acts on, all POVM elements of $\mathbb{M}$ are diagonal in the basis $\{ \ket{v_a} \}_ {a \in [h]}$, and:
\begin{equation}
    q_{\mathbf{a}_i}^{(\mathbf{x})} = \prod_{k = 1}^N \mathrm{tr} \left( \mathcal{A}_k (\ketbra{i}{i}) M_{x_k} \right) =  \prod_{k = 1}^N \bra{v_{a_{ik}}} M_{x_k} \ket{v_{a_{ik}}}.
\end{equation}

\noindent However, we note that it is no longer possible to assume that the channels $\mathcal{A}_k$ in Figure \ref{fig:non-adaptive-abelian} are isometries; indeed it is now possible to have $d > h$ in which case this is impossible.

As for the adaptive case, the general circuit of Figure \ref{fig:adaptive} can still be simplified by the process shown in Figure \ref{fig:adaptive-recursive}(a),(b) to the form of Figure \ref{fig:adaptive-recursive}(c). The error matrices $\mathcal{E}^{(\mathbf{x})} = \left( \varepsilon_{j|i}^{(\mathbf{x})} \right)$ obey the recursive relation:
\begin{equation}
    \varepsilon_{j|i}^{(\mathbf{x})} = \sum_{x \in [m]} \mathrm{tr} \left( \mathcal{B}^{(\mathbf{x})} ( \ketbra{i}{i} ) M_x \right) \, \varepsilon_{j|i}^{(\mathbf{x},x)},
\end{equation}

\noindent and so we have the recursively defined single-circuit achievable error sets:
\begin{align}
    E_d(\mathbb{M}, 0) &= \{ \mathbf{e}_j \mathbf{1}^T : j \in [d] \}; \\ 
    E_d(\mathbb{M}, N) &= \Bigg\{ \left( \sum_{x \in [m]} \varepsilon_{j|i}^{(x)} q_i^{(x)} \right)_{j,i \in [d]} : \nonumber \\[2ex]
    &\hspace{2em} \left( \varepsilon_{j|i}^{(x)} \right)_{j,i \in [d]} \in E_d(\mathbb{M}, N-1) \ \forall x \in [m], \nonumber \\
    &\hspace{4em} \mathbf{q}_i = \left( q_i^{(x)} \right)_{x \in [m]} \in \mathcal{R}(\mathbb{M}) \, \forall i \in [d] \Bigg\},
\end{align}

\noindent with the achievable error sets being their convex hulls:
\begin{align}
    \bar{E}_d(\mathbb{M},N) = \mathrm{conv}(E_d(\mathbb{M},N)).
\end{align}

\noindent Proposition \ref{prop:achievable-set-is-achievable} remains true, where $\mathcal{E} \in \bar{E}_d(\mathbb{M},N)$ if and only if a circuit exists with $N$ adaptive uses of $\mathbb{M}$ that gives error matrix $\mathcal{E}$. The proof remains essentially the same; any achievable error matrix $\mathcal{E}$ must lie in $\bar{E}_d(\mathbb{M},N)$ by reducing the circuit to the form of Figure \ref{fig:adaptive-recursive} and using induction, while the recursive construction of the sets gives a recipe to find a circuit that achieves any error matrix $\mathcal{E} \in \bar{E}_d(\mathbb{M},N)$. Therefore we define as before:
\begin{equation}
    \etaA_d(\mathbb{M},N) = d - \max_{\mathcal{E} \in E_d(\mathbb{M},N)} \left\{ \mathrm{tr}(\mathcal{E}) \right\}.
\end{equation}

\noindent Letting $\eta_d(\mathbb{M}) = \etaNA_d(\mathbb{M}, 1) = \etaA_d(\mathbb{M}, 1)$, we again have that adaptive cannot do worse than non-adaptive, and that non-adaptive is equivalent to a single measurement of $\mathbb{M}^{\otimes N}$:
\begin{equation}
    \etaA_d(\mathbb{M}, N) \leqslant \etaNA_d(\mathbb{M}, N) = \eta_d(\mathbb{M}^{\otimes N}).
\end{equation}

\noindent Proposition \ref{prop:adaptive-lower-bound} no longer applies, and in fact it is now possible to have $\eta_d(\mathbb{M}) \neq 0$ but $\etaA_d(\mathbb{M},N) = 0$ for some $N \geqslant 2$, so no such inequality can be found\footnote{For instance, suppose we have a projective measurement on a qubit (i.e. $\mathbb{Z}_{0,0}$ up to unitary equivalence), and we want to use it to measure a $d = 4$ dimensional system. This cannot be done perfectly with one use of $\mathbb{Z}_{0,0}$, because we have only two possible measurement outcomes to distinguish four states, but is trivial to do with two uses, by processing $\ket{0} \mapsto \ket{00}$, $\ket{1} \mapsto \ket{01}$, $\ket{2} \mapsto \ket{10}$, $\ket{3} \mapsto \ket{11}$. Therefore, $\eta_4(\mathbb{Z}_{0,0}) > 0$, but $\etaA_4(\mathbb{Z}_{0,0}, 2) = \etaNA_4(\mathbb{Z}_{0,0}, 2) = 0$.}. However, we can still lower bound $\etaA_d(\mathbb{M},N)$, as we'll see shortly.

We can however find an equivalent of Corollary \ref{cor:achievable-sets-as-imperfect-Zs} --- the correspondence between the set $\bar{E}(\mathbb{M},N)$ and the imperfect $Z$ measurements that can be reproduced by $N$ adaptive uses of $\mathbb{M}$. For a $d \times d$ column stochastic matrix $\mathcal{E}$, we define the abelian qudit POVM $\mathbb{M}_\mathcal{E}$ by elements:
\begin{equation}
    \mathbb{M}_\mathcal{E} = \left\{ M_x = \sum_{i \in [d]} \varepsilon_{x|i} \ketbra{i}{i} \right\}_{x \in [d]}. \label{eq:abelian-povm-from-matrix}
\end{equation}

\noindent Then for all $\mathcal{E} \in \bar{E}(\mathbb{M},N)$, one measurement of $\mathbb{M}_\mathcal{E}$ can be reproduced by $N$ adaptive measurements of $\mathbb{M}$, by the same method as Lemma \ref{lem:reproduction-of-noisy-Z} and Corollary \ref{cor:achievable-sets-as-imperfect-Zs}. In particular, we apply this to the generalised symmetric imperfect $Z$ measurement $\mathbb{S}_{d,\frac{1}{d} \eta_d^{(A)} (\mathbb{M},N)}$ (defined in general in \eqref{eq:generalised-symmetric-imperfect-Z}):

\begin{lemma}
    For any POVM $\mathbb{M}$ and $N \geqslant 0$, define the matrix $\mathcal{E} = (\varepsilon_{j|i})_{j,i \in [d]}$ by:
    \begin{equation}
        \varepsilon_{j|i} = \begin{cases} 1 - \tfrac{1}{d} \eta_d^{(A)} (\mathbb{M},N) & \text{if } j = i; \\ \tfrac{1}{d(d-1)} \eta_d^{(A)} (\mathbb{M},N) & \text{if } j \neq i. \end{cases}
    \end{equation}
    
    \noindent Then, $\mathcal{E} \in \bar{E}_d(\mathbb{M}, N)$, and so $N$ adaptive uses of $\mathbb{M}$ can reproduce one use of the generalised symmetric imperfect $Z$ measurement $\mathbb{S}_{d,\frac{1}{d} \eta_d^{(A)} (\mathbb{M},N)} = \mathbb{M}_\mathcal{E}$. \label{lem:symmetric-reproduced-measurement}
\end{lemma}

\begin{proof}
    See Appendix \ref{sec:proof-symmetric-reproduced-measurement}.
\end{proof}

\noindent We can find a lower bound on the minimal total error that a single use of $\mathbb{S}_{m,t}$ can achieve for any $m \geqslant d$ as follows:

\begin{lemma}
    Let $m \geqslant d$ and $t \in [0,1]$. Then:
    \begin{equation}
        \eta_d(\mathbb{S}_{m,t}) \leqslant (d-1) \tfrac{m}{m-1} t.
    \end{equation} \label{lem:gen-sym-total-error}
\end{lemma}

\begin{proof}
    See Appendix \ref{sec:proof-gen-sym-total-error}.
\end{proof}

\noindent Hence, combining Lemmas \ref{lem:symmetric-reproduced-measurement} and \ref{lem:gen-sym-total-error}, we can relate the minimal adaptive total error for $N$ uses of $\mathbb{M}$ in any two different dimensions $d' \leqslant d$:

\begin{proposition}
    Let $d' \leqslant d$. Then, for any POVM $\mathbb{M}$ and $N \geqslant 0$, we have:
    \begin{equation}
        \tfrac{1}{d'-1} \eta_{d'}^{(\mathrm{A})}(\mathbb{M},N) \leqslant \tfrac{1}{d-1} \etaA_d(\mathbb{M},N).
    \end{equation} \label{prop:qudit-total-error-lower-bound}
\end{proposition}

\begin{proof}
    See Appendix \ref{sec:proof-qudit-total-error-lower-bound}.
\end{proof}

\noindent In particular, by taking $d' = 2$ and combining with Proposition \ref{prop:adaptive-lower-bound}, we obtain the following lower bound for the minimal adaptive total error:

\begin{corollary}
    For any POVM $\mathbb{M}$, and any $d \geqslant 2$:
    \begin{equation}
        \etaA_d(\mathbb{M},N) \geqslant (d-1) \, \eta_2(\mathbb{M})^N.
    \end{equation} \label{cor:adaptive-lower-bound-d}
\end{corollary}

\noindent Thus, if one use of $\mathbb{M}$ cannot reproduce a projective measurement on a qubit ($\eta_2(\mathbb{M}) > 0$), then an arbitrary finite number of uses cannot perfectly reproduce any projective measurement, only able to approach zero error asymptotically.

\subsection{Further examples} \label{sec:qudit-examples}

\noindent Now, we look again at our example POVMs from Section \ref{sec:qubit-examples}, and see what they can achieve in approximating a $d$ dimensional projective measurement. In such generality, we don't aim to compute the minimal total error, instead just looking for upper bounds to complement Proposition \ref{prop:qudit-total-error-lower-bound} and Corollary \ref{cor:adaptive-lower-bound-d}. Indeed, we find some interesting results (namely Propositions \ref{prop:asymptotic-upper-bound} and \ref{prop:generalised-asymptotic-upper-bound}) by looking at the asymptotic properties of $\eta_d^{(\mathrm{NA/A})}(\mathbb{M}, N)$ as $N \to \infty$, rather than solely exact computations for finite $N$ as we focused on before.

\subsubsection{Imperfect $Z$}

\noindent For the imperfect $Z$ measurement \eqref{eq:imperfect-Z}, we will upper bound $\etaNA_d$ by constructing a non-adaptive circuit, although we will require a couple of concessions. First, we will take $p = q < \frac{1}{2}$, so we only consider symmetric imperfect $Z$ measurements. Secondly, we assume that $N = 2^d n$ is a multiple of $2^d$, as the circuit we have chosen to analyse will require each of $2^d$ channels to be applied the same number of times before a measurement. As we will focus on asymptotics, this is not an unreasonable restriction.

$\mathbb{Z}_{p,p}$ is an abelian POVM, so we refer to Figure \ref{fig:non-adaptive-abelian} for the circuit we should use. The channels $\mathcal{A}_k$ must map each input $\ket{i}$ to either $\ket{0}$ or $\ket{1}$, so we have $2^d$ options for each channel. For each $k$, we let:
\begin{equation}
    A_k = \{ i \in [d] : \mathcal{A}_k (\ketbra{i}{i}) = \ketbra{1}{1} \}
\end{equation}

\noindent i.e. the set of inputs that are mapped to $\ket{1}$, with channel $\mathcal{A}_k$ determined uniquely from $A_k$ as the rest of the inputs will be mapped to $\ket{0}$. As $N$ is a multiple of $2^d$, and there are $2^d$ possibilities for each $A_k \subseteq [d]$, we can ensure that every possible subset of $[d]$ occurs exactly $n = \frac{N}{2^d}$ times among the sets $\{A_k\}$. The idea of this circuit is that each possible pair of distinct inputs $i \neq j \in [d]$ will be distinguishable from each other, because for half the measurements made, $\ket{i}$ and $\ket{j}$ would have been mapped to different states before measurement.

For a vector of measurement outcomes $\mathbf{x} \in [2]^N$ and $i \in [d]$, the probability of obtaining $\mathbf{x}$ from input $\ket{i}$ is:
\begin{equation}
    \mathbb{P}(\mathbf{x} | \ket{i}) = \prod_{k=1}^N p_k(x_k|i),
\end{equation}

\noindent where:
\begin{equation}
    p_k(x_k|i) = \begin{cases} 1-p & \text{if } x_k = [i \in A_k]; \\ p & \text{if } x_k \neq [i \in A_k]. \end{cases} \label{eq:bernoulli-probabilities}
\end{equation}

\noindent Here, we use the Iverson bracket $[P]$, equal to 1 if $P$ is true and 0 if $P$ is false. In our case, $[i \in A_k]$ is equal to 1 if $i \in A_k$, and equal to 0 if $i \notin A_k$, so it serves as an indicator function for the set $A_k$.

Then, for $i,j \in [d]$ with $i \neq j$:
\begin{align}
    &\hspace{-1ex} \log \mathbb{P}(\mathbf{x} | \ket{i}) - \log \mathbb{P}(\mathbf{x} | \ket{j}) \nonumber \\
    &= \sum_{k = 1}^N \log p_k(x_k|i) - \log p_k(x_k|j) \nonumber \\
    &= \left( \log(1-p) - \log p \right) \times \nonumber \\
    &\hspace{2em} \big( \# \{ k: x_k = 0, i \notin A_k, j \in A_k \} \nonumber \\
    &\hspace{2em} - \# \{ k: x_k = 0, i \in A_k, j \notin A_k \} \nonumber \\
    &\hspace{2em} - \# \{ k: x_k = 1, i \notin A_k, j \in A_k \} \nonumber \\
    &\hspace{2em} + \# \{ k: x_k = 1, i \in A_k, j \notin A_k \} \big) \nonumber \\
    &= \left( \log(1-p) - \log p \right) \times \nonumber \\
    &\hspace{2em} \big( \# \{ k: [i \in A_k] \neq [j \in A_k], \ x_k = [i \in A_k] \} \nonumber \\
    &\hspace{2em} - \# \{ k: [i \in A_k] \neq [j \in A_k], \ x_k = [j \in A_k] \} \big) \nonumber \\
    &=  \left( \log(1-p) - \log p \right) \times \nonumber \\
    &\hspace{1em} \left( \tfrac{1}{2}N - 2 \# \{ k: [i \in A_k] \neq [j \in A_k], \ x_k = [j \in A_k] \} \right),
\end{align}

\noindent where in the final line, we use that exactly half of the sets $A_k$ contain one of $i$ and $j$ but not the other, so in total $ \#\{k : [i \in A_k] \neq [j \in A_k] \} = \frac{1}{2} N$.

As $\log(1-p) - \log p > 0$, then $\mathbb{P}(\mathbf{x} | \ket{j}) \geqslant \mathbb{P}(\mathbf{x} | \ket{i})$ if and only if $\# \{ k: [i \in A_k] \neq [j \in A_k], \ x_k = [j \in A_k] \} \geqslant \tfrac{1}{4} N$. Therefore, if the state $\ket{i}$ is input, the probability of getting some result $\mathbf{x} \in [2]^N$ where $\mathbb{P}(\mathbf{x} | \ket{j}) \geqslant \mathbb{P}(\mathbf{x} | \ket{i})$ is:
\begin{align}
    &\mathbb{P} \left(\# \{ k: [i \in A_k] \neq [j \in A_k], \ x_k = [j \in A_k] \} \geqslant \tfrac{1}{4} N \mid \ket{i} \right) \nonumber \\
    &= \mathbb{P}(\mathrm{Bin}(\tfrac{1}{2} N, p) \geqslant \tfrac{1}{4} N) \quad \text{by \eqref{eq:bernoulli-probabilities}} \nonumber \\
    &\leqslant \left( 2 \sqrt{p(1-p)} \right)^{\frac{1}{2} N} \quad \text{by Lemma \ref{lem:binomial-deviation}.}
\end{align}

\noindent For the input $\ket{i}$ to be misidentified as $\ket{j}$, the measurement result $\mathbf{x}$ must satisfy $\mathbb{P}(\mathbf{x} | \ket{j}) \geqslant \mathbb{P}(\mathbf{x} | \ket{i})$ due to maximum likelihood decoding\footnote{In fact, $\mathbb{P}(\mathbf{x} | \ket{j})$ must also be the greatest such probability among all elements of $[d]$, but we only need the upper bound here.}, giving us the upper bound:
\begin{equation}
    \varepsilon_{j|i} \leqslant \left( 2 \sqrt{p(1-p)} \right)^{\frac{1}{2} N}.
\end{equation}

\noindent Therefore, the minimal total error obeys:
\begin{equation}
    \etaNA_d(\mathbb{Z}_{p,p}, N) \leqslant d(d-1) \left( 2 \sqrt{p(1-p)} \right)^{\frac{1}{2} N}.
\end{equation}

\noindent Since $N$ was chosen to be a multiple of $2^d$, this bound is most relevant asymptotically as $N \to \infty$:
\begin{equation}
    \etaNA_d(\mathbb{Z}_{p,p}, N) = O \left( \left( 2 \sqrt{p(1-p)} \right)^{\frac{1}{2} N} \right). \label{eq:imperfect-Z-asymptotics}
\end{equation}

\noindent Moreover, we can infer from this a statement for all POVMs $\mathbb{M}$:

\begin{proposition}
    Let $\mathbb{M}$ be any POVM. Then, for any $d \geqslant 2$:
    \begin{equation}
        \etaNA_d(\mathbb{M}, N) = O \left( \left( \eta_2(\mathbb{M}) (2 - \eta_2(\mathbb{M})) \right)^{\frac{1}{4} N} \right).
    \end{equation} \label{prop:asymptotic-upper-bound}
\end{proposition}

\begin{proof}
    See Appendix \ref{sec:proof-asymptotic-upper-bound}.
\end{proof}

\noindent If $\eta_2(\mathbb{M}) < 1$, i.e. $\mathbb{M}$ is non-trivial, then $\eta_2(\mathbb{M}) (2 - \eta_2(\mathbb{M})) < 1$, so this decays exponentially to 0. Therefore any non-trivial POVM $\mathbb{M}$ can approximate an arbitrary projective measurement with exponentially decreasing error, agreeing with the findings of \cite{linden-2025}.

Moreover, we may consider the rate at which the minimal total error decays exponentially. We define the asymptotic rate $\kappa_d^{(\mathrm{NA})}(\mathbb{M})$ by:
\begin{equation}
    \kappa_d^{(\mathrm{NA})}(\mathbb{M}) = \exp \left( \lim_{N \to \infty} \left( \tfrac{1}{N} \log \eta_d^{(\mathrm{NA})}(\mathbb{M}, N) \right) \right),
\end{equation}

\noindent such that $\eta_d^{(\mathrm{NA})}(\mathbb{M}, N) = \Theta \left( \kappa_d^{(\mathrm{NA})}(\mathbb{M})^N \right)$. We also define $\kappa_d^{(\mathrm{A})}(\mathbb{M})$ equivalently.

Corollary \ref{cor:adaptive-lower-bound-d} gives that:
\begin{equation}
    \kappa_d^{(\mathrm{A})}(\mathbb{M}) \geqslant \eta_2(\mathbb{M}),
\end{equation}

\noindent while we have just shown in Proposition \ref{prop:asymptotic-upper-bound} that:
\begin{equation}
    \kappa_d^{(\mathrm{NA})}(\mathbb{M}) \leqslant \sqrt[4]{\eta_2(\mathbb{M}) (2 - \eta_2(\mathbb{M}))}.
\end{equation}

\noindent In particular, for non-trivial $\mathbb{M}$, the rate of exponential decay is bounded away from 1 as $d \to \infty$, which is perhaps surprising.

In our final example of this section, we will use the generalised symmetric imperfect $Z$ measurement to improve slightly more on this bound (see Proposition \ref{prop:generalised-asymptotic-upper-bound}), but for the moment we will turn our attention to the trine and the SIC-POVM, for which we can find this asymptotic rate exactly in the adaptive case.

\subsubsection{The trine}

\noindent For the trine \eqref{eq:trine}, we'll construct an adaptive circuit, giving an upper bound on $\etaA_d$. First, consider what a single trine measurement can do. Suppose we partition $[d]$ into three sets $A, B, C \subseteq [d]$ with $A \sqcup B \sqcup C = [d]$, and define a channel $\mathcal{A}$ with:
\begin{equation}
    \mathcal{A} : \ketbra{i}{i} \mapsto \begin{cases} \ketbra{\psi_0^\perp}{\psi_0^\perp} & \text{if } i \in A; \\ \ketbra{\psi_1^\perp}{\psi_1^\perp} & \text{if } i \in B; \\ \ketbra{\psi_2^\perp}{\psi_2^\perp} & \text{if } i \in C. \end{cases}
\end{equation}

\noindent where $\ket{\psi_0^\perp} = \ket{1}$, $\ket{\psi_1^\perp} = \frac{\sqrt{3}}{2} \ket{0} - \frac{1}{2} \ket{1}$, and $\ket{\psi_2^\perp} = \frac{\sqrt{3}}{2} \ket{0} + \frac{1}{2} \ket{1}$, such that $\braket{\psi_x}{\psi_x^\perp} = 0$ for all $x \in \{0,1,2\}$.

Then, if the output of $\mathcal{A}$ is measured, only two of the three outcomes of $\trine$ are possible, since $\mathrm{tr} \left( \ketbra{\psi_x^\perp}{\psi_x^\perp} M_x \right) = \frac{2}{3} \left| \braket{\psi_x^\perp}{\psi_x} \right|^2 = 0$. Put another way, whatever the outcome of the measurement, we can rule out $i$ being in one of the three sets $A$, $B$, and $C$ with certainty. Moreover, the problem is now reduced in dimension, from one of $d$ possible input states, to one of approximately $\frac{2}{3} d$, if $A$, $B$, and $C$ are of roughly equal size. Indeed, supposing we choose $A$, $B$, and $C$ optimally, we can ensure that worst-case $\left\lceil \frac{2}{3} d \right\rceil$ possibilities remain.

If we have $N$ measurements, our first step is to create $N$ copies of the input state, i.e. $\ket{i} \mapsto \ket{\mathbf{i}^N}$, which we do as usual with generalised CNOTs. After the first measurement, we can adaptively partition the remaining possible input states into three approximately equal portions, and again rule out a third of them with a second measurement. Then, we repeat this, reducing all the way down to two possible input states if $N$ is large enough. This idea fails for $d = 2$, since then one of the sets $A$, $B$, or $C$ is necessarily empty, so we can't guarantee that we are ruling out states with every measurement. However, from Section \ref{sec:qubit-examples}, we already know what to do optimally in the qubit case, so we do as in Figure \ref{fig:trine-optimal} with the remaining measurements.

What remains is to upper bound the number of times the set of input states has to be partitioned to get down to the qubit case, i.e. how many reductions in dimension $d \mapsto \left\lceil \frac{2}{3} d \right\rceil$ are needed to get down to just two? Since $\left\lceil \frac{2}{3} d \right\rceil \leqslant \tfrac{2}{3} d + \tfrac{2}{3}$, we can solve the difference equation:
\begin{equation}
    d_{n+1} = \tfrac{2}{3} d_n + \tfrac{2}{3} \quad \Rightarrow \quad d_n = \left( \tfrac{2}{3} \right)^n (d_0 - 2) + 2, \label{eq:trine-difference-equation}
\end{equation}

\noindent and so the dimension of the remaining space after $n$ adaptive measurements is at most $\left( \tfrac{2}{3} \right)^n (d - 2) + 2$. This is at most 3 when $n \geqslant \frac{\log(d-2)}{\log 3 - \log 2}$, and just one more measurement is needed to reduce the three dimensional case to the two dimensional case, hence for $N > \frac{\log(d-2)}{\log 3 - \log 2} + 1$ our circuit correctly identifies each input $\ket{i}$ with probability:
\begin{align}
    \varepsilon_{i|i} &\geqslant 1 - \tfrac{1}{2} \etaA_2 \left( \trine, N - \left\lceil \tfrac{\log(d-2)}{\log 3 - \log 2} + 1 \right\rceil \right) \nonumber \\
    &= 1 - \tfrac{1}{2} \tfrac{1}{3^{N - \left\lceil \frac{\log(d-2)}{\log 3 - \log 2} + 1 \right\rceil}} \quad \text{by \eqref{eq:eta-trine}} \nonumber \\
    &\geqslant 1 - \tfrac{1}{2} \tfrac{3^{\frac{\log(d-2)}{\log 3 - \log 2} + 2}}{3^N} \nonumber \\
    &= 1 - \tfrac{9}{2} (d-2)^{\frac{\log 3}{\log 3 - \log 2}} \tfrac{1}{3^N}.
\end{align}

\noindent We can upper bound the minimal adaptive total error by the total error achieved by this protocol, giving:
\begin{align}
    \etaA_d(\trine, N) &\leqslant d - \sum_{i \in [d]} \varepsilon_{i|i} \nonumber \\
    &\leqslant d - d \left( 1 - \tfrac{9}{2} (d-2)^{\frac{\log 3}{\log 3 - \log 2}} \tfrac{1}{3^N} \right) \nonumber \\
    &= \tfrac{9}{2} d (d-2)^{\frac{\log 3}{\log 3 - \log 2}} \tfrac{1}{3^N} \nonumber \\
    &\leqslant 5 d^4 \tfrac{1}{3^N} \quad \text{as } \tfrac{\log 3}{\log 3 - \log 2} = 2.70... < 3.
\end{align}

\noindent From the lower bound in Corollary \ref{cor:adaptive-lower-bound-d}, we have that $\etaA_d(\trine, N) \geqslant (d-1) \tfrac{1}{3^N}$, so we can deduce that $\etaA_d(\trine, N) = \Theta \left( \frac{1}{3^N} \right)$, i.e.:
\begin{equation}
    \kappa_d^{(\mathrm{A})}(\trine) = \tfrac{1}{3},
\end{equation}

\noindent for all $d$. An intriguing possibility, which we haven't been able to prove, is that non-adaptive protocols cannot reach this rate for some $d > 2$, meaning an asymptotic form of adaptive advantage where $\kappa_d^{(\mathrm{A})} < \kappa_d^{(\mathrm{NA})}$.

\subsubsection{The qubit SIC-POVM}

\noindent Given its similar definition to the trine, with all POVM elements being rank one, it perhaps comes as little surprise that we can make a similar argument for the qubit SIC-POVM \eqref{eq:tetra}, and construct an adaptive circuit to upper bound $\etaA_d$. To be specific, if we have a partition of $[d]$ into four sets $A,B,C,D \subseteq [d]$ with $A \sqcup B \sqcup C \sqcup D = [d]$, then we can map the computational basis states $\ket{i}$ to one of $\ket{\psi_0^\perp} = \ket{1}$, $\ket{\psi_1^\perp} = \frac{\sqrt{2}}{\sqrt{3}} \ket{0} - \frac{1}{\sqrt{3}} \ket{1}$, $\ket{\psi_2^\perp} = \frac{\sqrt{2}}{\sqrt{3}} \ket{0} - e^{\frac{2\pi i}{3}} \frac{1}{\sqrt{3}} \ket{1}$, or $\ket{\psi_3^\perp} = \frac{\sqrt{2}}{\sqrt{3}} \ket{0} - e^{\frac{4\pi i}{3}} \frac{1}{\sqrt{3}} \ket{1}$, depending on which of the four sets $i$ lies in. Now, since $\braket{\psi_x}{\psi_x^\perp} = 0$, the result of a SIC-POVM measurement on this state rules out $i$ being in one of the four sets $A$, $B$, $C$, $D$, reducing the dimension of the problem to about $\frac{3}{4} d$.

As before, we can repeat this adaptively, with each measurement ruling out a further quarter (approximately) possible inputs. We solve a similar difference equation to \eqref{eq:trine-difference-equation}, which this time is:
\begin{equation}
    d_{n+1} = \tfrac{3}{4} d_n + \tfrac{3}{4} \quad \Rightarrow \quad d_n = \left( \tfrac{3}{4} \right)^n (d_0 - 3) + 3.
\end{equation}

\noindent From this, we infer that it takes at most $\frac{\log(d-3)}{\log 4 - \log 3}$ measurements to reduce to the four dimensional case, and one more to get to $d = 3$. However, this time we can only deterministically get as far as $d = 3$, since we have a four outcome measurement, so one of the outcomes will be possible for all three inputs and not rule anything out.

Therefore, let's first suppose that $d = 3$, and solve that case first. We map $\ket{i} \mapsto \ket{\psi_i^\perp}$ for each $i \in \{0,1,2\}$, and measure as before. There is always a probability of $\frac{1}{3}$ of getting the measurement result 3, which gives us no information, but with probability $\frac{2}{3}$ we will get a result in $\{0,1,2\} \setminus \{i\}$, allowing us to rule out one of the other possible input states, and reduce to $d = 2$, which we have already solved in Section \ref{sec:qubit-examples}. Our strategy will therefore be to repeat this measurement on copies of $\ket{i}$ until we hit the $\frac{2}{3}$ probability of getting an informative outcome, and if we never succeed, the best we can do is to have a one in three guess at $i$.

Given an input $\ket{i} \in \{\ket{0}, \ket{1}, \ket{2}\}$, we calculate the success probability of the circuit by conditioning on the number of measurements $n$ that we make before reducing to the $d = 2$ case (if at all):
\begin{align}
    \varepsilon_{i|i} &= \sum_{n = 1}^N \tfrac{1}{3^{n-1}} \tfrac{2}{3} \left( 1 - \tfrac{1}{2} \etaA_2(\SIC, N-n) \right) + \tfrac{1}{3^N} \tfrac{1}{3} \nonumber \\
    &= 2 \sum_{n = 1}^N \tfrac{1}{3^n} - \sum_{n = 1}^N \tfrac{1}{3^n} \etaA_2(\SIC, N-n) + \tfrac{1}{3^{N+1}} \nonumber \\ 
    &= 2 \sum_{n = 1}^N \tfrac{1}{3^n} - \sum_{n = 1}^{N-1} \tfrac{1}{3^n} \tfrac{1}{2^{N-n-1}} \left( 1 - \tfrac{1}{\sqrt{3}} \right) - \tfrac{1}{3^N} + \tfrac{1}{3^{N+1}} \nonumber \\
    &\hspace{4em} \text{by \eqref{eq:eta-tetra} and } \etaA_2(\SIC, 0) = 1 \nonumber \\
    &= 3 \tfrac{\frac{1}{3} - \frac{1}{3^{N+1}}}{1 - \frac{1}{3}} - \tfrac{1}{2^{N-1}} \left( 1 - \tfrac{1}{\sqrt{3}} \right) \tfrac{\frac{2}{3} - \frac{2^N}{3^N}}{1 - \tfrac{2}{3}} - \tfrac{2}{3^{N+1}} \nonumber \\
    &= 1 - \tfrac{1}{3^N} - \left( 1 - \tfrac{1}{\sqrt{3}} \right) \left( \tfrac{1}{2^{N-2}} - \tfrac{2}{3^{N-1}} \right) - \tfrac{2}{3^{N+1}} \nonumber \\
    &= 1 - \left( 1 - \tfrac{1}{\sqrt{3}} \right) \tfrac{1}{2^{N-2}} + \left( 13 - 6 \sqrt{3} \right) \tfrac{1}{3^{N+1}}.
\end{align}

\noindent Then, the minimal total error is upper bounded by the total error that this protocol achieves, so:
\begin{align}
    \eta_3^{(\mathrm{A})}(\SIC, N) &\leqslant 3 -  \sum_{i \in [3]} \varepsilon_{i|i} \nonumber \\
    &= 3 \left( 1 - \tfrac{1}{\sqrt{3}} \right) \tfrac{1}{2^{N-2}} - 3 \left( 13 - 6 \sqrt{3} \right) \tfrac{1}{3^{N+1}} \nonumber \\
    & \leqslant \left( 3 - \sqrt{3} \right) \tfrac{1}{2^{N-2}}.
\end{align}

\noindent Returning to the case of general $d$, we know that we can reduce to the $d = 3$ case with the first $\left\lceil \frac{\log(d-3)}{\log 4 - \log 3} + 1 \right\rceil$ measurements. Then, the probability of error for any input $\ket{i}$ will be inherited from the $d = 3$ case, and will be $\tfrac{1}{3} \eta_3^{(\mathrm{A})} \left( \SIC, N - \left\lceil \frac{\log(d-3)}{\log 4 - \log 3} + 1 \right\rceil \right)$, so:
\begin{align}
    \etaA_d(\SIC, N) &\leqslant \tfrac{d}{3} \eta_3^{(\mathrm{A})} \left( \SIC, N - \left\lceil \tfrac{\log(d-3)}{\log 4 - \log 3} + 1 \right\rceil \right) \nonumber \\
    &\leqslant d \left( 1 - \tfrac{1}{\sqrt{3}} \right) \tfrac{1}{2^{N - \left\lceil \frac{\log(d-3)}{\log 4 - \log 3} + 1 \right\rceil - 2}} \nonumber \\
    &\leqslant d \left( 1 - \tfrac{1}{\sqrt{3}} \right) \tfrac{2^{\frac{\log(d-3)}{\log 4 - \log 3} + 4}}{2^N} \nonumber \\
    &= 16 \left( 1 - \tfrac{1}{\sqrt{3}} \right) d (d-3)^{\frac{\log 2}{\log 4 - \log 3}} \tfrac{1}{2^N} \nonumber \\
    &\leqslant 7 d^4 \tfrac{1}{2^N}.
\end{align}

\noindent as $\frac{\log 2}{\log 4 - \log 3} = 2.40... < 3$ and $16 \left( 1 - \frac{1}{\sqrt{3}} \right) = 6.76... < 7$.

From Corollary \ref{cor:adaptive-lower-bound-d}, we get a lower bound of $\etaA_d(\SIC, N) \geqslant (d-1) \tfrac{1}{2^{N-1}} \left( 1 - \frac{1}{\sqrt{3}} \right)$, and therefore we may conclude that:
\begin{equation}
    \kappa_d^{(\mathrm{A})}(\SIC) = \tfrac{1}{2},
\end{equation}

\noindent for all $d$. As with the trine, we leave as an open question the possibility of an asymptotic adaptive advantage where the best non-adaptive protocols cannot reach this rate for some $d > 2$.

\subsubsection{The generalised symmetric imperfect $Z$}

\noindent Our first example in this section was the symmetric imperfect $Z$. We have noted before that the generalised symmetric imperfect $Z$ \eqref{eq:generalised-symmetric-imperfect-Z} is a generalisation of this ($\mathbb{Z}_{p,p} = \mathbb{S}_{2,p}$), so for general $\mathbb{S}_{m,t}$, we will try to use a similar non-adaptive circuit to bound $\etaNA_d$. We will assume that $t < \frac{m-1}{m}$ (so that the likeliest outcome of measuring $\ket{i}$ is $i$), but otherwise, we will allow for $m$ and $t$ to be arbitrary. As a reminder, $m$ is the dimension of the system that our measuring device measures, which is in general different from the input dimension $d$.

As $\mathbb{S}_{m,t}$ is abelian, we will consider non-adaptive circuits of the form shown in Figure \ref{fig:non-adaptive-abelian}, where we only need to pick channels $\mathcal{A}_1, \mathcal{A}_2, ..., \mathcal{A}_N$ since the classical post-processing $f$ will be defined by maximum likelihood decoding. Every such choice will be a channel of the form $\mathcal{A}_\mathbf{a} : \mathcal{B}(\mathbb{C}^d) \to \mathcal{B}(\mathbb{C}^m)$, for some vector $\mathbf{a} \in [m]^d$, with the action\footnote{Since we know the diagonalising basis of the POVM is just the computational basis, we simplify the notation that we used previously to $\ket{a_i} = \ket{v_{a_i}}$.}:
\begin{equation}
    \mathcal{A}_\mathbf{a}: \ketbra{i}{i} \mapsto \ketbra{a_i}{a_i}.
\end{equation}

\noindent There are $m^d$ such channels, although we won't use all of them. Instead, we will only use those where $\mathbf{a}$ is as balanced as possible. We express $d = m q + r$ for integers $q \geqslant 0$ and $r \in [m]$, and then require that there are $m-r$ values of $x \in [m]$ that appear $q$ times in $\mathbf{a}$, and $r$ values of $x \in [m]$ that appear $q+1$ times in $\mathbf{a}$. We call the set of such vectors $S$, and let $T = |S|$. The exact value of $T$ doesn't matter to us, but for ease of calculation we will assume that $N$ is a multiple of $T$ (noting that this won't affect the asymptotics). Our non-adaptive circuit will then be of the form of Figure \ref{fig:non-adaptive-abelian}, where for each $\mathbf{a} \in S$, the channel $\mathcal{A}_\mathbf{a}$ appears exactly $\frac{N}{T}$ times.

To calculate error probabilities, we first need to know what proportion of these $T$ vectors have $a_i \neq a_j$, for arbitrary $i \neq j \in [d]$. We do this by choosing $\mathbf{a} \in S$ uniformly at random, so that the required proportion is the probability that $a_i \neq a_j$. Each entry of $\mathbf{a}$ appears either $q$ or $q+1$ times, so we will break the calculation down depending on how many times $a_i$ appears.

The probability that $a_i$ appears $q$ times in $\mathbf{a}$ is simply the proportion of the entries of $\mathbf{a}$ that appear $q$ times, which is $\frac{(m-r)q}{d}$. Then, given that $a_i$ appears $q$ times, there are $q-1$ other occurrences of that value in $\mathbf{a}$, and $a_j$ cannot be one of them, giving a factor of $1 - \frac{q-1}{d-1}$.

Conversely, the probability that $a_i$ appears $q+1$ times in $\mathbf{a}$ is the proportion of entries of $\mathbf{a}$ that appear $q+1$ times, equal to $\frac{r(q+1)}{d}$. Now, there are $q$ other occurrences of the value $a_i$ in $\mathbf{a}$ that we need to avoid with $a_j$, which happens with probability $1 - \frac{q}{d-1}$.

Putting this all together, the total number of $\mathbf{a} \in S$ with $a_i \neq a_j$ is:
\begin{align}
    &\#\{\mathbf{a} \in S : a_j \neq a_i \} \nonumber \\
    &= \mathbb{P}( a_j \neq a_i ) \times T \quad \text{if $\mathbf{a}$ chosen uniformly from $S$} \nonumber \\
    &= \left( \tfrac{(m-r)q}{d} (1 - \tfrac{q-1}{d-1}) + \tfrac{r(q+1)}{d} (1 - \tfrac{q}{d-1}) \right) T \nonumber \\
    &= \left( 1 - \tfrac{q(d-m+r)}{d(d-1)} \right) T.
\end{align}

\noindent For input state $\ket{i} \in \{ \ket{0}, \ket{1}, ..., \ket{d-1} \}$, the probability of the outcome $\mathbf{x} \in [m]^N$ is:
\begin{equation}
    \mathbb{P}(\mathbf{x} | \ket{i}) = \prod_{k=1}^N p_k(x_k|i),
\end{equation}

\noindent where:
\begin{equation}
    p_k(x_k|i) = \begin{cases} 1-t & \text{if } x_k = [\mathbf{a_k}]_i; \\ \frac{t}{m-1} & \text{if } x_k \neq [\mathbf{a_k}]_i. \end{cases}
\end{equation}

\noindent Thus, for $i \neq j \in [d]$, we have:
\begin{align}
    &\log \mathbb{P}(\mathbf{x} \, | \ket{i}) - \log \mathbb{P}(\mathbf{x} \, | \ket{j}) \nonumber \\
    &= \sum_{k = 1}^N \log p_k(x_k|i) - \log p_k(x_k|j) \nonumber \\
    &= \left( \log(1-t) - \log \left( \tfrac{t}{m-1} \right) \right) \times \nonumber \\
    & \hspace{2em} \big( \# \{ k: [\mathbf{a}_k]_i \neq [\mathbf{a}_k]_j, \ x_k = [\mathbf{a}_k]_i \} \nonumber \\
    &\hspace{2em} - \# \{ k: [\mathbf{a}_k]_i \neq [\mathbf{a}_k]_j, \ x_k = [\mathbf{a}_k]_j \} \big).
\end{align}

\noindent From our choice of circuit, there are $N' = \left( 1 - \tfrac{q(d-m+r)}{d(d-1)} \right) N$ values of $k$ for which $[\mathbf{a}_k]_i \neq [\mathbf{a}_k]_j$, which we may assume by reordering are exactly those $k \in [N']$. Thus:
\begin{align}
    &\log \mathbb{P}(\mathbf{x} \, | \ket{i}) - \log \mathbb{P}(\mathbf{x} \, | \ket{j}) \nonumber \\
    &= \left( \log(1-t) - \log \left( \tfrac{t}{m-1} \right) \right) \times \nonumber \\
    & \hspace{2em} \big( \# \{ x_k = [\mathbf{a}_k]_i \}_{k \in [N']} - \# \{ x_k = [\mathbf{a}_k]_j \}_{k \in [N']} \big).
\end{align}

\noindent Given an input state of $\ket{i}$, a necessary condition for a string of outcomes $\mathbf{x}$ to be decoded as $j \neq i$ is for $\mathbb{P}(\mathbf{x} \, | \ket{i}) \leqslant \mathbb{P}(\mathbf{x} \, | \ket{j})$, and since $t < \tfrac{m-1}{m}$, $\log(1-t) - \log(\frac{t}{m-1}) > 0$, so this occurs exactly when:
\begin{equation}
    \{ x_k = [\mathbf{a}_k]_j \}_{k \in [N']} \geqslant \# \{ x_k = [\mathbf{a}_k]_i \}_{k \in [N']}.
\end{equation}

\noindent We may now relabel measurement outcomes such that $[\mathbf{a}_k]_i = 0$ and $[\mathbf{a}_k]_j = 1$ for all $k \in [N']$, such that the distribution of $\mathbf{x}_{N'} = (x_k)_{k \in [N']}$ is exactly as given in Lemma \ref{lem:multivariate-deviation}(a). As such:
\begin{align}
    \varepsilon_{j|i} &\leqslant \mathbb{P} \left( \# \{ x_k = [\mathbf{a}_k]_j \}_{k \in [N']} \geqslant \# \{ x_k = [\mathbf{a}_k]_i \}_{k \in [N']} \mid \ket{i} \right) \nonumber \\
    &= \mathbb{P} \left( \# \{ x_k = 1 \}_{k \in [N']} \geqslant \# \{ x_k = 0 \}_{k \in [N']} \mid \ket{i} \right) \nonumber \\
    &\leqslant \left( \tfrac{(m-2)t}{m-1} + 2 \sqrt{\tfrac{t(1-t)}{m-1}} \right)^{N'} \quad \text{by Lemma \ref{lem:multivariate-deviation}(a)},
\end{align}

\noindent and so:
\begin{align}
    &\etaNA_d(\mathbb{S}_{m,t}, N) \nonumber \\
    &\hspace{2em} \leqslant d(d-1) \left( \tfrac{(m-2)t}{m-1} + 2 \sqrt{\tfrac{t(1-t)}{m-1}} \right)^{\left( 1 - \frac{q(d-m+r)}{d(d-1)} \right) N}. \label{eq:gen-sym-eta-upper-bound}
\end{align}

\noindent \eqref{eq:gen-sym-eta-upper-bound} is only proven for any $N$ that is a multiple of $T$, which is a very small proportion, so the more meaningful result is asymptotic as before:

\begin{proposition}
    Let $\mathbb{M}$ be any POVM. Let $d, \ell \geqslant 2$, with integers $q \geqslant 0$ and $r \in [\ell]$ defined by integer division such that $d = \ell q + r$. Then:
    \begin{equation}
        \etaNA_d(\mathbb{M}, N) = O \left( \kappa^N \right),
    \end{equation}

    \noindent where:
    \begin{equation}
        \kappa = \left( \tfrac{(\ell-2) \eta_\ell(\mathbb{M})}{\ell(\ell-1)} + 2 \sqrt{\tfrac{\eta_\ell(\mathbb{M}) (\ell - \eta_\ell(\mathbb{M}))}{\ell^2 (\ell-1)}} \right)^{ 1 - \frac{q(d-\ell+r)}{d(d-1)}}. \label{eq:upper-bound-rate}
    \end{equation} \label{prop:generalised-asymptotic-upper-bound}
\end{proposition}

\begin{proof}
    See Appendix \ref{sec:proof-generalised-asymptotic-upper-bound}.
\end{proof}

\noindent Note that here, we rename $m$ from \eqref{eq:gen-sym-eta-upper-bound} to $\ell$, to avoid confusion with the number of outcomes of the measurement $\mathbb{M}$, which we have otherwise been referring to as $m$ throughout the paper.

For $\ell = 2$, Proposition \ref{prop:generalised-asymptotic-upper-bound} gives a slightly tighter bound than the equivalent one we obtained in Proposition \ref{prop:asymptotic-upper-bound}. For an even tighter bound, if we are given $\mathbb{M}$ and $d$, we can choose $m$ to get the best asymptotic upper bound:
\begin{equation}
    \kappa_d^{(\mathrm{NA})}(\mathbb{M}) \leqslant \min_{\ell \geqslant 2} \left\{ \kappa \right\},
\end{equation}

\noindent for $\kappa$ as given in \eqref{eq:upper-bound-rate}.

\section{Discussion}

\noindent In this work, we have shown that adaptive circuits using multiple noisy measuring devices can implement effective measurements with reduced error compared to non-adaptive circuits. We have shown this in multiple examples, with particular emphasis on qubit measurements with two outcomes, where we found that for almost all such measurements, adaptivity gives an advantage when using as few as three or four measurements. We have also shown that the amount of advantage, quantified as the ratio of the minimal non-adaptive total error to the minimal adaptive total error, is unbounded, as the number of uses of the measuring device increases.

There are a couple of implementation details that are worth a mention here. Throughout, we have assumed that the only source of noise in our system is in the measurement. As such, our results will be most applicable in scenarios where measurement noise is much greater than other sources of noise, such as on the input state or on the unitary operations that make up the rest of the circuit. We also haven't factored in the challenge of implementing adaptive circuits, with the reading of measurement outcomes to the classical processing and the implementation of classically controlled operations needing to be performed accurately on millisecond timescales. Nonetheless, \cite{kim-2026} have shown that simple multiple measurement circuits can maintain their advantage under realistic noise on other circuit components. Moreover, recent work such as \cite{foss-feig-2023, iqbal-2024, pokharel-2025} shows that adaptive advantage in other tasks can be experimentally demonstrated, so these are evidently challenges that current knowledge and technology is beginning to overcome.

Our choice of figure of merit, the total error \eqref{eq:total-error-2}, is frequently used in an experimental setting, and has the added benefit of linearity which helps to simplify calculations. However, it is worth pointing out that this choice does have an impact on the detailed analysis of whether and to what extent there will be adaptive advantage for any given number of uses of any given measuring device. Nonetheless, we emphasise that the idea of adaptive advantage, where adaptive measurement circuits can act in a way that no non-adaptive measurement circuit could, transcends this choice. Preliminary calculations show that adaptive advantage can be found for alternative choices of figure of merit, and indeed we would be surprised if adaptive advantage did not similarly arise for any reasonable alternative choice.

Our work can be reframed in several ways. We have usually considered it in the setting of measurement reproduction of \cite{linden-2025}, where we are trying to approximate the effect of a projective measurement with our noisy measuring devices. Alternatively, it can also be viewed information-theoretically as an unusual state discrimination task, where we are given one of $d$ orthonormal input states, and we wish to distinguish them. Usually, this would be trivial with a projective measurement, but if we limit ourselves only to a noisy measuring device, we arrive at the problem central to this work. This form of state discrimination is dual to the usual setting of non-orthogonal states with no restrictions on measurements (such as in \cite{li-2022}, with a similar focus on adaptivity), and is worthy of further investigation.

Additionally, we have used the language of reducing measurement error throughout as a measure of the success of our circuits, but that is not the only interpretation. Instead, we may consider an adaptive advantage to mean a reduction in the total number of measurements needed to achieve a certain error threshold. This could be beneficial for systems where measurements have an associated cost, either in time or in adding noise to the system.

\section*{Acknowledgements}

The authors would like to thank Chris Corlett for useful discussions. JB acknowledges funding from the EPSRC Centre for Doctoral Training in Quantum Engineering (EP/S023607/1). 

\bibliography{refs.bib}

\appendix

\section{Optimal adaptive circuit using three imperfect $Z$ measurements} \label{sec:optimal-adaptive-noisy-Z-argument}

\noindent In this appendix, we complete the argument from Section \ref{sec:optimal-adaptive-noisy-Z} to determine the optimal adaptive circuit in the case of three imperfect $Z$ measurements.

\subsection{Reduction of search space}

\noindent Firstly, we reduce the search space of circuits to those of the form of Figure \ref{fig:adaptive-noisy-Z-simplification} as follows. After all the measurements are made, the choice of post-processing to minimise total error remains maximum likelihood decoding. Thus, if we notate the state of the first, second, and third qubits immediately before they are measured as $\sigma_i$, $\tau_i^{(x_1)}$, and $\upsilon_i^{(x_1,x_2)}$ respectively, then the total error $\eta$ can be written in terms of these states as:
\begin{align}
    \eta &= \sum_{\mathbf{x} \in \{ 0,1 \}^3} \min_{i \in \{0, 1\}} \Big\{ \mathrm{tr} \left( \sigma_i Z_{x_1} \right) \, \mathrm{tr} \left( \tau_i^{(x_1)} Z_{x_2} \right) \nonumber \\
    &\hspace{11em} \times \mathrm{tr} \left( \upsilon_i^{(x_1,x_2)} Z_{x_3} \right) \Big\} .\label{eq:total-error-sum-of-minima-adaptive}
\end{align}

\noindent We argue similarly to the non-adaptive case in Section \ref{sec:optimal-non-adaptive-noisy-Z} that the choice of these states can be reduced from 14 arbitrary qubit states to a finite list of options. First, we may note that off-diagonal elements of $\sigma_i$, $\tau_i^{(x_1)}$, and $\upsilon_i^{(x_1,x_2)}$ can be assumed to be zero by dephasing each state,  as all $Z_x$ are diagonal in the computational basis meaning that the off-diagonal elements don't contribute to the trace. Also, as a function of each state, each term of the sum in \eqref{eq:total-error-sum-of-minima-adaptive} is concave as before, because the minimum is concave. Therefore, $\eta$ is concave as a function of each state, so we may assume all states to be pure. As they are diagonal in the computational basis, this means that each state is either $\ketbra{0}{0}$ or $\ketbra{1}{1}$. 

Moreover, we can assume that $\sigma_0 \neq \sigma_1$, $\tau_0^{(x_1)} \neq \tau_1^{(x_1)}$, and $\upsilon_0^{(x_1,x_2)} \neq \upsilon_1^{(x_1,x_2)}$, as otherwise the measurement result would be independent of $i$ and the processing on that qubit could be replaced with whatever choice of quantum processing we like, with the measurement result ignored and a simulated result generated by classical randomness. Finally, we may assume without loss of generality that $\sigma_0 = \ketbra{0}{0}$ and $\sigma_1 = \ketbra{1}{1}$, since if the opposite is true, we can put an $X$ gate at the start and a bit flip at the end to swap all states between $\ketbra{0}{0}$ and $\ketbra{1}{1}$, which has the effect of swapping $\varepsilon_0$ and $\varepsilon_1$, but keeps their sum the same.

Therefore, we may reduce our search space to circuits of the form shown in Figure \ref{fig:adaptive-noisy-Z-simplification}, with the unitaries chosen such that $V^{(x_1)} \ketbra{i}{i} V^{(x_1)\dagger} = \tau_i^{(x_1)}$ and $W^{(x_1,x_2)} \ketbra{i}{i} W^{(x_1,x_2)\dagger} = \upsilon_i^{(x_1,x_2)}$, with $f$ determined by maximum likelihood decoding. We choose the unitaries accordingly as:
\begin{align}
    V^{(x_1)} &= \begin{cases} I & \text{if } \tau_0^{(x_1)} = \ketbra{0}{0}, \tau_1^{(x_1)} = \ketbra{1}{1}; \\ X & \text{if } \tau_0^{(x_1)} = \ketbra{1}{1}, \tau_1^{(x_1)} = \ketbra{0}{0}; \end{cases} \nonumber \\
    W^{(x_1,x_2)} &= \begin{cases} I & \text{if } \upsilon_0^{(x_1,x_2)} = \ketbra{0}{0}, \upsilon_1^{(x_1,x_2)} = \ketbra{1}{1}; \\ X & \text{if } \upsilon_0^{(x_1,x_2)} = \ketbra{1}{1}, \upsilon_1^{(x_1,x_2)} = \ketbra{0}{0}. \end{cases}
\end{align}

\noindent Thus, our task now becomes to choose the states $\tau_i^{(x_1)}$ and $\upsilon_i^{(x_1,x_2)}$ optimally, or equivalently, to make the optimal choice of $I$ or $X$ for each unitary.

\subsection{Choosing unitaries $W^{(x_1,x_2)}$ optimally}

\noindent To optimise the choice of unitaries, we work backwards, beginning by supposing that $x_1$, $x_2$, and $V^{(x_1)}$ are fixed, and looking for $W^{(x_1,x_2)}$. We let $\alpha_i(x_1,x_2)$ be the probability that an input of $\ket{i}$ gives the first two measurement outcomes $x_1$ and $x_2$:
\begin{equation}
    \alpha_i(x_1,x_2) = \mathrm{tr} \left( \sigma_i Z_{x_1} \right) \, \mathrm{tr} \left( \tau_i^{(x_1)} Z_{x_2} \right).
\end{equation}

\noindent This is fixed, since $x_1$, $x_2$, and $V^{(x_1)}$ are, so for brevity we drop the explicit dependence, and refer to $\alpha_i = \alpha_i(x_1,x_2)$.

Then, the only terms in the sum \eqref{eq:total-error-sum-of-minima-adaptive} that depend on $W^{(x_1,x_2)}$ are $\mathbf{x} = (x_1,x_2,0)$ and $\mathbf{x} = (x_1,x_2,1)$, so we look to choose $W^{(x_1,x_2)}$ to minimise the sum of just these terms, which we call $\eta^{(x_1,x_2)}$:
\begin{equation}
    \eta^{(x_1,x_2)} := \sum_{x_3 \in \{0,1\}} \min_{i \in \{0, 1\}} \left\{ \alpha_i \, \mathrm{tr} \left( \upsilon_i^{(x_1,x_2)} Z_{x_3} \right) \right\}. \label{eq:W-error-contribution}
\end{equation}

\noindent We calculate this for each choice $W^{(x_1,x_2)} = I$ and $W^{(x_1,x_2)} = X$, and find which is optimal for any given $\alpha_0$ and $\alpha_1$.

In the case $W^{(x_1,x_2)} = I$, i.e. $\upsilon_0^{(x_1,x_2)} = \ketbra{0}{0}$ and $\upsilon_1^{(x_1,x_2)} = \ketbra{1}{1}$, \eqref{eq:W-error-contribution} becomes:
\begin{align}
    &\hspace{-2ex} \eta^{(x_1,x_2)} \big|_{W^{(x_1,x_2)} = I} \nonumber \\
    &= \min \{ (1-p) \alpha_0, q \alpha_1 \} + \min\{ p \alpha_0, (1-q) \alpha_1 \} \nonumber \\
    &= \begin{cases} \alpha_0 & \text{if } \frac{\alpha_0}{\alpha_1} \leqslant \frac{q}{1-p}; \\ p \alpha_0 + q \alpha_1 & \text{if } \frac{q}{1-p} \leqslant \frac{\alpha_0}{\alpha_1} \leqslant \frac{1-q}{p}; \\ \alpha_1 & \text{if } \frac{\alpha_0}{\alpha_1} \geqslant \frac{1-q}{p}. \end{cases} \label{eq:W-is-I}
\end{align}

\noindent where we use that $p + q \leqslant 1$ without loss of generality.

Note that where the inequalities overlap, both expressions coincide. This is a deliberate choice, we want to emphasise at these boundary points that both expressions are minimal, and we continue this convention throughout.

Meanwhile, in the case $W^{(x_1,x_2)} = X$, i.e. $\upsilon_0^{(x_1,x_2)} = \ketbra{1}{1}$ and $\upsilon_1^{(x_1,x_2)} = \ketbra{0}{0}$, \eqref{eq:W-error-contribution} becomes:
\begin{align}
    &\hspace{-2ex} \eta^{(x_1,x_2)} \big|_{W^{(x_1,x_2)} = X} \nonumber \\
    &= \min \{ q \alpha_0, (1-p) \alpha_1 \} + \min\{ (1-q) \alpha_0, p \alpha_1 \} \nonumber \\
    &= \begin{cases} \alpha_0 & \text{if } \frac{\alpha_0}{\alpha_1} \leqslant \frac{p}{1-q}; \\ q \alpha_0 + p \alpha_1 & \text{if } \frac{p}{1-q} \leqslant \frac{\alpha_0}{\alpha_1} \leqslant \frac{1-p}{q}; \\ \alpha_1 & \text{if } \frac{\alpha_0}{\alpha_1} \geqslant \frac{1-p}{q}, \end{cases} \label{eq:W-is-X}
\end{align}

\begin{table*}[t]
    \centering
    \begin{tabular}{c||c|c||c|c}
        $\tfrac{\alpha_0(x_1,x_2)}{\alpha_1(x_1,x_2)} \in $ & $\eta^{(x_1,x_2)}$ if $W^{(x_1,x_2)} = I$ & $\eta^{(x_1,x_2)}$ if $W^{(x_1,x_2)} = X$ & Optimal choice of $W^{(x_1,x_2)}$ & Optimal $\eta^{(x_1,x_2)}$ \\ \hline \hline
        $( 0, \tfrac{p}{1-q} ]$ & $\alpha_0$ & $\alpha_0$ & $I$ or $X$ & $\alpha_0$ \\ \hline
        $[ \tfrac{p}{1-q}, \tfrac{q}{1-p} ]$ & $\alpha_0$ & $q \alpha_0 + p \alpha_1$ & $X$ & $q \alpha_0 + p \alpha_1$ \\ \hline
        $[ \tfrac{q}{1-p}, 1 ]$ & $p \alpha_0 + q \alpha_1$ & $q \alpha_0 + p \alpha_1$ & $X$ & $q \alpha_0 + p \alpha_1$ \\ \hline
        $[ 1, \tfrac{1-p}{q} ]$ & $p \alpha_0 + q \alpha_1$ & $q \alpha_0 + p \alpha_1$ & $I$ & $p \alpha_0 + q \alpha_1$ \\ \hline
        $[ \tfrac{1-p}{q}, \tfrac{1-q}{p} ]$ & $p \alpha_0 + q \alpha_1$ & $ \alpha_1$ & $I$ & $p \alpha_0 + q \alpha_1$ \\ \hline
        $[ \tfrac{1-q}{p}, \infty )$ & $\alpha_1$ & $ \alpha_1$ & $I$ or $X$ & $\alpha_1$
    \end{tabular}
    \caption{The optimal choice of unitary $W^{(x_1,x_2)}$ and corresponding optimal values of $\eta^{(x_1,x_2)}$, depending on the ratio $\frac{\alpha_0(x_1,x_2)}{\alpha_1(x_1,x_2)}$. We discount the ratio equalling $0$ and $\infty$, since if exactly one of $\alpha_0(x_1,x_2)$ and $\alpha_1(x_1,x_2)$ are zero, we already know $i$ and further processing doesn't matter, and if both are zero the outcomes $x_1$, $x_2$ are always impossible and again further processing doesn't matter. Note also that $0 \leqslant p \leqslant q \leqslant 1-p$ means that $p(1-p) \leqslant q(1-q)$, so $\frac{p}{1-q} \leqslant \frac{q}{1-p}$ and $\frac{1-p}{q} \leqslant \frac{1-q}{p}$.}
    \label{tab:optimal-third-unitaries}
\end{table*}

\noindent We can compare \eqref{eq:W-is-I} and \eqref{eq:W-is-X} for any given value of $\frac{\alpha_0}{\alpha_1}$, with the six cases listed out in Table \ref{tab:optimal-third-unitaries}, using the restriction to $0 \leqslant p \leqslant q \leqslant 1-p$ to correctly order the critical points where the expression for the minimum changes. Thus, given any values of $x_1,x_2$, we calculate $\alpha_0(x_1,x_2)$ and $\alpha_1(x_1,x_2)$, and then refer to this table to tell us what unitary $W^{(x_1,x_2)}$ we should perform, and what the contribution to the total error will be.

\subsection{Choosing unitaries $V^{(x_1)}$ optimally}

\noindent Now we look at the choices for $V^{(x_1)}$, for $x_1 = 0$ and $x_1 = 1$ in turn. As with the unitaries $W^{(x_1,x_2)}$, we look at only the terms in the sum \eqref{eq:total-error-sum-of-minima-adaptive} that depend on $V^{(x_1)}$, which is precisely:
\begin{equation}
    \eta^{(x_1)} := \eta^{(x_1,0)} + \eta^{(x_1,1)}. \label{eq:V-error-contribution}
\end{equation}

\noindent For each fixed $x_1$, we will try each option $V^{(x_1)} = I$ and $V^{(x_1)} = X$. We can calculate $\alpha_i(x_1,0)$ and $\alpha_i(x_1,1)$ exactly for this choice of $V^{(x_1)}$, and thus can calculate $\eta^{(x_1,0)}$ and $\eta^{(x_1,1)}$, which we will sum to get $\eta^{(x_1)}$. Then, our optimal choice for $V^{(x_1)}$ will be the one that minimises $\eta^{(x_1)}$. We begin with $x_1 = 0$, and compare the two choices for $V^{(0)}$.

First, we try $V^{(0)} = I$, i.e. $\tau_0^{(0)} = \ketbra{0}{0}$ and $\tau_1^{(0)} = \ketbra{1}{1}$. In this case:
\begin{align}
    \alpha_0(0,0) &= (1-p)^2; & \alpha_1(0,0) &= q^2; \nonumber \\
    \alpha_0(0,1) &= p(1-p); & \alpha_1(0,1) &= q(1-q).
\end{align}

\noindent $\eta^{(0,0)}$ is dependent on the ratio $\frac{\alpha_0(0,0)}{\alpha_1(0,0)} = \frac{(1-p)^2}{q^2}$, which corresponds to either row 5 or row 6 of Table \ref{tab:optimal-third-unitaries}, depending on $p$ and $q$. Specifically, the ratio lies in the range of row 5 if and only if:
\begin{equation}
    \tfrac{(1-p)^2}{q^2} \leqslant \tfrac{1-q}{p} \quad \Leftrightarrow \quad p - p^2 + pq - q^2 \leqslant 0.
\end{equation}

\noindent Reading from the table, we obtain:
\begin{equation}
    \eta^{(0,0)} \big|_{V^{(0)} = I} = \begin{cases}
        p(1-p)^2 + q^3 & \text{if } p - p^2 + pq - q^2 \leqslant 0; \\
        q^2 & \text{if } p - p^2 + pq - q^2 \geqslant 0.
    \end{cases}
\end{equation}

\noindent Then, $\eta^{(0,1)}$ depends on the ratio $\frac{\alpha_0(0,1)}{\alpha_1(0,1)} = \frac{p(1-p)}{q(1-q)}$, corresponding to either row 2 or row 3. However, the optimal choice of $W^{(x_1,x_2)}$, and the optimal $\eta^{(x_1,x_2)}$, is the same between these rows, so in either case:
\begin{equation}
    \eta^{(0,1)} \big|_{V^{(0)} = I} = qp(1-p) + pq(1-q).
\end{equation}

\noindent Summing these contributions, as in \eqref{eq:V-error-contribution}, gives:
\begin{equation}
    \eta^{(0)} \big|_{V^{(0)} = I} = \begin{cases}
        p - 2p^2 + 2pq + p^3 - p^2 q - p q^2 + q^3 \\
        \hspace{4em} \text{if } p - p^2 + pq - q^2 \leqslant 0; \\
        2pq + q^2 - p^2 q - pq^2 \\
        \hspace{4em} \text{if } p - p^2 + pq - q^2 \geqslant 0.
    \end{cases} \label{eq:V-0-is-I}
\end{equation}

\noindent Alternatively, we could choose $V^{(0)} = X$, i.e. $\tau_0^{(0)} = \ketbra{1}{1}$ and $\tau_1^{(0)} = \ketbra{0}{0}$. This makes the probabilities $\alpha(x_1,x_2)$:
\begin{align}
    \alpha_0(0,0) &= (1-p)q; & \alpha_1(0,0) &= (1-p)q; \nonumber \\
    \alpha_0(0,1) &= (1-p)(1-q); & \alpha_1(0,1) &= pq.
\end{align}

\noindent For $\eta^{(0,0)}$, we find that $\frac{\alpha_0(0,0)}{\alpha_1(0,0)} = \tfrac{(1-p)q}{(1-p)q} = 1$ lies on the boundary between rows 3 and 4 of Table \ref{tab:optimal-third-unitaries}. This means that the conclusions of both rows are equally optimal (i.e. both choices $W^{(0,0)} = I$ or $W^{(0,1)} = X$). The expression for the error contribution $\eta^{(0,0)}$ is also necessarily the same from either row, equal to:
\begin{equation}
    \eta^{(0,0)} \big|_{V^{(0)} = X} = p(1-p)q + q(1-p)q.
\end{equation}

\noindent The other ratio is $\frac{\alpha_0(0,1)}{\alpha_1(0,1)} = \frac{(1-p)(1-q)}{pq}$, which lies in the range of row 6. $\eta^{(0,1)}$ is thus:
\begin{equation}
    \eta^{(0,1)} \big|_{V^{(0)} = X} = pq,
\end{equation}

\noindent and so the sum per \eqref{eq:V-error-contribution} is:
\begin{equation}
    \eta^{(0)} \big|_{V^{(0)} = X} = 2pq + q^2 - p^2 q - pq^2. \label{eq:V-0-is-X}
\end{equation}

\noindent Comparing \eqref{eq:V-0-is-I} and \eqref{eq:V-0-is-X}, $V^{(0)} = I$ always gives the same or less\footnote{Note that from our optimisation of $\eta^{(0)}$ in \eqref{eq:V-0-is-I}, we have $p - 2p^2 + 2pq + p^3 - p^2 q - p q^2 + q^3 \leqslant 2pq + q^2 - p^2 q - pq^2$ when $p - p^2 + pq - q^2 \leqslant 0$.} total error contribution than $V^{(0)} = X$, so to achieve optimality we take $V^{(0)}= I$ without loss of generality. Referring back to the rows of Table \ref{tab:optimal-third-unitaries} that we used above, we also obtain the optimal choices of $W^{(0,0)}$ and $W^{(0,1)}$. Thus, an optimal choice of the branch of the circuit with $x_1 = 0$ is:
\begin{align}
    V^{(0)} &= I; \nonumber \\
    W^{(0,0)} &= \begin{cases}
        I & \text{if } p - p^2 + pq - q^2 \leqslant 0; \\
        I \text{ or } X & \text{if } p - p^2 + pq - q^2 \geqslant 0;
    \end{cases} \nonumber \\
    W^{(0,1)} &= X. \label{eq:optimal-unitaries-0}
\end{align}

\noindent We now repeat the same analysis for $x_1 = 1$, albeit condensed for brevity. If $V^{(1)} = I$, i.e. $\tau_0^{(1)} = \ketbra{0}{0}$ and $\tau_1^{(1)} = \ketbra{1}{1}$, then $\frac{\alpha_0(1,0)}{\alpha_1(1,0)} = \frac{p(1-p)}{q(1-q)}$ corresponds to row 2 or 3, and $\frac{\alpha_0(1,1)}{\alpha_1(1,1)} = \frac{p^2}{(1-q)^2}$, corresponds to row 1, so:
\begin{align}
    \eta^{(1,0)} \big|_{V^{(1)} = I} &= qp(1-p) + pq(1-q); \nonumber \\
    \eta^{(1,1)} \big|_{V^{(1)} = I} &= p^2,
\end{align}
\begin{equation}
    \Rightarrow \quad \eta^{(1)} \big|_{V^{(1)} = I} = p^2 + 2pq - p^2 q - pq^2. \label{eq:V-1-is-I}
\end{equation}

\noindent If instead $V^{(1)} = X$, i.e. $\tau_0^{(1)} = \ketbra{1}{1}$ and $\tau_1^{(1)} = \ketbra{0}{0}$, then $\frac{\alpha_0(1,0)}{\alpha_1(1,0)} = \frac{pq}{(1-p)(1-q)}$ lies in the range of row 1, and $\frac{\alpha_0(1,1)}{\alpha_1(1,1)} = \frac{p(1-q)}{p(1-q)} = 1$ lies between rows 3 and 4, so:
\begin{align}
    \eta^{(1,0)} \big|_{V^{(1)} = X} &= pq; \nonumber \\
    \eta^{(1,1)} \big|_{V^{(1)} = X} &= pp(1-q) + qp(1-q);
\end{align}
\begin{equation}
    \Rightarrow \quad \eta^{(1)} \big|_{V^{(1)} = X} = p^2 + 2pq - p^2 q - pq^2. \label{eq:V-1-is-X}
\end{equation}

\noindent The error contributions from \eqref{eq:V-1-is-I} and \eqref{eq:V-1-is-X} are the same, so either choice of $V^{(1)}$ will be optimal. We choose $V^{(1)} = I$, such that $V^{(x_1)}$ is independent of $x_1$, and hence the second measurement is not performed adaptively, only the third. Then, as before, we refer back to the relevant rows of Table \ref{tab:optimal-third-unitaries} to find an optimal choice of $W^{(1,0)}$ and $W^{(1,1)}$:
\begin{align}
    V^{(1)} &= I; \nonumber \\
    W^{(1,0)} &= X; \nonumber \\
    W^{(1,1)} &= I \text{ or } X. \label{eq:optimal-unitaries-1}
\end{align}

\noindent As we have chosen $V^{(0)} = V^{(1)} = I$, the total error of our circuit will be given by the sum of \eqref{eq:V-0-is-I} and \eqref{eq:V-1-is-I}, and this will be the optimal adaptive total error $\etaA$:
\begin{align}
    \etaA &= \eta^{(0)} + \eta^{(1)} \nonumber \\
    &= \begin{cases}
        p - p^2 + 4pq + p^3 - 2p^2 q - 2pq^2 + q^3 \\
        \hspace{4em} \text{if } p - p^2 + pq - q^2 \leqslant 0; \\
        p^2 + 4pq + q^2 - 2p^2 q - 2pq^2 \\
        \hspace{4em} \text{if } p - p^2 + pq - q^2 \geqslant 0.
    \end{cases}
\end{align}

\section{No adaptive advantage in edge cases of imperfect $Z$} \label{sec:edge-cases}

\noindent In this appendix, we examine the edge cases of imperfect $Z$, namely the symmetric imperfect $Z$ when $p = q$, the trivial imperfect $Z$ when $p + q = 1$, and the semi-perfect $Z$ when $p = 0$. We show as claimed in Section \ref{sec:optimal-adaptive-noisy-Z} that there is never any adaptive advantage to be gained when using these measurements.

First, if the measurement is symmetric ($p = q$), we know that an optimal circuit will take the form of Figure \ref{fig:adaptive-noisy-Z-simplification}, albeit generalised to $N$ measurements, by following an equivalent argument as we used in the three measurement case. Then, any $X$ gate before a measurement can be replaced with a bit-flip after the measurement, since $Z_1 = X Z_0 X$. Thus we may assume that all unitaries are $I$, which can be implemented non-adaptively as there is no choice made dependent on measurement outcomes, with all decision making turned into classical post-processing. Therefore our optimal adaptive circuit is also non-adaptive, so there is no adaptive advantage.

Next, if the measurement is trivial ($p + q = 1$), then all measurement results are independent of the measured state. Thus:
\begin{align}
    &\varepsilon_1 = \mathbb{P}(0 | \ket{i} = \ket{1}) = \mathbb{P}(0) \nonumber \\
    &\hspace{2em} = 1 - \mathbb{P}(1) = 1 - \mathbb{P}(1 | \ket{i} = \ket{0}) = 1 - \varepsilon_0
\end{align}

\noindent so a total error of $\eta = 1$ is unavoidable in either the non-adaptive or adaptive case.

Finally, if the measurement is semi-perfect ($p = 0$ and $q \neq 0, 1$), the simplest way to see why there is no adaptive advantage is to rely on the tools of Section \ref{sec:qubit-measurements}. Supposing we have only one use of $\mathbb{Z}_{0,q}$, we may argue as in the three measurement case that an optimal choice of pre-processing is $\ketbra{0}{0} \mapsto \rho_0 = \ketbra{0}{0}$ and $\ketbra{1}{1} \mapsto \rho_1 = \ketbra{1}{1}$. Maximum likelihood decoding is then given by $f(0) = 0$ and $f(1) = 1$, so we have $\varepsilon_0 = 0$ and $\varepsilon_1 = q$ optimally. This exactly meets the conditions for Proposition \ref{prop:no-adaptive-advantage}, which tells us if an optimal single measurement circuit can be constructed with $\varepsilon_0 = 0$, then we cannot get an adaptive advantage for any $N$.

\section{Arbitrary adaptive advantage} \label{sec:arbitrary-adaptive-advantage}

\noindent In this appendix, we use the circuits shown in Figure \ref{fig:arbitrary-adaptive-advantage}, along with $q = 2^{-N}$ and $p = q^N$, to show that the ratio $R$ of minimal non-adaptive total error $\etaNA$ to minimal adaptive total error $\etaA$ can be made arbitrarily large.

First, we lower bound the non-adaptive optimum. By the same argument as we used in the three measurement case in the main text, an optimal non-adaptive circuit will take the form shown in Figure \ref{fig:arbitrary-adaptive-advantage}(a), for some $r \leqslant \frac{1}{2} N$. For any given value of $r$ let $\eta_r$ be the total error from the circuit with $r$ $X$ gates:
\begin{align}
    \eta_r &= \sum_{\mathbf{x} \in \{0,1\}^N} \min \Big\{ \! \bra{\mathbf{0}^{N-r} \mathbf{1}^r} Z_\mathbf{x} \ket{\mathbf{0}^{N-r} \mathbf{1}^r}, \nonumber \\
    &\hspace{17ex} \bra{\mathbf{1}^{N-r} \mathbf{0}^r} Z_\mathbf{x} \ket{\mathbf{1}^{N-r} \mathbf{0}^r} \Big\}, \label{eq:eta-r}
\end{align}

\noindent such that $\etaNA = \min_{0 \leqslant r \leqslant {\frac{1}{2} N}} \{ \eta_r \}$.

We lower bound $\eta_r$ by considering only some of the terms that contribute to the sum \eqref{eq:eta-r}. In particular, if $\mathbf{x} = \mathbf{0}$:
\begin{align}
    \bra{\mathbf{0}^{N-r} \mathbf{1}^r} Z_\mathbf{0} \ket{\mathbf{0}^{N-r} \mathbf{1}^r} &= (1-p)^{N-r} q^r; \nonumber \\
    \bra{\mathbf{1}^{N-r} \mathbf{0}^r} Z_\mathbf{0} \ket{\mathbf{1}^{N-r} \mathbf{0}^r} &= (1-p)^r q^{N-r}.
\end{align}

\noindent Maximum likelihood decoding means that the smaller of these contributes to the total error. As $1 - p > q$ and $r \leqslant \frac{1}{2} N$, the second of the two is smaller, so:
\begin{equation}
    \eta_r \geqslant (1-p)^r q^{N-r}.
\end{equation}

\noindent For $r \geqslant 1$, we can lower bound this by:
\begin{align}
    \eta_r &\geqslant (1-p)^r q^{N-r} \nonumber \\
    &\geqslant \left( 1 - rp \right) q^{N-r} \nonumber \\
    &\geqslant \left( 1 - \tfrac{1}{2} N 2^{-N^2} \right) q^{N-1} \quad \text{as } r \geqslant 1 \nonumber \\
    &\geqslant \tfrac{1}{2} q^{N-1} \quad \text{as } 1 - \tfrac{1}{2} N 2^{-N^2} \geqslant \tfrac{1}{2},
\end{align}

\noindent independent of $r$.

However, for $r = 0$, we require a slightly tighter bound, so we consider an additional set of probabilities that contribute to the total error. If $\| \mathbf{x} \|_1 = 1$, i.e. $\mathbf{x}$ contains exactly one 1, then:
\begin{align}
    \bra{\mathbf{0}^N} Z_\mathbf{x} \ket{\mathbf{0}^N} &= p (1-p)^{N-1}; \nonumber \\
    \bra{\mathbf{1}^N} Z_\mathbf{x} \ket{\mathbf{1}^N} &= q^{N-1} (1-q).
\end{align}

\noindent The smaller of these two is the first, since:
\begin{align}
    & p (1-p)^{N-1} - q^{N-1} (1-q) \nonumber \\
    &= q^N (1 - q^N)^{N-1} - q^{N-1} (1-q) \nonumber \\
    &= q^{N-1} (1-q) \left( q (1 - q^N)^{N-2} (1 + q + \dots + q^{N-1}) - 1 \right) \nonumber \\
    &\leqslant q^{N-1} (1-q) (2^{-N} N - 1) < 0,
\end{align}

\noindent and there are $N$ such choices for $\mathbf{x}$, so overall (including the $q^N$ term from $\mathbf{x} = \mathbf{0}$):
\begin{align}
    \eta_0 &\geqslant q^N + Np(1-p)^{N-1} \nonumber \\
    &\geqslant q^N + N q^N (1 - (N-1)q^N) \nonumber \\
    &= q^N \left( N + 1 - N(N-1) 2^{-N^2} \right) \nonumber \\
    &\geqslant N q^N \quad \text{as } N(N-1) 2^{-N^2} \leqslant 1.
\end{align}

\noindent Therefore, the optimal non-adaptive total error is bounded below by:
\begin{equation}
    \etaNA \geqslant \min \{ \tfrac{1}{2} q^{N-1}, N q^N \} = Nq^N,
\end{equation}

\noindent where we use that $N < \tfrac{1}{2} 2^N$ for $N \geqslant 3$.

\begin{table*}[t]
    \centering
    \begin{tabular}{c||c|c|c}
        $\mathbf{x}$ & $\mathbb{P}(\mathbf{x} | \ket{i} = \ket{0})$ & $\mathbb{P}(\mathbf{x} | \ket{i} = \ket{1})$ & Contribution to total error \\ \hline \hline
        $\mathbf{0}^N$ & $(1-p)^N$ & $q^N$ & $q^N$ \\ \hline
        $\mathbf{0}^{N-1} 1$ & $p(1-p)^{N-1}$ & $q^{N-1} (1-q)$ & $p(1-p)^{N-1} \leqslant p$ \\ \hline
        $\mathbf{0}^{r} 1 \mathbf{0}^{N-r-1}$, $0 \leqslant r \leqslant N-2$ & $p (1-p)^r q^{N-r-1}$ & $(1-p)^{N-r-1} q^r (1-q)$ & $p (1-p)^r q^{N-r-1} \leqslant pq$ \\ \hline
        $\| \mathbf{x} \|_1 = 2$, consecutive ones & $p (1-p)^{N-2} (1-q)$ & $p q^{N-2} (1-q)$ & $p q^{N-2} (1-q) \leqslant pq$ \\ \hline
        $\| \mathbf{x} \|_1 = 2$, non-consecutive ones & $\leqslant pq$ & $-$ & $\leqslant pq$ \\ \hline
        $\| \mathbf{x} \|_1 \geqslant 3$ & $\leqslant p^2$ & $-$ & $\leqslant pq$
    \end{tabular}
    \caption{The probabilities of each possible string of outcomes $\mathbf{x}$ from the adaptive circuit in Figure \ref{fig:arbitrary-adaptive-advantage}(b). The contribution to the total error in each row is the minimum of the two probabilities in the second and third columns, due to maximum likelihood decoding. Recall that $p = q^N$. In the fourth row, we use that $N \geqslant 3$. In the fifth and sixth rows, we don't calculate $\mathbb{P}(\mathbf{x} | \ket{i} = \ket{1})$ as the error contribution is upper bounded by $\mathbb{P}(\mathbf{x} | \ket{i} = \ket{0})$, which is small enough already.}
    \label{tab:adaptive-errors-noisy-Z-N}
\end{table*}

Now, we upper bound the adaptive optimum by using the circuit shown in Figure \ref{fig:arbitrary-adaptive-advantage}(b). This time we need to consider all of the probabilities for every string of outcomes $\mathbf{x}$, and upper bound the total error contribution of each. Then, the sum of these contributions upper bounds the total error of this circuit, and so upper bounds $\eta^{(A)}$.

In Table \ref{tab:adaptive-errors-noisy-Z-N}, we collate upper bounds for the error contribution for each string of outcomes $\mathbf{x} \in \{0,1\}^N$. Hence:
\begin{equation}
    \etaA \leqslant q^N + p + (2^N - 2)pq \leqslant q^N (1 + 1 + 2^N q) = 3q^N.
\end{equation}

\noindent With our bounds on $\etaNA$ and $\etaA$, we obtain the desired result:
\begin{equation}
    R = \frac{\etaNA}{\etaA} \geqslant \frac{Nq^N}{3q^N} = \frac{N}{3}.
\end{equation}

\section{Proofs} \label{sec:proofs}

\noindent In this appendix, we list the proofs of the various results we have stated throughout Sections \ref{sec:qubit-measurements} and \ref{sec:qudit-measurements}.

\subsection{Lemma \ref{lem:channel-existence} - Construction of measure-and-prepare channel without measurements} \label{sec:proof-channel-existence}

\noindent \textit{Let $\rho_i \in \mathcal{D}(\mathcal{H})$ for each $i \in [d]$. Then, there exists a quantum channel $\mathcal{A}$ such that $\mathcal{A}(\ketbra{i}{i}) = \rho_i$ for all $i$.}

\begin{proof}
    For each $i$, choose $\ket{\Psi_i} \in \mathcal{H} \otimes \mathcal{H}$ to be a purification of $\rho_i$. Then, define a completely positive map $\Lambda$ by Kraus operators $K_i = \ketbra{\Psi_i}{i}$. We have:
    \begin{equation}
        \sum_{i \in [d]} K_i^\dagger K_i = \sum_{i \in [d]} \left| i \rangle \hspace{-0.5ex} \langle \Psi_i | \Psi_i \rangle \hspace{-0.5ex} \langle i \right| = \sum_{i \in [d]} \ketbra{i}{i} = \mathbb{I}_d,
    \end{equation}
    
    \noindent so $\Lambda$ is trace-preserving, and hence a quantum channel. By composition, so is $\mathrm{tr}_2 \circ \Lambda$, with:
    \begin{equation}
       (\mathrm{tr}_2 \circ \Lambda) (\ketbra{i}{i}) = \mathrm{tr}_2(\ketbra{\Psi_i}{\Psi_i}) = \rho_i,
    \end{equation}
    
    \noindent and $\mathcal{A} = \mathrm{tr}_2 \circ \Lambda$ is by construction the channel we require.
\end{proof}

\subsection{Proposition \ref{prop:achievable-set-is-achievable} - The achievable error set is the set of achievable errors} \label{sec:proof-achievable-set-is-achievable}

\noindent \textit{A pair of errors $(\varepsilon_0, \varepsilon_1)$ can be achieved using an adaptive circuit with $N$ measurements of $\mathbb{M}$ if and only if $(\varepsilon_0, \varepsilon_1) \in \bar{E}(\mathbb{M}, N)$, where:
\begin{equation}
    \bar{E}(\mathbb{M}, N) = \mathrm{conv}(E(\mathbb{M}, N)).
\end{equation}}

\begin{proof}
    We proceed by induction. To begin with, suppose that $N = 0$. With no measurements available, the input state is irrelevant, and all we can do is choose an output of either $j = 0$ or $j = 1$. A mixed protocol combining these options outputs $j = 0$ with probability $t$ and $j = 1$ with probability $1-t$ for some $t \in [0,1]$. Therefore the set of possible error rates is exactly:
    \begin{equation}
        \{ (1-t, t) : t \in [0,1] \} = \bar{E}(\mathbb{M}, 0).
    \end{equation}

    \noindent Now suppose that $N \geqslant 1$, and let $(\varepsilon_0, \varepsilon_1)$ be error rates achieved using $N$ uses of the measurement $\mathbb{M}$. Mixed protocols result in errors that are convex sums of errors from single-circuit protocols, so by convexity of the target set $\bar{E}(\mathbb{M}, N)$, we may assume we have a single-circuit protocol, of the form shown in Figure \ref{fig:adaptive}. We rearrange the circuit to the form of Figure \ref{fig:adaptive-recursive}(c) with $k = N$, such that by \eqref{eq:recursive-epsilons}, we can express $\varepsilon_0, \varepsilon_1$ as:
    \begin{equation}
        \varepsilon_i = \sum_{x \in [m]} \mathrm{tr} \left( \mathcal{B}(\ketbra{i}{i}) M_x \right) \, \varepsilon_i^{(x)}, \label{eq:recursive-errors}
    \end{equation}

    \noindent where $\left( \varepsilon_0^{(x)}, \varepsilon_1^{(x)} \right)$ are the errors achieved by the subcircuit $\mathcal{C}_{N-1}^{(x)}$. Each such subcircuit uses $N-1$ measurements, so by inductive assumption, $\left( \varepsilon_0^{(x)}, \varepsilon_1^{(x)} \right) \in \bar{E}(\mathbb{M}, N-1)$. Thus we can express $\left( \varepsilon_0^{(x)}, \varepsilon_1^{(x)} \right)$ as a convex combination of points in $E(\mathbb{M}, N-1)$:
    \begin{equation}
        \left( \varepsilon_0^{(x)}, \varepsilon_1^{(x)} \right) = \sum_{r_x \in I_x} t_{r_x}^{(x)} \, \left( \varepsilon_{0,r_x}^{(x)}, \varepsilon_{1,r_x}^{(x)} \right).
    \end{equation}
    
    \noindent Here, for each $x \in [m]$, $I_x$ is an index set where for all $r_x \in I_x$, we have $\left( \varepsilon_{0,r_x}^{(x)}, \varepsilon_{1,r_x}^{(x)} \right) \in E(\mathbb{M}, N-1)$ and $t_{r_x}^{(x)} \geqslant 0$, with $\sum_{r_x \in I_x} t_{r_x}^{(x)} = 1$.

    We now combine these sets into a Cartesian product representative of the convex combination for each $x \in [m]$. Let $I = \prod_{x \in [m]} I_x$, and for $\mathbf{r} = (r_x)_{x \in [m]} \in I$, let $t_\mathbf{r} = \prod_{x \in [m]} t_{r_x}^{(x)}$. This choice means that for any $x \in [m]$:
    \vspace{-1em}
    \begin{align}
        \sum_{\mathbf{r} \in I} t_\mathbf{r} \varepsilon_{i,r_x}^{(x)} &= \sum_{\mathbf{r} \in I} \left( \prod_{y \in [m]} t_{r_y}^{(y)} \right) \sum_{s_x \in I_x} \delta_{r_x s_x} \varepsilon_{i,s_x}^{(x)} \nonumber \\
        &= \sum_{s_x \in I_x} \varepsilon_{i,s_x}^{(x)} \left( \sum_{r_x \in I_x} \delta_{r_x s_x} t_{r_x}^{(x)} \right) \left( \prod_{y \neq x} \sum_{r_y \in I_y} t_{r_y}^{(y)} \right) \nonumber \\
        &= \sum_{s_x \in I_x} \varepsilon_{i,s_x}^{(x)} \left( t_{s_x}^{(x)} \right) \left( \prod_{y \neq x} 1 \right) \nonumber \\
        &= \sum_{s_x \in I_x} t_{s_x}^{(x)} \varepsilon_{i,s_x}^{(x)} \nonumber \\
        &= \varepsilon_i^{(x)}. \label{eq:convex-errors}
    \end{align}
    
    \noindent Additionally, let $q_i^{(x)} = \mathrm{tr} \left( \mathcal{B}(\ketbra{i}{i}) M_x \right)$, so that $\mathbf{q}_i = (q_i^{(x)})_{x \in [m]} \in \mathcal{R}(\mathbb{M})$. Then, by combining \eqref{eq:recursive-errors} and \eqref{eq:convex-errors}:
    \begin{equation}
        (\varepsilon_0, \varepsilon_1) = \sum_{\mathbf{r} \in I} t_\mathbf{r} \left( \sum_{x \in [m]} \varepsilon_{0,r_x}^{(x)} q_0^{(x)}, \sum_{x \in [m]} \varepsilon_{1,r_x}^{(x)} q_1^{(x)} \right),
    \end{equation}

    \noindent a convex sum of elements in $E(\mathbb{M}, N)$, so in $\bar{E}(\mathbb{M}, N)$.

    Finally, we justify the converse, that any $(\varepsilon_0, \varepsilon_1) \in \bar{E}(\mathbb{M}, N)$ is an achievable pair of errors. This is because any $(\varepsilon_0, \varepsilon_1) \in E(\mathbb{M}, N)$ is achievable by translating the recursive construction of the set into a recursive adaptive circuit like in Figure \ref{fig:adaptive-recursive}(c), while the convex hull is achievable using mixed protocols.
\end{proof}

\subsection{Proposition \ref{prop:adaptive-lower-bound} - Lower bound on $\etaA$} \label{sec:proof-adaptive-lower-bound}

\noindent \textit{For any $\mathbb{M}$ and $N \geqslant 0$:
\begin{equation}
    \etaA(\mathbb{M},N) \geqslant \eta(\mathbb{M})^N. \label{eq:adaptive-lower-bound-appendix}
\end{equation}}

\begin{proof}
    Let $(\varepsilon_0, \varepsilon_1) \in E(\mathbb{M}, N)$. By definition of the single circuit achievable set \eqref{eq:achievable-errors}, we can find $\left( \varepsilon_0^{(x)}, \varepsilon_1^{(x)} \right) \in E(\mathbb{M}, N-1)$ for all $x \in [m]$, and $\mathbf{q}_0, \mathbf{q}_1 \in \mathcal{R}(\mathbb{M})$ with:
    \begin{align}
        \varepsilon_0 + \varepsilon_1 &= \sum_{x \in [m]} \left( \varepsilon_0^{(x)} q_0^{(x)} + \varepsilon_1^{(x)} q_1^{(x)} \right) \nonumber \\
        &\geqslant \min_{x \in [m]} \left\{ \varepsilon_0^{(x)} + \varepsilon_1^{(x)} \right\} \nonumber \\
        &\hspace{3em} \times \left( \sum_{x \in [m]} \tfrac{\varepsilon_0^{(x)}}{\varepsilon_0^{(x)} + \varepsilon_1^{(x)}} q_0^{(x)} + \tfrac{\varepsilon_1^{(x)}}{\varepsilon_0^{(x)} + \varepsilon_1^{(x)}} q_1^{(x)} \right).
    \end{align}

    \noindent The first term in the product is a total error achievable with $N-1$ measurements, so is at least $\etaA(\mathbb{M}, N-1)$. Also, since $\left( \frac{\varepsilon_0^{(x)}}{\varepsilon_0^{(x)} + \varepsilon_1^{(x)}}, \frac{\varepsilon_1^{(x)}}{\varepsilon_0^{(x)} + \varepsilon_1^{(x)}} \right) \in \bar{E}(\mathbb{M},0)$ for all $x \in [m]$, we can deduce that:
    \begin{equation}
        \left( \sum_{x \in [m]} \tfrac{\varepsilon_0^{(x)}}{\varepsilon_0^{(x)} + \varepsilon_1^{(x)}} q_0^{(x)}, \! \sum_{x \in [m]} \tfrac{\varepsilon_1^{(x)}}{\varepsilon_0^{(x)} + \varepsilon_1^{(x)}} q_1^{(x)} \right) \in \bar{E}(\mathbb{M},1).
    \end{equation}
    
    \noindent Therefore, the second term in the product is at least $\etaA(\mathbb{M}, 1) = \eta(\mathbb{M})$, and:
    \begin{equation}
        \varepsilon_0 + \varepsilon_1 \geqslant \etaA(\mathbb{M}, N-1) \ \eta(\mathbb{M}).
    \end{equation}

    \noindent Minimising over the left hand side gives that:
    \begin{equation}
        \etaA(\mathbb{M}, N) \geqslant \etaA(\mathbb{M}, N-1) \ \eta(\mathbb{M}).
    \end{equation}
    
    \noindent Finally, we use that $\etaA(\mathbb{M},0) = 1$ and $\etaA(\mathbb{M},1) = \eta(\mathbb{M})$, so \eqref{eq:adaptive-lower-bound-appendix} follows inductively. 
\end{proof}

\subsection{Lemma \ref{lem:reproduction-of-noisy-Z} - Reproducing an imperfect $Z$ measurement} \label{sec:proof-reproduction-of-noisy-Z}

\noindent \textit{Let $\mathcal{C}_N$ be any circuit with error rates $(\varepsilon_0, \varepsilon_1)$. Then, the action of the dephasing channel followed by $\mathcal{C}_N$ is the same as the action of measuring with the imperfect $Z$ measurement $\mathbb{Z}_{\varepsilon_0, \varepsilon_1}$.}

\begin{proof}
    Let $\rho = \rho_{00} \ketbra{0}{0} + \rho_{01} \ketbra{0}{1} + \rho_{10} \ketbra{1}{0} + \rho_{11} \ketbra{1}{1}$, a general qubit state. We will show that the two circuits:
    \begin{itemize}
        \item Dephasing followed by applying $\mathcal{C}_N$;
        \item Measuring with $\mathbb{Z}_{\varepsilon_0, \varepsilon_1}$,
    \end{itemize}
    
    \noindent act identically on $\rho$, i.e. map $\rho$ to the same probability vector. Since $\rho$ is arbitrary, this will mean that the circuits are equivalent.

    In the first case, dephasing maps $\rho \mapsto \rho_{00} \ketbra{0}{0} + \rho_{11} \ketbra{1}{1}$. Then, $\mathcal{C}_N$ takes input $\ketbra{0}{0}$ to probability vector $(1-\varepsilon_0, \varepsilon_0)^T$, and $\ketbra{1}{1}$ to $(\varepsilon_1, 1-\varepsilon_1)^T$, so the overall effect is:
    \begin{equation}
        \rho \mapsto \rho_{00} \begin{pmatrix} 1 - \varepsilon_0 \\ \varepsilon_0 \end{pmatrix} + \rho_{11} \begin{pmatrix} \varepsilon_1 \\ 1 - \varepsilon_1 \end{pmatrix} = \begin{pmatrix} \rho_{00} (1-\varepsilon_0) + \rho_{11} \varepsilon_1 \\ \rho_{00} \varepsilon_0 + \rho_{11} (1-\varepsilon_1) \end{pmatrix}.
    \end{equation}

    \noindent In the second case, we just measure with $\mathbb{Z}_{\varepsilon_0, \varepsilon_1}$, with probabilities:
    \begin{equation}
        \rho \mapsto \begin{pmatrix} \mathrm{tr}(\rho Z_0) \\ \mathrm{tr}(\rho Z_1) \end{pmatrix} = \begin{pmatrix} \rho_{00} (1-\varepsilon_0) + \rho_{11} \varepsilon_1 \\
        \rho_{00} \varepsilon_0 + \rho_{11} (1-\varepsilon_1) \end{pmatrix},
    \end{equation}

    \noindent which is the same as the first circuit.
\end{proof}

\stepcounter{subsection}

\subsection{Proposition \ref{prop:no-adaptive-advantage} - Sufficient condition for there to be no adaptive advantage for any $N$} \label{sec:proof-no-adaptive-advantage}

\noindent \textit{If $(0, \eta(\mathbb{M})) \in \bar{E}(\mathbb{M}, 1)$, then for any $N \geqslant 0$:
\begin{equation}
     \etaA(\mathbb{M}, N) = \etaNA(\mathbb{M}, N) = \eta(\mathbb{M})^N,
\end{equation}}

\begin{proof}
    This result is trivially true for $N = 0$, so we assume that $N \geqslant 1$. Using Proposition \ref{prop:adaptive-lower-bound}, we already have that $\eta(\mathbb{M})^N \leqslant \etaA(\mathbb{M}, N) \leqslant \etaNA(\mathbb{M}, N)$, so it suffices to prove that $\etaNA(\mathbb{M}, N) \leqslant \eta(\mathbb{M})^N$.

    $(0, \eta(\mathbb{M})) \in \bar{E}(\mathbb{M}, 1)$, so by Corollary \ref{cor:achievable-sets-as-imperfect-Zs} we can use one use of $\mathbb{M}$ to reproduce one use of the semi-perfect $Z$ measurement $\mathbb{Z}_{0,q}$, where $q = \eta(\mathbb{M})$. Since there is only one $\mathbb{M}$ being used in this reproduction, we do not need to use adaptivity to do so. We process our input state $\ket{\psi} = \ket{i}$ using CNOT gates into $\ket{i}^{\otimes N}$, then measure each of the qubits with our reproduction of $\mathbb{Z}_{0,q}$, as shown in Figure \ref{fig:non-adaptive-optimal-noisy-Z}. The post-processing function $f$ outputs 1 if any of the measurement results are 1, and 0 if all are 0.

    When $\ket{\psi} = \ket{0}$, this circuit never makes an error, always outputting $j = 0$ correctly. When $\ket{\psi} = \ket{1}$, the only way for an error to occur is if every measurement resulted in a 0, which happens with probability $q^N$. Therefore, we have total error:
    \begin{equation}
        \eta = \varepsilon_0 + \varepsilon_1 = 0 + q^N = \eta(\mathbb{M})^N,
    \end{equation}

    \noindent and as we have achieved this with a non-adaptive circuit, we conclude that indeed $\etaNA(\mathbb{M}, N) \leqslant \eta(\mathbb{M})^N$.    
\end{proof}

\subsection{Proposition \ref{prop:achievable-error-polygon} - Achievable error set is a polygon} \label{sec:proof-achievable-error-polygon}

\noindent \textit{For any $\mathbb{M}$ and $N \geqslant 0$:
\begin{equation}
    \mathrm{conv}(e(\mathbb{M}, N) \cup g(\mathbb{M}, N)) = \bar{E}(\mathbb{M}, N), \label{eq:finite-convex-hull-appendix}
\end{equation}
where $g(\mathbb{M}, N) = \left\{ (1-\varepsilon_1, 1-\varepsilon_0) : (\varepsilon_0, \varepsilon_1) \in e(\mathbb{M}, N) \right\} $ is the reflection of $e(\mathbb{M}, N)$ in the line $\varepsilon_0 + \varepsilon_1 = 1$.}

\begin{proof}
    If $N = 0$, \eqref{eq:finite-convex-hull-appendix} is trivially true, as $e(\mathbb{M}, 0) \cup g(\mathbb{M}, 0) = \{ (0,1), (1,0) \} = E(\mathbb{M}, 0)$, so their convex hulls are also the same. Then, for $N \geqslant 1$, we use induction, noting that since $e(\mathbb{M}, N) \cup g(\mathbb{M}, N) \subseteq E(\mathbb{M}, N) \subseteq \bar{E}(\mathbb{M}, N)$, the inclusion of the left-hand side of \eqref{eq:finite-convex-hull-appendix} in the right is trivial. As such, we need only prove the inclusion of the right in the left, and indeed, since both are convex sets, it suffices to show that:
    \begin{equation}
        E(\mathbb{M}, N) \subseteq \mathrm{conv}(e(\mathbb{M}, N) \cup g(\mathbb{M}, N)).
    \end{equation}
    
    \noindent Thus, let $(\varepsilon_0, \varepsilon_1) \in E(\mathbb{M}, N)$, and by \eqref{eq:union-of-rectangles}, we can choose $\left\{ \left( \varepsilon_0^{(x)}, \varepsilon_1^{(x)} \right) \right\} \subseteq E(\mathbb{M}, N-1)$ such that:
    \begin{equation}
        (\varepsilon_0, \varepsilon_1) \in \mathcal{I} \left( \sum_{x \in [m]} \varepsilon_0^{(x)} M_x \right) \times \mathcal{I} \left( \sum_{x \in [m]} \varepsilon_1^{(x)} M_x \right).
    \end{equation}

    \noindent By inductive assumption, we have that $\left( \varepsilon_0^{(x)}, \varepsilon_1^{(x)} \right) \in E(\mathbb{M}, N-1) \subseteq \mathrm{conv}(e(\mathbb{M}, N-1) \cup g(\mathbb{M}, N-1))$ for all $x \in [m]$. As in Proposition \ref{prop:achievable-set-is-achievable}, $\left( \varepsilon_0^{(x)}, \varepsilon_1^{(x)} \right)$ can be expressed as a convex combination of points in $e(\mathbb{M}, N-1) \cup g(\mathbb{M}, N-1)$, and so following the same argument, we may assume that $\left( \varepsilon_0^{(x)}, \varepsilon_1^{(x)} \right) \in e(\mathbb{M}, N-1) \cup g(\mathbb{M}, N-1)$, with the rest of the convex hull following by linearity.

    For each $x \in [m]$, let $\left( \zeta_0^{(x)}, \zeta_1^{(x)} \right) \in e(\mathbb{M},N)$ such that $\left( \varepsilon_0^{(x)}, \varepsilon_1^{(x)} \right) \in \left\{ \left( \zeta_0^{(x)}, \zeta_1^{(x)} \right), \left( 1 - \zeta_1^{(x)}, 1 - \zeta_0^{(x)} \right) \right\}$, i.e. such that $\left( \varepsilon_0^{(x)}, \varepsilon_1^{(x)} \right)$ is either that point in $e(\mathbb{M},N)$, or the reflected point in $g(\mathbb{M},N)$. $\zeta_0^{(x)} + \zeta_1^{(x)} \leqslant 1$, so $\zeta_i^{(x)} \leqslant \varepsilon_i^{(x)} \leqslant 1 - \zeta_{1-i}^{(x)}$ for all $i,x$, and:
    \begin{equation}
        \sum_{x \in [m]} \zeta_i^{(x)} M_x \leqslant \sum_{x \in [m]} \varepsilon_i^{(x)} M_x \leqslant \mathbb{I} - \sum_{x \in [m]} \zeta_{1-i}^{(x)} M_x.
    \end{equation}

    \noindent Therefore, for each $i = 0,1$:
    \begin{equation}
        \lambda_\mathrm{min} \left( \sum_{x \in [m]} \zeta_i^{(x)} M_x \right) \leqslant \varepsilon_i \leqslant 1 - \lambda_\mathrm{min} \left( \sum_{x \in [m]} \zeta_{1-i}^{(x)} M_x \right).
    \end{equation}

    \noindent Letting $\lambda_i = \lambda_\mathrm{min} \left( \sum_{x \in [m]} \zeta_i^{(x)} M_x \right)$, then $(\lambda_0, \lambda_1) \in e(\mathbb{M},N)$ by \eqref{eq:finite-achievable-errors}, and $(1-\lambda_1, 1-\lambda_0) \in g(\mathbb{M},N)$, so:
    \begin{align}
        (\varepsilon_0, \varepsilon_1) &\in [\lambda_0, 1-\lambda_1] \times [\lambda_1, 1-\lambda_0] \nonumber \\
        &\subseteq \mathrm{conv} \left( \left\{ (0,1), (1,0), (\lambda_0, \lambda_1), (1-\lambda_1, 1-\lambda_0) \right\} \right) \nonumber \\
        &\subseteq \mathrm{conv}(e(\mathbb{M},N) \cup g(\mathbb{M},N)).
    \end{align}    
\end{proof}

\stepcounter{subsection}

\subsection{Lemma \ref{lem:binomial-deviation} - Large deviation bound on binomial distribution} \label{sec:proof-binomial-deviation}

\noindent \textit{
Let $p \leqslant \frac{1}{2}$. Then:
\begin{equation}
    \mathbb{P} \left( \mathrm{Bin}(n,p) \geqslant \tfrac{1}{2} n \right) \leqslant \left( 2 \sqrt{p(1-p)} \right)^n,
\end{equation}
where $\mathrm{Bin}(n,p)$ is the binomial distribution with $n$ trials and probability of success $p$.}

\begin{proof}
    We derive this from the result of Hoeffding \cite[Theorem 1]{hoeffding-1963} (see also \cite[Theorem 1]{arratia-1989}), which states that if $0 \leqslant p < a < 1$, then:
    \begin{equation}
        \mathbb{P} \left( \mathrm{Bin}(n,p) \geqslant an \right) \leqslant e^{-nD(a \| p)},
    \end{equation}

    \noindent where:
    \begin{equation}
        D(a \| p) = a \log \tfrac{a}{p} + (1-a) \log \tfrac{1-a}{1-p}.
    \end{equation}

    \noindent Assuming for now that $p < \frac{1}{2}$, we may take $a = \frac{1}{2}$, with:
    \begin{equation}
        D(\tfrac{1}{2} \| p) = - \log 2 - \tfrac{1}{2} \log p - \tfrac{1}{2} \log (1-p),
    \end{equation}

    \noindent so $e^{-D(\frac{1}{2} \| p)} = 2 \sqrt{p (1-p)}$, and:
    \begin{equation}
        \mathbb{P} \left( \mathrm{Bin}(n,p) \geqslant \tfrac{1}{2} n \right) \leqslant \left( 2 \sqrt{p(1-p)} \right)^n,
    \end{equation}

    \noindent as required.

    In the edge case that $p = \frac{1}{2}$, then $2 \sqrt{p(1-p)} = 1$, so the inequality is trivially true as probabilities are always upper bounded by 1.
\end{proof}

\subsection{Lemma \ref{lem:multivariate-deviation} - Large deviation bound on multinomial distribution} \label{sec:proof-multivariate-deviation}

\noindent \textit{Let $X$ be a random variable with distribution:
\begin{equation}
    X = \begin{cases} 0 & \text{with probability } 1-t; \\ x & \text{with probability } \frac{t}{m-1} \quad \text{ for all } x \in [m] \setminus \{0\}. \end{cases}
\end{equation}
Let $\mathbf{X}_N = (X_k)_{k = 1}^N$ be a vector of iid random variables with this same distribution. Then, for any $j \neq 0$:
\begin{itemize}
    \item[(a)] If $t \leqslant \frac{m-1}{m}$, (i.e. $0$ more likely than the rest) then:
    \begin{equation}
        \mathbb{P} \left( \#\{X_k = j\} \geqslant \#\{X_k = 0\} \right) \leqslant \left( \tfrac{(m-2)t}{m-1} + 2 \sqrt{\tfrac{t(1-t)}{m-1}} \right)^N.
    \end{equation}
    \item[(b)] If $t \geqslant \frac{m-1}{m}$, (i.e. $0$ less likely than the rest) then:
    \begin{equation}
        \mathbb{P} \left( \#\{X_k = j\} \leqslant \#\{X_k = 0\} \right) \leqslant \left( \tfrac{(m-2)t}{m-1} + 2 \sqrt{\tfrac{t(1-t)}{m-1}} \right)^N.
    \end{equation}
\end{itemize}}

\begin{proof}
    In each case, we will condition on $n \in \{0, 1,..., N\}$, the number of $X_k$ that are either $0$ or $j$. The probability of having $\#\{X_k = 0\} + \#\{X_k = j\} = n$ is:
    \begin{equation}
        \mathbb{P}(n) = \binom{N}{n} \left( \tfrac{(m-2)t}{m-1} \right)^{N-n} \left( 1 - \tfrac{(m-2)t}{m-1} \right)^n,
    \end{equation}

    \noindent where $1 - \frac{(m-2)t}{m-1} = (1 - t) + \frac{t}{m-1}$ is the probability of getting $0$ or $j$, occurring $n$ times, and $\frac{(m-2)t}{m-1}$ is the probability of getting anything else, occurring $N - n$ times, with $\binom{N}{n}$ possible combinations of these occurrences.
    
    Meanwhile, conditional on $\#\{X_k = 0\} + \#\{X_k = j\} = n$, $\#\{X_k = 0\}$ and $\#\{X_k = j\}$ have the binomial distributions:
    \begin{align}
        \#\{X_k = 0\} \, | \, n &\sim \mathrm{Bin} \left(n, \tfrac{(m-1)(1-t)}{(m-1) - (m-2)t} \right); \nonumber \\
        \#\{X_k = j\} \, | \, n &\sim \mathrm{Bin} \left(n, \tfrac{t}{(m-1) - (m-2)t} \right),
    \end{align}

    \noindent as $\frac{1-t}{\frac{t}{m-1} + 1-t} = \frac{(m-1)(1-t)}{(m-1)-(m-2)t}$, $\frac{\frac{t}{m-1}}{\frac{t}{m-1} + 1-t} = \frac{t}{(m-1)-(m-2)t}$.
    
    In case (a), conditional on $n$, we see that $\#\{X_k = j\} \geqslant \#\{X_k = 0\}$ if and only if $\#\{X_k = j\} \geqslant \frac{1}{2} n$, so overall:
    \begin{align}
        &\mathbb{P} \left( \#\{X_k = j\} \geqslant \#\{X_k = 0\} \right) \nonumber \\
        &= \sum_{n = 0}^N \mathbb{P}(n) \, \mathbb{P} \left( \#\{X_k = j\} \geqslant \tfrac{1}{2} n \mid n \right) \nonumber \\
        &= \sum_{n = 0}^N \binom{N}{n} \left( \tfrac{(m-2)t}{m-1} \right)^{N-n} \left( 1 - \tfrac{(m-2)t}{m-1} \right)^n \times \nonumber \\
        &\hspace{4em} \mathbb{P} \left( \mathrm{Bin} \left( n, \tfrac{t}{(m-1)-(m-2)t} \right) \geqslant \tfrac{1}{2} n \right),
    \end{align}

    \noindent As $t \leqslant \tfrac{m-1}{m}$ if and only if $\frac{t}{(m-1)-(m-2)t} \leqslant \frac{1}{2}$, we can then apply Lemma \ref{lem:binomial-deviation} with $p = \frac{t}{(m-1)-(m-2)t}$:
    \begin{align}
        &\mathbb{P} \left( \#\{X_k = j\} \geqslant \#\{X_k = 0\} \right) \nonumber \\
        &\leqslant \sum_{n = 0}^N \binom{N}{n} \left( \tfrac{(m-2)t}{m-1} \right)^{N-n} \left( \tfrac{(m-1) - (m-2)t}{m-1} \right)^n \times \nonumber \\
        &\hspace{4em} \left(2 \sqrt{ \tfrac{t}{(m-1)-(m-2)t} \tfrac{(m-1)(1-t)}{(m-1)-(m-2)t} } \right)^n \nonumber \\
        &= \sum_{n = 0}^N \binom{N}{n} \left( \tfrac{(m-2)t}{m-1} \right)^{N-n} \left(2 \sqrt{ \tfrac{t(1-t)}{m-1} } \right)^n \nonumber \\
        &= \left( \tfrac{(m-2)t}{m-1} + 2 \sqrt{\tfrac{t(1-t)}{m-1}} \right)^N.
    \end{align}

    \noindent Meanwhile, in case (b), with $t \geqslant \tfrac{m-1}{m}$, the flipped inequality means that we are now interested in finding when $\#\{X_k = 0\} \geqslant \tfrac{1}{2} n$:
    \begin{align}
        &\mathbb{P} \left( \#\{X_k = j\} \leqslant \#\{X_k = 0\} \right) \nonumber \\
        &= \sum_{n = 0}^N \mathbb{P}(n) \, \mathbb{P} \left( \#\{X_k = 0\} \geqslant \tfrac{1}{2} n \mid n \right) \nonumber \\
        &= \sum_{n = 0}^N \binom{N}{n} \left( \tfrac{(m-2)t}{m-1} \right)^{N-n} \left( 1 - \tfrac{(m-2)t}{m-1} \right)^n \times \nonumber \\
        &\hspace{4em} \mathbb{P} \left( \mathrm{Bin} \left( n, \tfrac{(m-1)(1-t)}{(m-1)-(m-2)t} \right) \geqslant \tfrac{1}{2} n \right),
    \end{align}

    \noindent with $\tfrac{(m-1)(1-t)}{(m-1)-(m-2)t} \leqslant \frac{1}{2}$. As before, we apply Lemma \ref{lem:binomial-deviation}:
    \begin{align}
        &\mathbb{P} \left( \#\{X_k = j\} \leqslant \#\{X_k = 0\} \right) \nonumber \\
        &\leqslant \sum_{n = 0}^N \binom{N}{n} \left( \tfrac{(m-2)t}{m-1} \right)^{N-n} \left( \tfrac{(m-1) - (m-2)t}{m-1} \right)^n \times \nonumber \\
        &\hspace{4em} \left(2 \sqrt{ \tfrac{(m-1)(1-t)}{(m-1)-(m-2)t} \tfrac{t}{(m-1)-(m-2)t} } \right)^n \nonumber \\
        &= \sum_{n = 0}^N \binom{N}{n} \left( \tfrac{(m-2)t}{m-1} \right)^{N-n} \left(2 \sqrt{ \tfrac{t(1-t)}{m-1} } \right)^n \nonumber \\
        &= \left( \tfrac{(m-2)t}{m-1} + 2 \sqrt{\tfrac{t(1-t)}{m-1}} \right)^N.
    \end{align}
\end{proof}

\subsection{Lemma \ref{lem:symmetric-reproduced-measurement} - Reproducing a generalised symmetric imperfect $Z$ measurement} \label{sec:proof-symmetric-reproduced-measurement}

\noindent \textit{For any POVM $\mathbb{M}$ and $N \geqslant 0$, define the matrix $\mathcal{E} = (\varepsilon_{j|i})_{j,i \in [d]}$ by:
\begin{equation}
    \varepsilon_{j|i} = \begin{cases} 1 - \tfrac{1}{d} \eta_d^{(A)} (\mathbb{M},N) & \text{if } j = i; \\ \tfrac{1}{d(d-1)} \eta_d^{(A)} (\mathbb{M},N) & \text{if } j \neq i. \end{cases} \label{eq:reproduced-gen-sym-matrix}
\end{equation}
Then, $\mathcal{E} \in \bar{E}_d(\mathbb{M}, N)$, and so $N$ adaptive uses of $\mathbb{M}$ can reproduce one use of the generalised symmetric imperfect $Z$ measurement $\mathbb{S}_{d,\frac{1}{d} \eta_d^{(A)} (\mathbb{M},N)}$.}

\begin{proof}
    Let $\mathcal{E}' \in \bar{E}_d(\mathbb{M}, N)$ be any error matrix with minimal total error, i.e.:
    \begin{equation}
        \eta_d^{(A)} (\mathbb{M},N) = d - \mathrm{tr}(\mathcal{E}').
    \end{equation}
    
    \noindent Also, let $\mathcal{C}'$ be a circuit with $N$ uses of $\mathbb{M}$ that gives error probabilities characterised by $\mathcal{E}'$. Let $f: [d] \to [d]$ be a permutation of $[d]$. We consider the circuit where $f$ is applied to the input (in the form of a unitary $U_f : \ket{i} \mapsto \ket{f(i)}$), then $\mathcal{C}'$ is run, and then $f$ is undone on the output. This gives stochastic matrix $\mathcal{E}^{(f)}$, with elements:
    \begin{equation}
        \varepsilon_{j|i}^{(f)} = \varepsilon_{f(j)|f(i)}'.
    \end{equation}

    \noindent Then, we consider the mixed protocol given by a uniformly random choice of $f$ from $\mathrm{Sym}([d])$, the set of all permutations of $[d]$. Its stochastic matrix is $\mathcal{E}$ where:
    \begin{equation}
        \varepsilon_{j|i} = \tfrac{1}{d!} \sum_{f \in \mathrm{Sym}([d])} \varepsilon_{f(j)|f(i)}'.
    \end{equation}

    \noindent If $i = j$, then $f(i) = f(j)$ for all $f$, so for every $k \in [d]$, there are $(d-1)!$ functions $f \in \mathrm{Sym}([d])$ with $f(i) = f(j) = k$. Thus we have:
    \begin{equation}
        \varepsilon_{j|i} = \tfrac{(d-1)!}{d!} \sum_{k \in [d]} \varepsilon'_{k|k} = \tfrac{1}{d} \mathrm{tr}(\mathcal{E}') = 1 - \tfrac{1}{d} \eta_d^{(A)} (\mathbb{M},N),
    \end{equation}

    \noindent as in \eqref{eq:reproduced-gen-sym-matrix}.
    
    Meanwhile, if $i \neq j$, then $f(i) \neq f(j)$, since $f$ is a permutation. Therefore, for every $k, l \in [d]$ with $k \neq l$, there are $(d-2)!$ functions $f \in \mathrm{Sym}([d])$ with $f(i) = k$ and $f(j) = l$. The matrix element in this case is hence:
    \begin{equation}
        \varepsilon_{j|i} = \tfrac{(d-2)!}{d!} \sum_{k \neq l} \varepsilon_{l|k}' = \tfrac{1}{d(d-1)} \eta_d^{(A)} (\mathbb{M},N),
    \end{equation}

    \noindent also matching \eqref{eq:reproduced-gen-sym-matrix}, so proving that $\mathcal{E} \in \bar{E}(\mathbb{M}, N)$ via this mixed protocol.
    
    Then, using \eqref{eq:abelian-povm-from-matrix}, the abelian POVM corresponding to $\mathcal{E}$ is $\mathbb{M}_\mathcal{E}$ with elements:
    \begin{align}
        M_x &= \left( 1 - \tfrac{1}{d} \eta_d^{(A)} (\mathbb{M},N) \right) \ketbra{x}{x} \nonumber \\
        &\qquad + \tfrac{1}{d(d-1)} \eta_d^{(A)} (\mathbb{M},N) \sum_{y \neq x} \ketbra{y}{y} \nonumber \\
        &= S_x,
    \end{align}

    \noindent so $M_\mathcal{E} = \mathbb{S}_{d,\frac{1}{d} \eta_d^{(A)} (\mathbb{M},N)}$.    
\end{proof}

\subsection{Lemma \ref{lem:gen-sym-total-error} - Lower bound on $\eta_d$ for generalised symmetric imperfect $Z$} \label{sec:proof-gen-sym-total-error}

\noindent \textit{Let $m \geqslant d$ and $t \in [0,1]$. Then:
\begin{equation}
    \eta_d(\mathbb{S}_{m,t}) \leqslant (d-1) \tfrac{m}{m-1} t.
\end{equation}}

\begin{proof}
    Consider the following circuit, making use of one measurement of $\mathbb{S}_{m,t}$:
    \begin{itemize}
        \item Take input $\ket{i} \in \{\ket{0}, \ket{1},..., \ket{d-1}\}$;
        \item Apply channel mapping $\ket{i}$ in $d$ dimensions to $\ket{i}$ in $m$ dimensions;
        \item Measure with $\mathbb{S}_{m,t}$;
        \item Post-process the measurement result with any function $f: [m] \to [d]$ such that $f(x) = x$ for all $x \in [d] \subseteq [m]$.
    \end{itemize}

    \noindent The error probabilities of this circuit are:
    \begin{align}
        \varepsilon_{j|i} &= \sum_{x \in f^{-1}(j)} \mathrm{tr}(\ketbra{i}{i} S_x) \nonumber \\
        &= \sum_{x \in f^{-1}(j)} (1-t) \delta_{ix} + \tfrac{t}{m-1} (1 - \delta_{ix}),
    \end{align}

    \noindent using the definition of $\mathbb{S}_{m,t}$ in \eqref{eq:generalised-symmetric-imperfect-Z}.

    Therefore, its total error is:
    \begin{align}
        & d - \sum_{i \in [d]} \varepsilon_{i|i} \nonumber \\
        &= d - \sum_{i \in [d]} \sum_{x \in f^{-1}(i)} \left( (1-t) \delta_{ix} + \tfrac{t}{m-1} (1 - \delta_{ix}) \right) \nonumber \\
        &= d - \sum_{x \in [m]} \sum_{i \in [d]} \delta_{i f(x)} \left( (1-t) \delta_{ix} + \tfrac{t}{m-1} (1 - \delta_{ix}) \right) \nonumber \\
        &= d - \sum_{x \in [m]} \left( (1-t) \delta_{f(x)x} + \tfrac{t}{m-1} (1 - \delta_{f(x)x}) \right) \nonumber \\
        &= d - d(1-t) - (m-d) \tfrac{t}{m-1} \nonumber \\
        &= (d-1) \tfrac{m}{m-1} t,
    \end{align}

    \noindent where in the fifth line, we use that $f(x) = x$ for $x \in [d]$, and $f(x) \neq x$ for $x \in [m] \setminus [d]$.

    As this is an achievable total error for one use of $\mathbb{S}_{m,t}$, we may deduce that the minimal total error $\eta_d$ is equal or lesser than $(d-1) \frac{m}{m-1} t$ as claimed.
\end{proof}

\subsection{Proposition \ref{prop:qudit-total-error-lower-bound} - Lower bound on $\etaA_d$} \label{sec:proof-qudit-total-error-lower-bound}

\noindent \textit{Let $d' \leqslant d$. Then, for any POVM $\mathbb{M}$ and $N \geqslant 0$, we have:
\begin{equation}
    \tfrac{1}{d'-1} \eta_{d'}^{(\mathrm{A})}(\mathbb{M},N) \leqslant \tfrac{1}{d-1} \etaA_d(\mathbb{M},N).
\end{equation}}

\begin{proof}
    By Lemma \ref{lem:symmetric-reproduced-measurement}, one use of $\mathbb{S}_{d, \frac{1}{d} \etaA_d(\mathbb{M}, N)}$ can be reproduced by $N$ adaptive uses of $\mathbb{M}$. Then, we can apply Lemma \ref{lem:gen-sym-total-error}, substituting $d \mapsto d'$, $m \mapsto d$, and $t \mapsto \frac{1}{d} \etaA_d(\mathbb{M}, N)$ to obtain:
    \begin{align}
        \etaA_{d'} (\mathbb{M}, N) &\leqslant \eta_{d'} \left( \mathbb{S}_{d, \frac{1}{d} \etaA_d(\mathbb{M}, N)} \right) &\text{by Lemma \ref{lem:symmetric-reproduced-measurement}} \nonumber \\
        &\leqslant (d' - 1) \tfrac{d}{d-1} \tfrac{1}{d} \etaA_d (\mathbb{M}, N) &\text{by Lemma \ref{lem:gen-sym-total-error}} \nonumber \\
        &= \tfrac{d'-1}{d-1} \etaA_d (\mathbb{M}, N).
    \end{align}
\end{proof}

\stepcounter{subsection}

\subsection{Proposition \ref{prop:asymptotic-upper-bound} - Asymptotic upper bound on $\etaNA_d$ for an arbitrary POVM} \label{sec:proof-asymptotic-upper-bound}

\noindent \textit{Let $\mathbb{M}$ be any POVM. Then, for any $d \geqslant 2$:
\begin{equation}
    \etaNA_d(\mathbb{M}, N) = O \left( \left( \eta_2(\mathbb{M}) (2 - \eta_2(\mathbb{M})) \right)^{\frac{1}{4} N} \right).
\end{equation}}

\begin{proof}
    First, we note that if $\eta_2(\mathbb{M}) = 1$, then this bound is trivial, so we may assume that our measurement $\mathbb{M}$ is non-trivial, i.e. $\eta_2(\mathbb{M}) < 1$.
    
    By Corollary \ref{cor:achievable-sets-as-imperfect-Zs}, one use of $\mathbb{M}$ can reproduce one use of $\mathbb{Z}_{\frac{1}{2} \eta_2(\mathbb{M}), \frac{1}{2} \eta_2(\mathbb{M})}$. Therefore, any minimal total error for $\mathbb{M}$ is upper bounded by the equivalent quantity for $\mathbb{Z}_{\frac{1}{2} \eta_2(\mathbb{M}), \frac{1}{2} \eta_2(\mathbb{M})}$, specifically:
    \begin{align}
        \etaNA_d(\mathbb{M}, N) &\leqslant \etaNA_d \left( \mathbb{Z}_{\frac{1}{2} \eta_2(\mathbb{M}), \frac{1}{2} \eta_2(\mathbb{M})}, N \right) \nonumber \\
        &= O \left( \left( 2 \sqrt{\tfrac{1}{2} \eta_2(\mathbb{M}) \left( 1 - \tfrac{1}{2} \eta_2(\mathbb{M}) \right)} \right)^{\frac{1}{2} N} \right) \nonumber \\
        &= O \left( \left( \eta_2(\mathbb{M}) (2 - \eta_2(\mathbb{M})) \right)^{\frac{1}{4} N} \right),
    \end{align}
    
    \noindent by taking $p = \frac{1}{2} \eta_2(\mathbb{M}) < \tfrac{1}{2}$ in \eqref{eq:imperfect-Z-asymptotics}.
\end{proof}

\subsection{Proposition \ref{prop:generalised-asymptotic-upper-bound} - Tighter asymptotic upper bound on $\etaNA_d$ for an arbitrary POVM} \label{sec:proof-generalised-asymptotic-upper-bound}

\noindent \textit{Let $\mathbb{M}$ be any POVM. Let $d, \ell \geqslant 2$, with integers $q \geqslant 0$ and $r \in [\ell]$ defined by integer division such that $d = \ell q + r$. Then:
\begin{equation}
    \etaNA_d(\mathbb{M}, N) = O \left( \kappa^N \right),
\end{equation}
where:
\begin{equation}
    \kappa = \left( \tfrac{(\ell-2) \eta_\ell(\mathbb{M})}{\ell(\ell-1)} + 2 \sqrt{\tfrac{\eta_\ell(\mathbb{M}) (\ell - \eta_\ell(\mathbb{M}))}{\ell^2 (\ell-1)}} \right)^{ 1 - \frac{q(d-\ell+r)}{d(d-1)}}. \label{eq:upper-bound-rate-appendix}
\end{equation}}

\begin{proof}
    We note first that we have replaced $m$ in \eqref{eq:gen-sym-eta-upper-bound} by $\ell$, to avoid confusion with the number of outcomes of the measurement $\mathbb{M}$.

    The proof of this follows exactly the same route as that of Proposition \ref{prop:asymptotic-upper-bound}. We begin as before by seeing that trivial POVMs, i.e. those with $\eta_\ell(\mathbb{M}) = \ell-1$, satisfy this bound trivially, as \eqref{eq:upper-bound-rate-appendix} becomes:
    \begin{align}
        \kappa &= \left( \tfrac{(\ell-2)(\ell-1)}{\ell(\ell-1)} + 2 \sqrt{\tfrac{(\ell-1) (\ell - (\ell-1))}{\ell^2 (\ell-1)}} \right)^{ 1 - \frac{q(d-\ell+r)}{d(d-1)}} \nonumber \\
        &= \left( \tfrac{\ell-2}{\ell} + 2 \sqrt{\tfrac{1}{\ell^2}} \right)^{ 1 - \frac{q(d-\ell+r)}{d(d-1)}} \nonumber \\
        &= \left( \tfrac{\ell-2}{\ell} + \tfrac{2}{\ell} \right)^{ 1 - \frac{q(d-\ell+r)}{d(d-1)}} \nonumber \\
        &= 1.
    \end{align}
    
    \noindent Hence we may assume our measurement $\mathbb{M}$ is non-trivial. Instead of Corollary \ref{cor:achievable-sets-as-imperfect-Zs}, we use Lemma \ref{lem:symmetric-reproduced-measurement} to see that one use of $\mathbb{M}$ can reproduce one use of $\mathbb{S}_{\ell, \frac{1}{\ell} \eta_\ell(\mathbb{M})}$, so using \eqref{eq:gen-sym-eta-upper-bound} with $m \mapsto \ell$, $t = \frac{1}{\ell} \eta_\ell(\mathbb{M}) < \tfrac{\ell-1}{\ell}$:
    \begin{align}
        &\etaNA_d(\mathbb{M}, N) \nonumber \\
        &\leqslant \etaNA_d \left( \mathbb{S}_{\ell, \frac{1}{\ell} \eta_\ell(\mathbb{M})}, N \right) \nonumber \\
        &\leqslant d(d-1) \times \nonumber \\
        &\hspace{1ex} \left( \tfrac{(\ell-2) \frac{1}{\ell} \eta_\ell(\mathbb{M}) }{\ell-1} + 2 \sqrt{\tfrac{\frac{1}{\ell} \eta_\ell(\mathbb{M}) \left( 1 - \frac{1}{\ell} \eta_\ell(\mathbb{M}) \right)}{\ell-1}} \right)^{\left( 1 - \frac{q(d-\ell+r)}{d(d-1)} \right) N} \nonumber \\
        &= d(d-1) \times \nonumber \\
        &\hspace{1ex} \left( \tfrac{(\ell-2) \eta_\ell(\mathbb{M}) }{\ell(\ell-1)} + 2 \sqrt{\tfrac{ \eta_\ell(\mathbb{M}) \left( \ell - \eta_\ell(\mathbb{M}) \right)}{\ell^2(\ell-1)}} \right)^{\left( 1 - \frac{q(d-\ell+r)}{d(d-1)} \right) N} \nonumber \\
        &= d(d-1) \kappa^N.
    \end{align}
\end{proof}

\end{document}